\documentclass[a4paper,onecolumn,11pt,accepted=2022-12-05]{quantumarticle}
\pdfoutput=1

\usepackage{subfig}%
\usepackage{amsfonts,amsmath,amssymb,amsthm}
\usepackage[boxruled, vlined]{algorithm2e}
\SetKwInput{KwInput}{Input}                %
\SetKwInput{KwOutput}{Output} 
\DontPrintSemicolon
\usepackage[breaklinks=true, colorlinks=true, linkcolor={blue!90!black}, citecolor={blue!90!black}, urlcolor={blue!90!black}]{hyperref}
\usepackage{lscape}
\usepackage{graphicx}
\usepackage{tikz}
\usepackage{caption}
\usepackage{float}
\usetikzlibrary{shapes,arrows, positioning}
\usepackage[noend]{algorithmic}
\usepackage{adjustbox}
\usepackage[figuresright]{rotating}
\usetikzlibrary{trees}
\usepackage{pgfplots}
\usepackage[square, comma, sort&compress, numbers]{natbib}
\usepackage[flushleft]{threeparttable}
\usepackage{mathtools}
\usepackage{url}
\usetikzlibrary{shapes,positioning,fit,calc}

\pgfplotsset{compat=1.13}

\usepackage{colortbl}
\usepackage{pdflscape}
\usepackage{xcolor}
\usepackage{cite}               %
\usepackage{bbold}
\usepackage{multicol}
\usepackage{multirow}

\newcolumntype{R}{>{\columncolor{red!20}}c}
\newcolumntype{G}{>{\columncolor{green!20}}c}
\newcolumntype{B}{>{\columncolor{blue!20}}c}
\newcolumntype{Y}{>{\columncolor{yellow!20}}c}
\newcolumntype{K}{>{\columncolor{black!20}}c}

\usepackage{listings}
\usepackage[title]{appendix}

\newtheorem{theorem}{Theorem}

\newtheorem{proposition}{Proposition}
\newtheorem{corollary}{Corollary}

\newtheorem{remark}{Remark}

\graphicspath{{../figures/}}

\lstdefinestyle{C++}{
	language=C++,
	keywordstyle=\bfseries\color{purple},
	stringstyle=\color{blue},
	commentstyle=\color{gray},
	comment=[l]{/*}
}

\lstdefinestyle{AMPL}{
	language=AMPL,
	aboveskip=3mm,
	belowskip=3mm,
	showstringspaces=false,
	columns=flexible,
	keywordstyle=\bfseries,
	breaklines=true,
	breakatwhitespace=true,
	tabsize=3,
}

\lstdefinelanguage{AMPL}{keywords={let,set,param,var,arc,integer,minimize,maximize,subject,to,node,
sum,in,Current,complements,integer,solve_result_num,IN,contains,less,suffix,INOUT,default,logical,
Infinity,dimen,max,symbolic,Initial,div,min,table,LOCAL,else,option,then,OUT,environ,setof,union,
all,exists,shell_exitcodeuntil,binary,forall,solve_exitcodewhile,by,if,solve_messagewithin,check,
solve_result},sensitive=true,comment=[l]{\#}}

\newcommand{\be}{\begin{equation}}
\newcommand{\ee}{\end{equation}}
\newcommand{\bea}{\begin{eqnarray}}
\newcommand{\eea}{\end{eqnarray}}

\newcommand{\bvec}{\left(\begin{array}{c}}
\newcommand{\evec}{\end{array}\right)}
\newcommand{\bsub}{\begin{subequations}}
\newcommand{\esub}{\end{subequations}}

\def\argmax{\mathop{\rm arg\,max}}%
\def\argmin{\mathop{\rm arg\,min}}%

\usepackage{longtable}

\SetKwFunction{FRecurs}{FnRecursive}%
\SetKwProg{Function}{}{}{}\SetKwFunction{FRecurs}{void FnRecursive}%
\SetKwComment{Comment}{/* }{ */}

\usepackage{environ}

\definecolor{peach}{HTML}{FFCCAC}
\definecolor{butter}{HTML}{FFEB94}
\definecolor{babyblue}{HTML}{C1E1DC}

\NewEnviron{globalization_mechanism}{%
	\fcolorbox{peach}{peach}{%
	\begin{minipage}[l]{0.88\columnwidth}
	\BODY
	\end{minipage}
	}
}

\NewEnviron{globalization_strategy}{%
	\fcolorbox{butter}{butter}{%
	\begin{minipage}[l]{0.9\columnwidth}
	\BODY
	\end{minipage}
	}
}
\NewEnviron{step_computation}{%
	\fcolorbox{babyblue}{babyblue}{%
	\begin{minipage}[l]{0.9\columnwidth}
	\BODY
	\end{minipage}
	}
}

\usepackage{physics}

\usepackage[normalem]{ulem} %

\begin{document}
\allowdisplaybreaks
\graphicspath{{images/}} 

\title{Binary Control Pulse Optimization for Quantum Systems} 

\author{Xinyu Fei}
\affiliation{Department of Industrial and Operations Engineering, University of Michigan at Ann Arbor}
\email{xinyuf@umich.edu}
\orcid{0000-0001-7010-8664}

\author{Lucas T. Brady}
\affiliation{Joint Center for Quantum Information and Computer Science, NIST/University of Maryland}
\email{lucas.brady@nist.gov}
\orcid{0000-0001-7696-7689}

\author{Jeffrey Larson}
\affiliation{Mathematics and Computer Science Division, Argonne National Laboratory}
\email{jmlarson@anl.gov}
\orcid{0000-0001-9924-2082}

\author{Sven Leyffer}
\affiliation{Mathematics and Computer Science Division, Argonne National Laboratory}
\email{leyffer@anl.gov}
\orcid{0000-0001-8839-5876}

\author{Siqian Shen}
\affiliation{Department of Industrial and Operations Engineering, University of Michigan at Ann Arbor}
\email{siqian@umich.edu}
\orcid{0000-0002-2854-163X}

\maketitle

\begin{abstract}
Quantum control aims to manipulate quantum systems toward specific quantum
states or desired operations. Designing highly accurate and effective control
steps is  vitally important to various quantum applications, including energy
minimization and circuit compilation. In this paper we focus on discrete
binary quantum control problems and apply different optimization algorithms and
techniques to improve computational efficiency and solution quality.
Specifically, we develop a generic model and extend it in several ways. We introduce 
a squared $L_2$-penalty function to handle additional side constraints, 
to model requirements such as allowing at most one control to be active. We introduce a total variation (TV) regularizer to reduce the number of switches in the control. 
We modify the popular gradient ascent pulse engineering (GRAPE) algorithm, 
develop a new alternating direction method of multipliers (ADMM) algorithm 
to solve the continuous relaxation of the penalized model, and then apply rounding
techniques to obtain binary control solutions. We propose a modified
trust-region method to further improve the solutions. Our
algorithms can obtain high-quality control results, as demonstrated by numerical
studies on diverse quantum control
examples. \end{abstract}

\section{Introduction}
\label{sec: intro}

Quantum control~\citep{Rabitz2000,Werschnik2007,Brif2010} is a rich field that
started with applications in quantum
chemistry~\citep{Shi1988,Peirce1988,Shi1989,Kosloff1989,Jakubetz1990} but has
since expanded to other subfields such as atomic, molecular, and optical physics~\citep{khaneja2005optimal,Gorshkov2008,Bijnen2015}. Quantum control theory has been used in the quantum information community for a long time,
mostly for designing pulse controls and gates in quantum
devices~\citep{Palao2002,Palao2003,Montangero2007,Grace2007,Waldherr2014,Dolde2014,Venturelli2018,Omran2019,Khatri2019};
but more recently, researchers have considered using control theories and optimization for the high-level design of quantum algorithms~\citep{Yang2017,Bapat2019,Mbeng2019,Lin2019,brady2020optimal,brady2021annealing,Venuti2021}.
In this paper, we focus primarily on the quantum information applications of control theory, but our methods are generalizable to other settings. We do
not focus on the analytic aspects of quantum control but on the more mechanistic aspects of how to apply optimization techniques to the problem
of designing and finding solutions to a variety of quantum control problems.

Quantum computing architecture falls into different categories, such as quantum
circuits, which utilize discrete unitary gates to control a system, and quantum
annealing~\citep{Kadowaki1998,Farhi2000}, which uses smooth Hamiltonian
evolution instead. Quantum control is already familiar within quantum
circuit architecture design as a method for designing quantum gates, and
several of our examples described in Section~\ref{sec: model} belong to this
category. 
Within the framework of quantum algorithms, quantum control is often used in
variational quantum algorithms, which can straddle the line between discrete
gates and smooth evolution. This paper focuses mainly on discrete binary
control, which is more applicable to a quantum gate architecture. Variational
techniques have already been implemented experimentally (see, e.g.,~\citep{Pagano2020,Harrigan2021}), 
and our techniques can potentially enhance the
performance of the classical variational loop on such near-term devices.

Most quantum optimal control algorithms are based on gradient descent for
better convergence than using gradient-free algorithms~\citep{nocedal2006numerical}.  \citet{khaneja2005optimal} proposed the
gradient ascent pulse engineering (GRAPE) algorithm for designing pulse
sequences in nuclear magnetic resonance. 
The authors approximated the control function by a piecewise-constant function
and evaluated the explicit derivative. Based on the GRAPE
algorithm, \citet{Larocca_2021} proposed a
K-GRAPE algorithm by using a Krylov subspace to estimate quantum states during
the process of time evolution. \citet{brady2020optimal} applied an analytical
framework based on gradient descent to discuss the optimal control
procedure of a special control problem that minimizes the energy of a quantum
state with a combination of two Hamiltonians. 
Another class of quantum optimal
control algorithms is based on the chopped random basis (CRAB) optimal control
technique, which describes the control space by a series of basis functions and
optimizes the corresponding coefficients~\citep{doria2011optimal,
caneva2011chopped}. 
\citet{sorensen2018quantum} combined the GRAPE algorithm and
CRAB algorithm to achieve better results and faster convergence.
However, all these algorithms are designed for unconstrained continuous quantum control problems.

In this work, we focus on a quantum pulse optimization problem having binary
control variables and restricted feasible regions derived by linear
constraints. These constraints describe a so-called bang-bang control in a model
that corresponds to a quantum circuit design similar to that in the quantum
approximate optimization algorithm (QAOA)~\citep{farhi2014quantum} and other
variational quantum algorithms~\citep{Bharti2022,Cerezo2021}. Some literature formulates the QAOA algorithm into bang-bang control problems with a fixed evolution time and investigates the performance of multiple methods including Stochastic Descent and Pontryagin’s minimum principle for optimizing control models~\citep{liang2020investigating,bao2018optimal}. However, they only consider quantum systems with two controllers, while we propose a more general solution framework for quantum systems with multiple controllers.  
In this paper, we introduce four quantum control examples: (i) energy minimization
problem, (ii) controlled NOT (CNOT) gate estimation problem, (iii) NOT gate estimation problem with the energy leakage, and (iv) circuit compilation problem.
The nonconvexity, binary variables, and restricted feasible
sets in these examples lead to extreme challenges and difficulties in solving
the related binary quantum control problems.  

Methods for binary control mainly include genetic algorithms~\citep{muhlenbein1988evolution}, branch-and-bound~\citep{lawler1966branch, leyffer2001integrating},
and local search~\citep{vogt2021binary}. 
To the best of our knowledge, genetic algorithms have not
been applied to binary quantum control, but \citet{zahedinejad2014evolutionary}
showed that such algorithms fail to obtain high-quality solutions even to continuous quantum
control problems. Branch-and-bound can find
high-quality solutions for binary optimal control, but the long computational
time makes the algorithm intractable and hard to use for large-scale practical
problems. 
\citet{vogt2021binary} introduced a trust-region method~\citep{nocedal2006numerical} to solve the single flux quantum control problem
with binary controls. 
We extend their approach by allowing the addition of constraints, such as min-up-time and max-switching constraints to control the number of switches.

The main contributions of this paper are as follows. First, we propose a generic 
model including continuous and discretized versions for binary quantum control 
problems with linear constraints indicating that at each time step there can 
be only one active control. 
Second, we introduce an exact penalty function for linear
constraints and develop an algorithmic framework combining the GRAPE algorithm
and rounding techniques to solve it. Third, to prevent chattering on controls, we 
propose a new model that includes a total variation (TV) regularizer, and then we propose a new approach based on the alternating direction method of multipliers (ADMM) to solve this model. 
Fourth, we introduce a 
modified trust-region method to improve the solutions obtained
from these algorithms. Compared with other methods, our algorithmic
framework can obtain high-quality binary control sequences and prevent frequent switches. We demonstrate the performance on multiple quantum control examples with various objective functions and controllers.

The paper is organized as follows. In Section~\ref{sec:
model} we construct generic models of quantum binary control problems
using continuous and discretized formulations. 
We also introduce the corresponding specific formulations for four quantum
control examples. 
In Section~\ref{sec: alg-binary} we review the GRAPE algorithm, utilize a
penalty function to relax certain linear constraints, and introduce our rounding
techniques. In Section~\ref{sec: alg-infrequent} we propose and solve the
penalty model with the TV regularizer and the corresponding ADMM algorithm. In Section
\ref{sec: alg-tr} we derive trust-region subproblems and propose an approximate local-branching method based on the trust-region algorithm. In Section~\ref{sec: results} we present numerical simulations for multiple instances of the three quantum control
examples in Section~\ref{sec: model} to demonstrate the computational efficacy
and solution performance of our models and algorithms. In Section~\ref{sec:
conclusion} we summarize our work  and propose future research directions. 

\section{Formulations for Quantum Control}
\label{sec: model}
Consider a quantum system with $q$ qubits and a time horizon $[0,t_f]$, where
$t_f>0$ is defined as the evolution time. Let $H^{(0)}\in \mathbb{C}^{2^q\times
2^q}$ be the intrinsic Hamiltonian. Let $N$ be the number of control operators
and $H^{(j)}\in \mathbb{C}^{2^q\times 2^q},\ j=1,\ldots,N$ be the given control Hamiltonians. 
Let $X_{\textrm{init}},\ X_{\textrm{targ}}\in \mathbb{C}^{2^q\times 2^q}$ be the
initial and target unitary operators of the quantum system, respectively. 
Define variables $u_j(t)\in
\left\{0,1\right \},\ j=1,\ldots,N,\ \forall t$ as the real-valued control functions at time $t$, and use $u$ to represent the corresponding vector form.
Define variables $H(t)\in \mathbb{C}^{2^q\times 2^q}$ as the time-dependent
control Hamiltonian function and $X(t)\in \mathbb{C}^{2^q\times 2^q}$ as the
time-dependent unitary quantum operator function. A generic quantum control
problem for optimizing a function of quantum operators is formulated as 
\begin{subequations}
\label{eq:model-c-1}
\begin{align}
    \label{eq:model-c-1-obj}
    \min \quad & F(X(t_f))\\
    \label{eq:model-c-1-cons-h}
    \textrm{s.t.}\quad & H(t) = H^{(0)} + \sum_{j=1}^N u_j(t)H^{(j)},\ \forall t\in [0,t_f]\\
    \label{eq:model-c-1-cons-s}
    & \frac{d}{dt} X(t) = -iH(t) X(t),\ \forall t\in [0,t_f]\\
    \label{eq:model-c-1-cons-i}
    & X(0) = X_\textrm{init}\\
    \label{eq:model-c-1-cons-u-sum}
    & \sum_{j=1}^N u_j(t) = 1,\ \forall t\in [0,t_f]\\
    \label{eq:model-c-1-cons-u}
    & u_j(t)\in \left\{0,1\right\},\ j=1,\ldots,N.
\end{align}
\end{subequations}
The objective function \eqref{eq:model-c-1-obj} is a general function with
respect to the final unitary operator $X(t_f)$. We assume that $F$ is a differentiable function. In Sections~\ref{sec: model-energy}--\ref{sec: model-compilation},  
we provide the specific formulations of the objective function
for four quantum control examples. Constraint \eqref{eq:model-c-1-cons-h}
describes the computation of the time-dependent Hamiltonian function based on
the intrinsic Hamiltonian, control Hamiltonians, and control functions.
Constraint \eqref{eq:model-c-1-cons-s} is the ordinary differential equation
describing the time evolution of the operator of the quantum system 
with the same units for energy and frequency, setting $\hbar=1$.
Constraint \eqref{eq:model-c-1-cons-i} is the initial condition of the
operator. Constraint \eqref{eq:model-c-1-cons-u-sum} indicates that the
summation of all the control functions should be one at all times, which is
described by the Special Ordered Set of Type 1 (SOS1) property in optimal
control theory~\citep{sager2012integer}. In binary control, this SOS1 property guarantees that only one control field will be active at a time that mimics the bang-bang nature of some quantum applications such as the quantum approximate optimization algorithm  and 
the variational quantum eigensolver (VQE). Constraints \eqref{eq:model-c-1-cons-u} require all time-based values of the control
functions to be feasible. 

Following the time discretization process in~\citep{khaneja2005optimal}, we
explicitly integrate constraint \eqref{eq:model-c-1-cons-s}.
In particular, we divide the time horizon $[0,t_f]$ into $T$
time intervals $[t_{k-1}, t_{k}),\ k=1,\ldots,T$. 
We consider time intervals with an equal length represented by
$\Delta t$, but our work readily extends to nonuniform discretizations. 
For each controller $j=1,\ldots,N$ and each time step
$k=1,\ldots,T$, we define discretized binary control variables as $u_{jk}\in
\left\{0,1\right\}$. For each time step $k=1,\ldots,T$, we define the discretized
time-dependent Hamiltonians as $H_{k}\in \mathbb{C}^{2^q\times 2^q}$ and 
the quantum operators as $X_k\in \mathbb{C}^{2^q\times 2^q}$. 
The differential equation 
\eqref{eq:model-c-1-cons-s} is thus approximated by 
\begin{align}
\label{eq:diff-linear-eq}
    \frac{d}{dt} X(t)=-iH_kX(t),\ \forall t\in [t_{k-1}, t_k),\ k=1,\ldots,T.
\end{align}
For each $k=1,\ldots,T$, the linear differential equation \eqref{eq:diff-linear-eq} is a 
Schr\"odinger equation of a unitary operator, and 
we obtain an explicit solution as $$X(t)=\exp\left\{-iH_k(t-t_{k-1})\right\}X(t_{k-1}),\ \forall t\in [t_{k-1},t_k),$$  
because $H_k$ is time independent.
We then obtain a discretized quantum control problem (DQCP) as a discretization of the differential control model \eqref{eq:model-c-1} using explicit solutions on each interval $[t_{k-1}, t_k)$ for unitary operators as 
\begin{subequations}
\makeatletter
\def\@currentlabel{DQCP}
\makeatother
\label{eq:model-d-1} 
\begin{align}
    \label{eq:model-d-1-obj}
    (DQCP)\quad \min \quad & F(X_T)\\
    \label{eq:model-d-1-cons-h}
    \textrm{s.t.}\quad & H_k = H^{(0)} + \sum_{j=1}^N u_{jk}H^{(j)},\ k=1,\ldots,T\\
    \label{eq:model-d-1-cons-s}
    & X_{k}=e^{-i H_k \Delta t}X_{k-1},\ k=1,\ldots,T \\
    \label{eq:model-d-1-cons-i}
    & X_0 = X_\textrm{init}\\
    \label{eq:model-d-1-cons-u-sum}
    & \sum_{j=1}^N u_{jk} = 1,\ k=1,\ldots,T\\
    \label{eq:model-d-1-cons-u}
    & u_{jk}\in \left\{0,1\right\},\ j=1,\ldots,N,\ k=1,\ldots,T.
\end{align}
\end{subequations}
The objective function \eqref{eq:model-d-1-obj} is the objective function \eqref{eq:model-c-1-obj} evaluated at the approximated final operator $X_T\approx X(t_f)$. Constraints
\eqref{eq:model-d-1-cons-h}--\eqref{eq:model-d-1-cons-u} are the discretized
formulations of constraints
\eqref{eq:model-c-1-cons-h}--\eqref{eq:model-c-1-cons-u}. 
In the following sections we use this discretized formulation to develop our algorithms. We present the following discussion on the relationship between the continuous and discretized formulation. 
\begin{remark}
Problem \eqref{eq:model-d-1} is equivalent to \eqref{eq:model-c-1} for 
piecewise-constant controls $u_j,\ j=1,\ldots,N$. In addition, we have
\begin{align}
    \int_{0}^{t_{k}} u_j(\tau)d\tau = \sum_{\tau=1}^{k} u_{j\tau} \Delta t,\ j=1,\ldots,N,\ k=1,\ldots,T.
\end{align}
\end{remark}

\begin{remark}
The model can be generalized to instances with multiple active controllers in any time interval by considering each possible combination of the controllers. For an instance with $L$ control Hamiltonians $H^{(1)}, \ldots,H^{(L)}$, we consider $N=2^L$ combinations $H^{(l_1,\ldots,l_k)},\ \forall \left\{l_1,\ldots,l_k\right\}\subseteq \left\{1,\ldots,L\right\}$ and define control functions for each combination. 
Thus, the original problem is converted to a new problem with only one active controller in any time interval.
\end{remark}
In our paper we consider four examples with different objective functions and
control Hamiltonians. 
In the following sections we introduce the 
continuous model of four example problems in quantum control. They can be 
formulated as discretized models \eqref{eq:model-d-1} following the above discretization process.

\subsection{Energy Minimization Problem}
\label{sec: model-energy}

Consider a
quantum spin  system consisting of $q$ qubits, no intrinsic Hamiltonian, and two
control Hamiltonians $H^{(1)},\ H^{(2)}$. The initial operator
$X_{\textrm{init}}$ is a $2^q$-dimensional identity matrix. Let $\ket{\psi_0}$
be the ground state of the first control Hamiltonian $H^{(1)}$. The ground state energy of this quantum system is the minimum eigenvalue of the control Hamiltonian $H^{(2)}$, represented by $E_{\textrm{min}}$ in our model. 
Since the ground state energy differs across problem instances, it is important to compare performance consistently across these instances. Therefore, we maximize the approximation ratio, defined as the achieved energy divided by the true ground state energy. The problems we study are well defined so that the initial state of the system has zero energy with respect to $H^{(2)}$; furthermore, the ground state energy will always be a large negative number. Therefore, the approximation ratio reflects the energy level achieved by our procedure
from the initial energy of zero to the true ground state energy.  Because our procedure can  improve the state  relative only to its initial configuration, a negative approximation ratio is not possible.

To ensure consistency with the objective functions of our other two examples, we minimize one minus the approximation ratio. Because the true ground state energy is negative, this formula is equivalent to minimizing the achieved energy. 
The continuous-time formulation is 
\begin{subequations}
\label{eq:model-energy-c}
\begin{align}
\label{eq:model-energy-c-obj}
    \min\quad  \displaystyle & 1 - \left \langle \psi_0\right| X(t_f)^\dagger H^{(2)} X(t_f) \left|\psi_0 \right\rangle / E_{\textrm{min}}\\
    \textrm{s.t.}\quad %
                    & H(t) = u_1(t)H^{(1)} + u_2(t)H^{(2)},\ \forall t\in [0,t_f]\\
                    \label{eq:model-energy-h1}
                    & H^{(1)} = -\sum_{i=1}^q \sigma_i^x
                    ,\ H^{(2)} = \sum_{ij} J_{ij}\sigma_i^z \sigma_j^z\\
                    & \textrm{Constraints \eqref{eq:model-c-1-cons-s}--\eqref{eq:model-c-1-cons-u-sum}}\nonumber\\
                    & u_1(t),\ u_2(t) \in \left\{0,1\right\}, 
\end{align}
\end{subequations}
where $\cdot^\dagger$ represents the conjugate transpose of a matrix and
$\langle \cdot |$ represents the conjugate transpose of a quantum state vector
$\ket{\cdot}$. The matrix $[J_{ij}],\ i,j=1,\ldots,q$ is the adjacency matrix
of a randomly generated graph with $q$ nodes, and $\sigma_i^x,\ \sigma_i^z$ are 
Pauli matrices of qubit $i$. 

\subsection{CNOT Gate Estimation Problem}
\label{sec: model-cnot}
Our second example is defined on an isotropic Heisenberg
spin-1/2 chain with length 2 according to~\citep{pawela2016various}. For our study we consider just the unitary evolution with no energy leakage.  The quantum system includes 2 qubits, an
intrinsic Hamiltonian $H^{(0)}$, and two control Hamiltonians $H^{(1)},
H^{(2)}$. The initial operator $X_\textrm{init}$ is a $4$-dimensional identity
matrix and the target operator is the matrix formulation of the CNOT
gate represented by $X_\textrm{CNOT}\in \mathbb{C}^{4\times 4}$. This problem's continuous-time formulation is 
\begin{subequations}
\label{eq:model-cnot-c}
\begin{align}
    \min_{u(t), X(t), H(t)}\quad  & 1 - \frac{1}{4}\left|\operatorname{tr}\left\{X_\textrm{CNOT}^\dagger X(t_f)\right \}\right| \\
    \textrm{s.t.}\quad  & H(t) = H^{(0)} + u_1(t)H^{(1)} + u_2(t)H^{(2)},\ \forall t\in [0,t_f]\\
                    \label{eq:model-cnot-h0}
                    & H^{(0)} = \sigma^x_{1} \sigma^x_{2}+\sigma^y_{1} \sigma^y_{2}+\sigma^z_{1} \sigma^z_{2}\\
                    & H^{(1)} = \sigma^x_1
                    ,\ H^{(2)} = \sigma^y_1\\
                    & \textrm{Constraints \eqref{eq:model-c-1-cons-s}--\eqref{eq:model-c-1-cons-i}}\nonumber\\
                    & u_1(t),\ u_2(t) \in \left\{0,1\right\}, 
\end{align}
\end{subequations}
where $\sigma_i^x,\ \sigma_i^y,\ \sigma_i^z,\ i=1,2$ are Pauli matrices of two qubits. 
The goal is to minimize the objective function, defined as the infidelity
between the final operator and the target operator. 

\subsection{NOT Gate Estimation with Energy Leakage}
\label{sec: model-not}
We consider an anharmonic multilevel system that is widely used in quantum experiments to implement physical qubits. Following~\citep{Motzoi2009}, we consider the lowest three levels of an energy spectrum with only nearest-level coupling for a single qubit $q=1$. The first two levels comprise the qubit and the third level models the energy leakage. Under this setting, the system contains an intrinsic Hamiltonian $H^{(0)}$ and two control Hamiltonians $H^{(1)},\ H^{(2)}$. The initial operator $X_\textrm{init}$ is a 3-dimensional identity matrix, and the target operator is the matrix formulation of the NOT gate corresponding to the lowest three levels of the energy spectrum represented by $X_\textrm{NOT}\in \mathbb{C}^{3\times 3}$, which is
\begin{align}
    X_{\textrm{NOT}} = \begin{pmatrix}
    0 & 1 & 0\\
    1 & 0 & 0\\
    0 & 0 & 0\\
    \end{pmatrix}.
\end{align}
The continuous-time formulation of this problem is 
\begin{subequations}
\label{eq:model-not-c}
\begin{align}
    \min_{u(t), X(t), H(t)}\quad  & 1 - \frac{1}{2}\left|\operatorname{tr}\left\{X_\textrm{NOT}^\dagger X(t_f)\right \}\right| \\
    \textrm{s.t.}\quad  & H(t) = H^{(0)} + u_1(t)H^{(1)} + u_2(t)H^{(2)},\ \forall t\in [0,t_f]\\                \label{eq:model-not-h0}
                    & H^{(0)} = \mu_1 |1\rangle \langle 1| + \mu_2 |2\rangle \langle 2|\\ 
                    \label{eq:model-not-h1}
                    & H^{(1)} = {\omega_1}(|0\rangle \langle 1| + |1\rangle \langle 0|)
                    + {\omega_2}(|1\rangle \langle 2| + |2\rangle \langle 1|)\\
                    \label{eq:model-not-h2}
                    & H^{(2)} = {\omega_1}(i|0\rangle \langle 1| -i |1\rangle \langle 0|)
                    + {\omega_2}(i|1\rangle \langle 2| -i |2\rangle \langle 1|)\\
                    & \textrm{Constraints \eqref{eq:model-c-1-cons-s}--\eqref{eq:model-c-1-cons-i}}\nonumber\\
                    & u_1(t),\ u_2(t) \in \left\{0,1\right\}, 
\end{align}
\end{subequations}
where $|j\rangle$ represents a vector with the $j$-th element as 1 and other elements as 0. The parameters $\mu_1,\ \mu_2$ weigh the relative strength of transitions and the parameters $\omega_1,\ \omega_2$ correspond to the drive frequency. We present the specific values of the parameters used in numerical simulations in Section~\ref{sec: exp-design}.

\subsection{Circuit Compilation Problem}
\label{sec: model-compilation}
For a general quantum algorithm, each quantum circuit needs to be compiled to a representation imposed by specific controllers and constraints in order to execute the algorithm on specific quantum devices. 
We take quantum circuits generated by the unitary coupled-cluster single-double (UCCSD) method~\citep{bartlett2007coupled, romero2018strategies} for estimating the ground state energy of molecules in quantum chemistry as examples. 
We consider the Hamiltonian controllers in the gmon-qubit system~\citep{chen2014qubit} with $q$ qubits. Each qubit has a flux-drive controller and a charge-drive controller. In addition, there is a rectangular-grid topology with nearest-neighbor connectivity denoted by $E$, and each pair of connected qubits $(j_1, j_2)\in E$ has a controller. Define $w_{j_1j_2}(t)$ as control functions for the connected qubits controllers and use $w(t)$ to represent the corresponding vector form of $w_{j_1j_2}(t)$. The initial operator $X_\textrm{init}$ is a $2^q$-dimensional identity matrix. The target
operator is the matrix formulation of the quantum circuit generated by UCCSD for specific molecules represented by
$X_\textrm{UCCSD}\in \mathbb{C}^{2^q\times 2^q}$.  
We refer interested readers to \citet{gokhale2019partial} for more details. The continuous-time formulation of this problem is
\begin{subequations}
\label{eq:model-molecule-c}
\begin{align}
    \min\quad  & 1 - \frac{1}{2^q}\left|\operatorname{tr}\left\{X_\textrm{UCCSD}^\dagger X(t_f)\right \}\right| \\
    \textrm{s.t.}\quad  & H(t) = \sum_{j=1}^{2q} u_j(t)H^{(j)} + \sum_{(j_1,j_2)\in E} w_{j_1j_2}(t)H^{(j_1j_2)},\ \forall t\in [0,t_f]\\
                    \label{eq:model-molecule-h1-1}
                    & H^{(2j-1)} = J_c \sigma^x_{j},\ j=1,\ldots,q\\
                    \label{eq:model-molecule-h1-2}
                    & H^{(2j)} = J_f \left(\begin{array}{ll}
                                0 & 0 \\
                                0 & 1
                                \end{array}\right),\ j=1,\ldots,q\\
                    \label{eq:model-molecule-hc}
                    & H^{(j_1j_2)} = J_e \sigma^x_{j_1}\sigma^x_{j_2},\ (j_1, j_2)\in E\\
                    & \textrm{Constraints \eqref{eq:model-c-1-cons-s}--\eqref{eq:model-c-1-cons-u-sum}}\nonumber\\
                    & u_j(t) \in \left\{0,1\right\},\ w_{j_1j_2}(t)\in \left\{0, 1\right\},\ j=1,\ldots,2q,\ (j_1, j_2)\in E,
\end{align}
\end{subequations}
where $J_c,\ J_f,\ J_e$ are parameters corresponding to the quantum system and $\sigma_j^x$ is the Pauli matrix of qubit $j$. 
We finish this section by showing a general property of the objective function in optimal control problems.
\begin{theorem}
For a quantum system with $q$ qubits, the infidelity $1 - \frac{1}{2^q}\left|\operatorname{tr}\left\{X_{\text {targ }}^{\dagger}
X(t_f)\right\}\right| \in [0,1]$. The objective functions of the optimal control models for CNOT gate \eqref{eq:model-cnot-c} and circuit compilation problem \eqref{eq:model-molecule-c} are bounded between
zero and one.
\end{theorem}
\begin{proof}
By definition, the objective function value cannot be larger than $1$. From the definition of our quantum control problem, the quantum operators $X_\textrm{targ}$ and $X(t_f)$ are both unitary matrices. Therefore, we can write $X(t_f)$ as $X(t_f) = \sum_{i=1}^{2^q} x_i e_i^T$, where $x_i$ are the orthonormal column vectors of $X(t_f)$ and $e_i$ are $2^q\times 1$ vectors with the $i$th component as $1$ and all the other components as $0$. Thus we have
\begin{align}
    \left|\operatorname{tr}\left\{X_{\text {targ }}^{\dagger} X(t_f)\right\}\right| 
    = \left| \sum_{i=1}^{2^q} \operatorname{tr}\left\{X_{\text {targ }}^{\dagger} x_i e_i^T\right\}\right|
    \leq \sum_{i=1}^{2^q}  \left\|X_{\text {targ }}^{\dagger} x_i\right\|_\infty
    \leq 2^q.
\end{align}
The first inequality follows from the fact that the trace of a product is
invariant under cyclic permutations of the factors. The second inequality
follows from $\left\|X_{\text {targ }}^{\dagger} x_i\right\|_\infty\leq 1$. As a
result, the objective function is no less than $0$, and this completes the proof. 
\end{proof}
\begin{remark}
All the problems in this section fall into a class known as right-invariant
systems~\citep{jurdjevic1972control} that consists of control problems on Lie
groups.  Specifically, this class consists of all control problems described by
a Hamiltonian of the form $H(t) = H^{(0)}+\sum_{j} u_j(t)H^{(j)}$.
Controllability results on this class of control problems are well known~\citep{jurdjevic1972control,ramakrishna1995control}.  A test of controllability
of these problems consists of examining the dynamical Lie algebra generated by
the Hamiltonians, $H^{(j)}$, and their nested commutators and comparing that Lie
algebra with the Lie algebra of the subgroup defined by the symmetries in a given
problem. For the systems considered in this paper, there are indeed enough
control Hamiltonians to reach the target state by using piecewise continuous control functions,
$u_j(t)$.
\end{remark}

\section{Binary Quantum Control Algorithms}
\label{sec: alg-binary}
In this section we propose models and algorithms to obtain binary quantum
control results for the discretized model \eqref{eq:model-d-1}. 
In Section~\ref{sec: alg-grape} we review the adjoint approximation method of gradients used in 
the GRAPE algorithm for solving the continuous relaxation of the optimal control problem 
\eqref{eq:model-d-1} and derive the gradient
update rule for the aforementioned four specific examples. 
In Section~\ref{sec: alg-penalty} we propose an exact penalty function for the SOS1
property that allows us to add constraints to GRAPE. In Section~\ref{sec: alg-rounding} we introduce 
rounding techniques to obtain binary controls. These algorithms form the basis of our discrete framework discussed in Sections~\ref{sec: alg-infrequent}--\ref{sec: alg-tr}.

The overall framework is outlined in Figure~\ref{fig:framework}, which presents the methods in Sections~\ref{sec: alg-binary}--\ref{sec: alg-tr} and the corresponding tables and results. We summarize the algorithms of our framework in Table~\ref{tab:list-alg}. 
In the figure, the rectangles represent the three
formulations of our model, and the ellipses represent the algorithms. The
parallelograms show the result tables corresponding to algorithms, where the
dashed parallelograms refer to the results of the continuous models and the solid parallelograms refer to the results of the discrete models.
\tikzstyle{block} = [rectangle, draw, fill=white, text width=8em, text centered, minimum height=4em, font=\fontsize{10}{10}\selectfont]
\tikzstyle{line} = [draw, -latex']
\tikzstyle{cloud} = [draw, ellipse, text width=5.5em, text centered, node distance=3cm, minimum height=4em, font=\fontsize{10}{10}\selectfont]
\tikzstyle{cloud2} = [draw, ellipse, aspect=3, text width=8.5em, text centered, node distance=3cm, minimum height=4em, font=\fontsize{10}{10}\selectfont]
\tikzstyle{result} = [draw, trapezium, trapezium left angle = 65, trapezium right angle = 115, trapezium stretches, text width=8em, text centered, minimum width=2.5cm, minimum height=2em, font=\fontsize{10}{10}\selectfont]
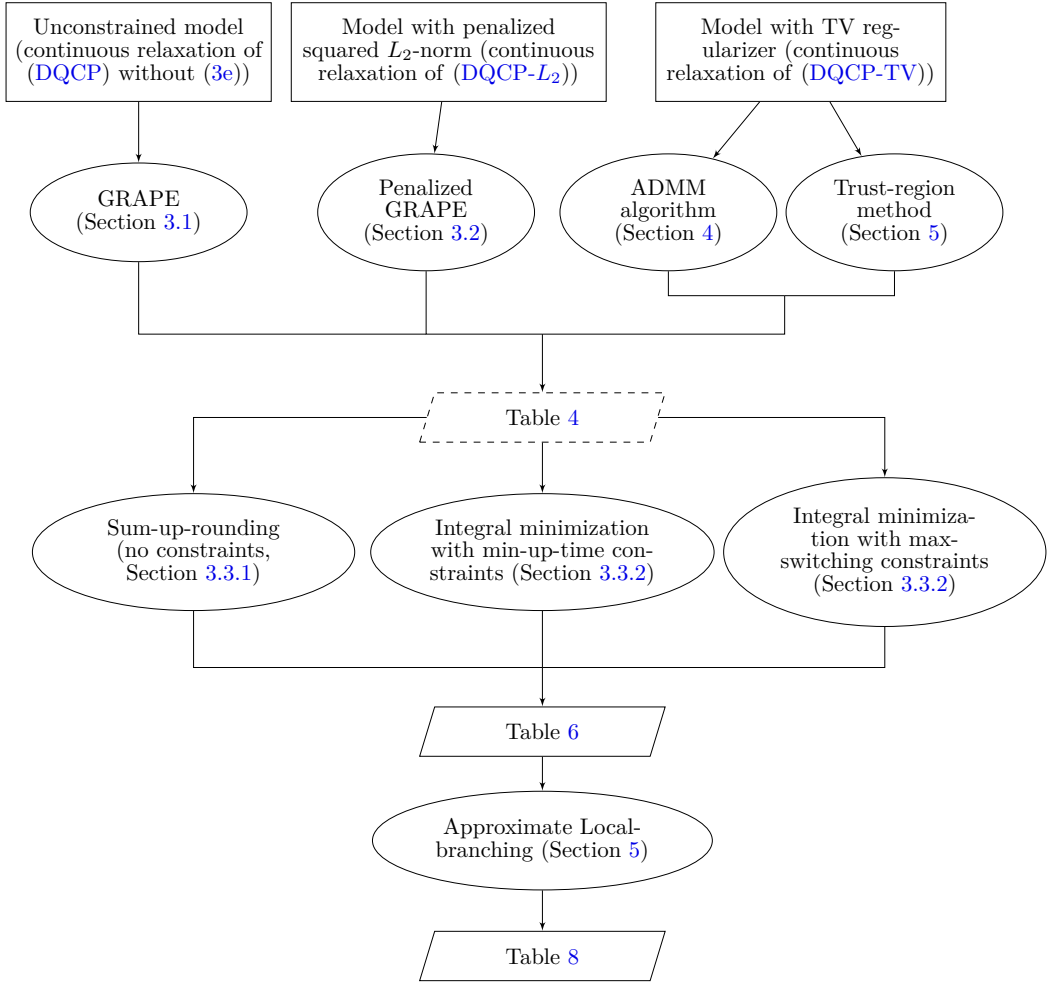
\begin{figure}[ht!]
\centering
\begin{adjustbox}{scale=0.85}
\begin{tikzpicture}[node distance=5cm, auto]
    \node [block, text width=10em] (unconstrained) at (-1.5,0) {Unconstrained model (continuous relaxation of \eqref{eq:model-d-1} without \eqref{eq:model-d-1-cons-u-sum})};
    \node [block, right of=unconstrained, xshift=-0.2cm, text width=12em] (L2) {Model with penalized squared $L_2$-norm (continuous relaxation of \eqref{eq:model-d-l2-p})};
    \node [block, right of=L2, text width=11em, xshift=0.45cm] (TV) {Model with TV regularizer (continuous relaxation of \eqref{eq:model-tv})};
    \node [cloud, below of=unconstrained, yshift=0.5cm] (grape) {GRAPE (Section~\ref{sec: alg-grape})};
    \node [cloud, right of=grape, xshift=1.45cm] (pGRAPE) {Penalized GRAPE (Section~\ref{sec: alg-penalty})};
    \node [cloud, right of=pGRAPE,xshift=0.75cm] (ADMM) {ADMM algorithm (Section~\ref{sec: alg-infrequent})};
    \node [cloud, right of=ADMM,xshift=0.5cm] (TR) {Trust-region method (Section~\ref{sec: alg-tr})};
    \node[coordinate, below of=pGRAPE, yshift=3.1cm, xshift=1.8cm] (coor1) {};
    \node[coordinate, below of=ADMM, yshift=3.7cm, xshift=1.8cm] (coor2) {};
    \node[coordinate, below of=coor2, yshift=4.4cm] (coor4) {};
    \node [result, below of=coor1, yshift=3.7cm, dashed] (resc) {Table~\ref{tab: res-obj-c}};
    \node[cloud2, below of=resc, yshift=0.9cm, xshift=-5.4cm] (sur) {Sum-up-rounding (no constraints, Section~\ref{sec: alg-sur})};
    \node[cloud2, right of=sur, xshift=2.4cm, text width=9.1em] (minup) {Integral minimization with min-up-time constraints (Section~\ref{sec: alg-bnb})};
    \node[cloud2, right of=minup, xshift=2.3cm] (maxswitch) {Integral minimization with max-switching constraints (Section~\ref{sec: alg-bnb})};
    \node[coordinate, below of=minup, yshift=3.2cm] (coor5) {};
    \node[result, below of=coor5, yshift=4cm] (resr) {Table~\ref{tab: res-obj-r}};
    \node[cloud2, below of=resr, yshift=1.3cm, text width=9em] (ALB)
    {Approximate Local-branching (Section~\ref{sec: alg-tr})};
    \node[result, below of=ALB, yshift=3.2cm] (resalb) {Table~\ref{tab: res-obj-improve}};
    \path [line] (unconstrained) -- (grape);
    \path [line] (L2) -- (pGRAPE);
    \path [line] (TV) -- (TR);
    \path [line] (TV) -- (ADMM);
    \draw (pGRAPE) |- (coor1);
    \draw (grape) |- (pGRAPE|-coor1);
    \draw (TR) |- (coor2);
    \draw (ADMM) |- (coor2);
    \draw (coor2) |- (coor1);
    \path [line] (coor1) -- (resc);
    \path [line] (resc) -| (sur);
    \path [line] (resc) -- (minup);
    \path [line] (resc) -| (maxswitch);
    \draw (sur) |- (coor5);
    \draw (minup) |- (coor5);
    \draw (maxswitch) |- (coor5);
    \path [line] (coor5) -- (resr);
    \path [line] (resr) -- (ALB);
    \path [line] (ALB) -- (resalb);
    
\end{tikzpicture}
\end{adjustbox}
\caption{Overview and results of the algorithmic framework.}
\label{fig:framework}
\end{figure}

\begin{table}[h!]
    \centering
        \caption{List of acronyms of algorithms in numerical simulations. The pGRAPE and GRAPE algorithms are identical when the SOS1 constraint is absent. }
    \begin{adjustbox}{width=\textwidth}
        \begin{tabular}{ll}
        \hline
         Name& Algorithms \\
         \hline
         pGRAPE 
         & Solving the continuous relaxation with squared $L_2$-penalty function by penalized GRAPE. \\
         \hline
         TR& Solving the continuous relaxation with TV regularizer and SOS1 property by a trust-region method.\\
         \hline
         ADMM& Solving the continuous relaxation with TV regularizer and SOS1 property by ADMM.\\
         \hline
         SUR & Rounding the continuous solutions without hard control constraints.\\
         \hline
         MT& Rounding the continuous solutions with min-up-time constraints.\\
         \hline
         MS& Rounding the continuous solutions with max-switching constraints.\\
         \hline
        ALB & Improving the binary solutions by the approximate local-branching method. \\
         \hline
    \end{tabular}
    \end{adjustbox}   
    \label{tab:list-alg}
\end{table}

\subsection{GRAPE Approach for Solving Continuous Relaxations}
\label{sec: alg-grape}
The GRAPE algorithm~\citep{khaneja2005optimal} is a gradient descent algorithm for 
solving the unconstrained continuous control problem. 
We derive a gradient descent algorithm based on the adjoint approach used in the GRAPE algorithm 
to solve the continuous relaxation of our discretized model \eqref{eq:model-d-1} without the SOS1 property 
enforced by constraint \eqref{eq:model-d-1-cons-u-sum}. We eliminate variables $X$ and $H$ by converting them into implicit functions of the control variables $u$ using constraints 
\eqref{eq:model-d-1-cons-h}--\eqref{eq:model-d-1-cons-i}. Then the minimization problem of $F(X_T)$ on $u,\ X,\ H$ is converted to the 
minimization problem of $F(X_T(u))$ over the variables $u$. 
The goal of the adjoint approach is to approximate the gradient of the objective function $F$ with
respect to control variables $u_{jk},\ j=1,\ldots,N,\ k=1,\ldots,T$ based on the discretized 
time horizon. For
simplicity, we define propagators for each time step $k=1,\ldots,T$ as $U_k
= \exp\left\{-iH_k\Delta t\right\}$. For small $\Delta t$, the gradient of $U_k$
corresponding to $u_{jk}$ is estimated as  
\begin{align}
    \frac{\partial U_k}{\partial u_{jk}}= -i\Delta t H^{(j)} U_k,\ j=1,\ldots,N,\ k=1,\ldots,T.
\end{align}
The gradient of the objective function with respect to the propagators depends
on its specific formulation. We discuss gradients of two specific functions for 
the examples in Section~\ref{sec: model}. 

\paragraph{Energy function} For generality, we use $\bar{H}$ to represent the
Hamiltonian in the energy objective function, which means that the goal is to
minimize the energy corresponding to $\bar{H}$. Specifically, we have $\bar{H}
= H^{(2)}$ in the example in Section~\ref{sec: model-energy}. The general
energy function can be expressed as
\begin{align}
    \langle \psi_0| X_T^\dagger \bar{H} X_T \ket{\psi_0} = \langle \psi_0| U_1^\dagger \cdots U_T^\dagger \bar{H} U_T \cdots U_1 \ket{\psi_0}.
\end{align}
For each time step $k=1,\ldots,T-1$, define variables $\ket{\kappa_k} = U_{k+1}^\dagger \cdots U_{T}^\dagger \bar{H}X_T\ket{\psi_0}$. For the last time step $T$, we define a variable $\ket{\kappa_T} = \bar{H}X_T\ket{\psi_0}$. Then the gradient with respect to $u_{jk}$ is computed as 
\begin{align}
    \frac{\partial F}{\partial u_{jk}}= \frac{2}{E_\textrm{min}}\textrm{Re}\left[i\Delta t \langle \kappa_{k}| H^{(j)} X_k\ket{\psi_0}\right], \ j=1,\ldots,N,\ k=1,\ldots,T.
\end{align}
\paragraph{Infidelity function} The infidelity function can be expressed as
\begin{align}
    1 - \frac{1}{2^q}\left|\operatorname{tr}\left\{X_{\textrm{targ}}^{\dagger} X_T\right\}\right|
     = 1 - \frac{1}{2^q}\left|\operatorname{tr}\left\{X_{\textrm{targ}}^\dagger U_T\cdots U_1 X_0\right\}\right|.
\end{align}
For each time step $k=1,\ldots,T-1$, define variables $\lambda_k = U_{k+1}^\dagger \cdots U_T^\dagger X_\textrm{targ}$. For the last time step $T$, we define a variable $\lambda_T = X_\textrm{targ}$.
Then the gradient of the trace with respect to $u_{jk}$ is computed as 
\begin{align}
    \frac{\partial \operatorname{tr}\left\{X_{\textrm{targ}}^{\dagger} X_T\right\}}{\partial u_{jk}}= -i\Delta t \lambda_{k}^\dagger H^{(j)} X_k,\ j=1,\ldots,N,\ k=1,\ldots,T.
\end{align}
Using the definition of objective function $F$, we compute the gradient as  
\begin{align}
  \frac{\partial F}{\partial u_{jk}}= \frac{1}{2^q}\textrm{Re}\left[i\Delta t \,\textrm{tr}\left\{\lambda_{k}^\dagger H^{(j)} X_k\right\} e ^{-i\textrm{arg} \left(\operatorname{tr}\left\{X_{\textrm{targ}}^{\dagger} X_T\right\}\right)}\right],\ j=1,\ldots,N,\ k=1,\ldots,T, 
\end{align}
where $\textrm{arg}(\cdot)$ represents the argument of a complex number. 

With these computed gradients, we minimize the reduced function $F(X_T(u))$ with respect to control variables $u\in [0,1]^{N\cdot T}$ by L-BFGS-B \citep{byrd1995limited}, a well-known quasi-Newton algorithm for solving bound-constrained optimization problems. 
\begin{remark}
Our implementation of GRAPE extends the classical gradient-ascend scheme (see, e.g.,~\citep{khaneja2005optimal,Larocca_2021}) in a number of ways: We use a bound-constrained quasi-Newton method to solve problems with bounded control variables, and we introduce a penalized GRAPE (denoted by pGRAPE) that adds a quadratic  penalty of the SOS1 property in Section \ref{sec: alg-penalty}.
\end{remark}

We take the circuit compilation problem (Section~\ref{sec: model-compilation}) on 
the molecule H$_2$ (dihydrogen) as an example to show the control results, 
and we present the continuous control results without
the SOS1 property obtained by the GRAPE algorithm in Figure~\ref{fig:pre-grape-control}.
We define the absolute violation of the SOS1 property at each time step $k=1,\ldots,T$ as $|\sum_{j=1}^N u_{jk} - 1|$ and present the value at each time step in Figure~\ref{fig:pre-grape-violation}. The figure shows that the results violate the SOS1 property at all time steps, with the maximum absolute violation being $2.404$. Hence we introduce a penalized function in Section~\ref{sec: alg-penalty} to impose the SOS1 property. 
\begin{figure}[ht!]
    \centering
    \subfloat[Control results of $5$ controllers \label{fig:pre-grape-control}]{\includegraphics[width=0.5\textwidth]{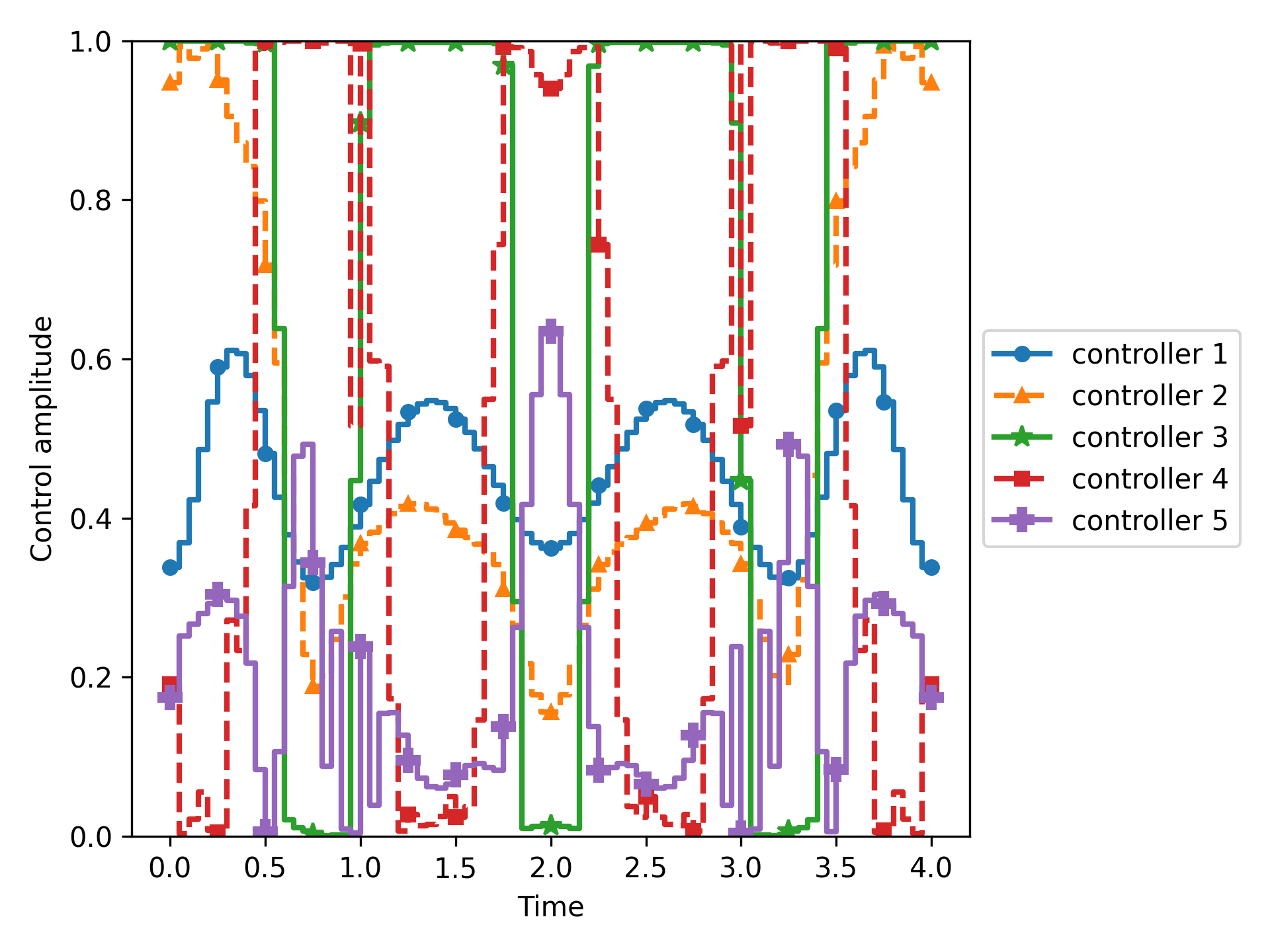}}
    \subfloat[Absolute violation at each step \label{fig:pre-grape-violation}]{\includegraphics[width=0.5\textwidth]{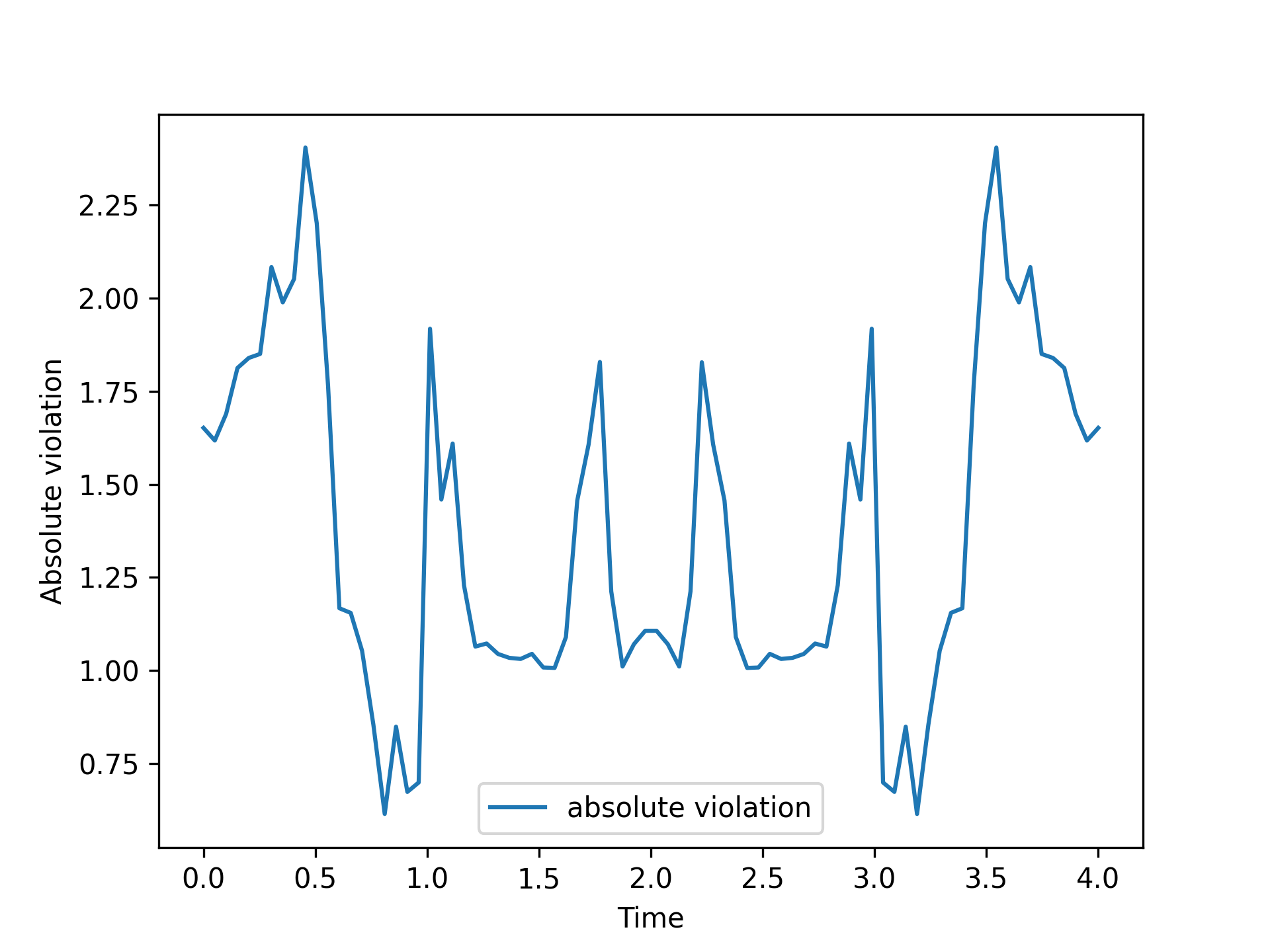}}
    \caption{(a) Control results of the continuous relaxation of the discretized model \eqref{eq:model-d-1} without the SOS1 property constraint \eqref{eq:model-d-1-cons-u-sum} of the circuit compilation example on the molecule H$_2$. 
    (b) Absolute violation at each time step.}
    \label{fig:pre-grape}
\end{figure}
\subsection{Penalty Function for SOS1 Property}
\label{sec: alg-penalty}
Our GRAPE algorithm is defined only for unconstrained or bound-constrained problems~\citep{khaneja2005optimal}. 
Hence, we convert the SOS1 property \eqref{eq:model-d-1-cons-u-sum} into a squared
$L_2$-penalty term and add it to the objective function. Define the 
squared $L_2$ function of the constraint violation as 
\begin{align}
    l(u, T) = \sum_{k=1}^T \left(\sum_{j=1}^N u_{jk} - 1\right)^2.
\end{align}
We let a constant $\rho> 0$
be the penalty parameter. Then  the model with the penalty function is 
\begin{align}
\label{eq:model-d-l2-p} 
\tag{DQCP-$L_2$}
    \min \ F(X_T) + \rho l(u,T) \quad
    \textrm{s.t.} \ \textrm{Constraints \eqref{eq:model-d-1-cons-h}--\eqref{eq:model-d-1-cons-u-sum}}\nonumber.
\end{align}
We use $l(u^*_\rho,T)$ to denote the optimized value of the  squared $L_2$ function $l(u,T)$ under the penalty parameter $\rho$. 
The following theorem discusses the exactness of the penalized objective function. 
\begin{theorem}
\label{theo: exactness}
[Exactness of Squared $L_2$-Penalty] Let $z_1$ and $z_2(\rho)$ be the optimal values of \eqref{eq:model-d-1} and \eqref{eq:model-d-l2-p}, respectively. 
There exists $\tilde{\rho}<\infty$ such that $z_1=z_2(\rho)$ for all $\rho>\tilde{\rho}$.
\end{theorem}
A more general version of Theorem~\ref{theo: exactness} is proved in \citep{sinclair1986exact}. 
Next we discuss the value of the optimized squared $L_2$-penalty term with respect to the penalty parameter $\rho$ for the continuous relaxation of \eqref{eq:model-d-l2-p}. 
\begin{theorem}
\label{theo:l2-bound}
For the continuous relaxation of \eqref{eq:model-d-l2-p}, if the original objective function $F$ is upper bounded by a constant $C_F$ and \eqref{eq:model-d-1} is feasible, then for any $T$ the optimized squared $L_2$-penalty term uniformly holds that $|l(u^*_\rho, T)|\leq 2C_F/\rho$.
\end{theorem}
A detailed proof is provided in Appendix~\ref{app: proof-sos1}. 
\begin{remark}
The GRAPE algorithm is still able to solve the unconstrained continuous relaxation with the penalty function \eqref{eq:model-d-l2-p}. We propose a penalized version (pGRAPE) in Section \ref{sec: alg-grape} by adding a term $2\rho \left(\sum_{j=1}^N u_{jk}-1\right)$ to the approximated gradient of the objective function. 
\end{remark}
\begin{remark}
For the quantum control problem with two controllers, the SOS1 can be enforced directly by substituting 
\begin{align}
    u_{2k} = 1 - u_{1k},\ k=1,\ldots,T
\end{align}
in the model instead of using the penalty function. 
\end{remark}
We show the control results obtained from the continuous relaxation with the squared $L_2$-penalty function and penalty parameter $\rho = 1.0$ in Figure~\ref{fig:pre-sos1}. 
The value of the squared penalty term is $5.55\times 10^{-7}$, leading to a small error between continuous and binary results with the SOS1 property obtained by rounding techniques. However, we still need to provide rounding techniques to obtain binary results. 

\begin{figure}[ht!]
    \centering
    \includegraphics[width=0.5\textwidth]{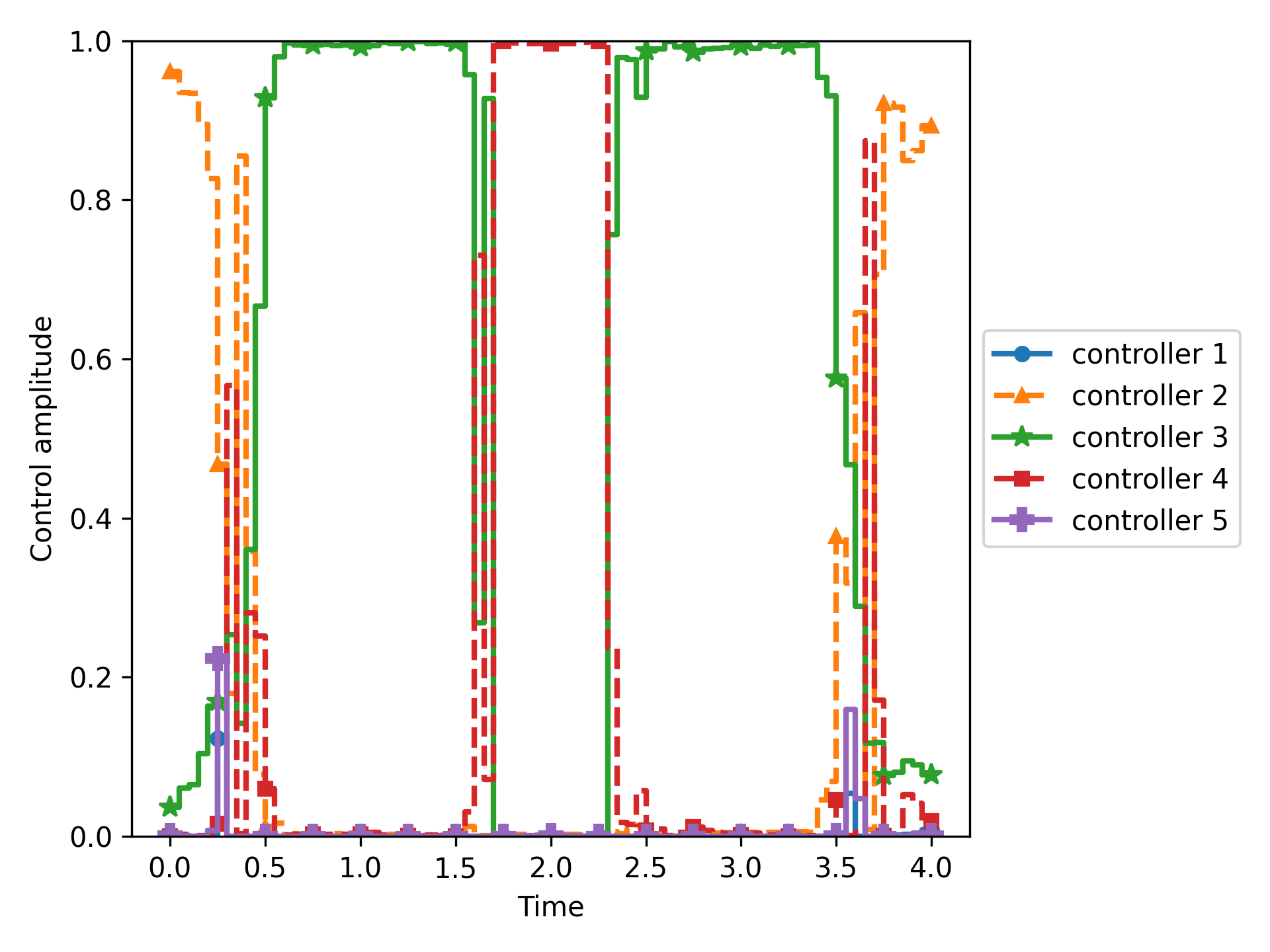}
    \caption{Control results of the continuous relaxation of the discretized model with a squared $L_2$-penalty function \eqref{eq:model-d-l2-p} of the circuit compilation example on the molecule H$_2$ (objective value: 4.37E$-$07).}
    \label{fig:pre-sos1}
\end{figure}
\subsection{Rounding Techniques and Optimality Guarantees}
\label{sec: alg-rounding}
In this section we introduce rounding techniques to obtain binary control results and investigate their optimality guarantees. In Section~\ref{sec: alg-sur} we review the sum-up rounding technique designed for binary controls without constraints and discuss the difference compared with the continuous results. In Section~\ref{sec: alg-bnb} we introduce the integral minimization problem for rounding to obtain restricted binary controls. 
\subsubsection{Sum-Up Rounding}
\label{sec: alg-sur}
The sum-Up rounding (SUR) strategy proposed by \citet{sager2012integer} is a well-known method to obtain integer controls from continuous, relaxed ones in optimal control theory; see, for example, \citep{you2011mixed, manns2020multidimensional}. 
One concern when applying SUR is that the solution of the relaxation does not satisfy the SOS1 constraints because the exactness of the squared $L_2$-penalty is only for binary controls. We show that as long as the violation of the SOS1 constraint is small, the solution that is constructed via SUR still satisfies the strong convergence properties of SUR. 
We define the discretized continuous control function $u^c$ as $u^c_{jk}\in [0,1],\ j=1,\ldots,N,\ k=1,\ldots,T$. We  define the discretized binary control function $u^b$ as $u^b_{jk}\in \left\{0,1\right\},\ j=1,\ldots,N,\ k=1,\ldots,T$. For each time step $k=1,\ldots,T$, we denote the vector form of control variables as $u^c_k$ and $u^b_k$. 
We define the current cumulative deviation between continuous and binary controls as $\hat{p}_{jk},\ j=1,\ldots,N,\ k=1,\ldots,T$. The SUR approach to obtain binary controls is described in Algorithm~\ref{alg:sur}.

\begin{algorithm}[!ht] \caption{Sum-Up Rounding with the SOS1 Property. \label{alg:sur} } 
    \DontPrintSemicolon
    \DontPrintSemicolon
    \SetNoFillComment
    \KwInput{Continuous control $u^c$ on uniform discretization.}
    \tcc{Iterate over each time step}
    \For{$k=1,\ldots,T$}
    {\tcc{Iterate over each controller}
    \For{$j=1,\ldots,N$} 
    {Compute cumulative deviation as 
    $\displaystyle \hat{p}_{jk} = \sum_{\tau=1}^k u^c_{j\tau} \Delta t -\sum_{\tau=1}^{k-1} u^b_{j\tau } \Delta t$.}
    Choose controller $j^* = \argmax_{j=1,\ldots,N} \hat{p}_{jk}$. If there is a tie, 
    we break the tie by choosing the smallest index.\\
    Set binary control $u^b_{j^*k}=1$ and $u^b_{jk}=0,\ \forall j\neq j^*$.}
    \KwOutput{Binary control $u^b$.}
\end{algorithm}

The construction of $u_{jk}^b$ ensures that the binary control satisfies the SOS1 property. \citet{sager2012integer} proved that if the continuous control $u^c$ satisfies the SOS1 property, then the rounded control $u^b$ will converge to $u^c$ when the length of a time interval $\Delta t$ converges to zero. This does not hold for arbitrary continuous controls without the SOS1 property, however, as the next proposition shows. 
\begin{proposition}
\label{prop: nosos1}
Let continuous control $u^c$ and binary control $u^b$ defined in Algorithm \ref{alg:sur} be given for $N\geq 2$. Define
\begin{align}
    \epsilon(\Delta t) \coloneqq \max_{k=1,\ldots,T} \left |\sum_{\tau=1}^{k} \left(\sum_{j=1}^N u^c_{j\tau} - 1\right) \Delta t \right |.
\end{align}
If $t_f <\infty $, then $u^b$ satisfies the SOS1 property $\sum_{j=1}^{N} u^b_{jk} = 1$ at any time step $k=1,\ldots,T$, and we have 
\begin{align}
    \max_{k=1,\ldots,T} \left\|\sum_{\tau=1}^{k} \left( u^c_{\tau} - u^b_{\tau}\right) \Delta t\right\|_\infty \geq \frac{1}{N}\epsilon(\Delta t).
\end{align}
\end{proposition}
The proof is provided in Appendix~\ref{app: proof-sur}. 
We have the following estimate of the difference between the continuous control without the SOS1 property and the rounded control. 
\begin{theorem}
\label{theo:sur-convergence}
With the same definition of $\epsilon (\Delta t)$ in Proposition~\ref{prop: nosos1}, it holds that for any time step $k=1,\ldots,T$, 
\begin{align}
    \label{eq:sur-convergence}
    \left\|\sum_{\tau=1}^k \left( u^c_{\tau} - u^b_{\tau}\right)\Delta t \right\|_\infty \leq \left(N-1\right) \Delta t + \frac{2N - 1}{N} \epsilon (\Delta t).
\end{align}
\end{theorem}
Because $u^b$ and $u^c$ are bounded, we have $\epsilon(\Delta t)<\infty$ if $t_f<\infty$. 
The proof of the inequality \eqref{eq:sur-convergence} extends the proof for Theorem 5 in~\citep{sager2012integer}. We modify the right-hand side of the inequality by adding a term $\frac{2N-1}{N}\epsilon(\Delta t)$. We provide the detailed proof in Appendix~\ref{app: proof-sur}. 
In the following corollary we show that for our bounded discretized quantum control problem, $\epsilon(\Delta t)$ converges to zero as $\Delta t$ converges to zero, 
ensuring the convergence of the rounded control to the continuous control. 
\begin{corollary}
\label{cor:sur-convergence-1}
Let $l(u^c,T)$ be the value of the optimized squared $L_2$ term in the continuous relaxation of the model \eqref{eq:model-d-l2-p}. 
Then it holds that 
\begin{align}
    \label{eq:sur-convergence-prop}
    \epsilon(\Delta t)\leq \sqrt{t_f{l(u^c,T)}\Delta t},
\end{align}
where $t_f$ is the evolution time. Furthermore, if the original objective function $F$ is bounded, the rounded control will converge to the continuous control when $\Delta t$ is small enough.
\end{corollary}
The proof is provided in Appendix~\ref{app: proof-sur}. 
Based on the convergence of binary results to continuous results, we present the following proposition to guarantee the optimality of binary results. 
\begin{proposition}
\label{prop: b-c-convergence}
Under the discretized setting of the quantum control problem, let $u^c$ be the continuous control and $u^b$ be the binary control obtained by Algorithm \ref{alg:sur}. 
Then the state of time evolution with binary control $X^b$ converges to the state with continuous control $X^c$ at each time step,  namely, 
\begin{align}
    \lim_{\Delta t\rightarrow 0} X^b_k = \lim_{\Delta t \rightarrow 0} X^c_k,\ k=1,\ldots,T, 
\end{align}
leading to $\displaystyle \lim_{\Delta t\rightarrow 0} F(X_T^b) = \lim_{\Delta t\rightarrow 0} F(X_T^c)$ for a 
continuous objective function. 
\end{proposition}
The proof is provided in Appendix~\ref{app: proof-sur}.

\subsubsection{Combinatorial Integral Approximation}
\label{sec: alg-bnb}
To obtain discretized binary controls, \citet{sager2005numerical} proposed a more general rounding technique by minimizing the integral difference between continuous and binary controls with certain additional constraints on the binary controls called combinatorial integral approximation (CIA). Let $\mathcal{U}_B\subseteq \left\{0,1\right\}^{N\cdot T}$ be the feasible region for binary controls. For each controller $j=1,\ldots,N$ and each time step $k=1,\ldots,T$, define $u^c_{jk}$ and $u^b_{jk}$ as the discretized continuous and binary controls. We use $u^c$ and $u^b$ to represent the corresponding vector forms. The optimization problem for rounding is formulated as a mixed-integer problem 
\begin{subequations}
\label{eq:model-milp-rounding}
\begin{align}
    \min_{u^b} \quad & \max_{j=1,\ldots,N} \max_{k=1,\ldots,T} \left|\sum_{\tau=1}^k (u^c_{j\tau} - u^b_{j\tau})\Delta t\right|\\
    \textrm{s.t.} \quad & u^b\in \mathcal{U}_B.
\end{align}
\end{subequations}
The rounding optimization problem can be solved by diverse algorithms for solving integer programs. In this paper we choose a branch-and-cut algorithm~\citep{wolsey2020integer}. 
Furthermore, the rounding result obtained from SUR is an optimal solution of model \eqref{eq:model-milp-rounding} with 
$\mathcal{U}_B = \left\{0,1\right\}^{N\cdot T}$~\citep{sager2005numerical}. 
In practice, researchers and engineers have proposed diverse restrictions on binary controls to avoid frequent switches. 
We formulate two main types of constraints as linear constraints. 

\paragraph{Min-Up-Time Constraints.} Let $T_\textrm{minup}$ be the minimum number of steps $\Delta t$ that the controller is active. The min-up-time constraints enforce that each controller is active for at least $T_\textrm{minup}$ time steps. The restricted feasible region is formulated by the following linear constraints: 
\begin{align}
    \mathcal{U}_B = \Bigg\{(u,v): & -v_{jk}\leq u_{jk} - u_{jk+1}\leq v_{jk},\ j=1,\ldots,N,\ k=1,\ldots,T-1 \nonumber \\
    & \sum_{k=t}^{t+T_{\textrm{minup}}-1} v_{jk}\leq 1,\ j=1,\ldots, N,\ t=1,\ldots,T-T_{\textrm{minup}} \Bigg\}.
\end{align}
\paragraph{Max-Switching Constraints.} Let $S$ be the maximum number of switches during the evolution time horizon $[0,t_f]$. The max-switching constraints enforce an upper bound $S$ on the total number of switches for each controller. We describe the restricted feasible region by linear constraints as 
\begin{align}
    \mathcal{U}_B = \Bigg\{(u,v): & -v_{jk}\leq u_{jk} - u_{jk+1}\leq v_{jk},\ j=1,\ldots,N,\ k=1,\ldots,T-1 \nonumber \\
    & \sum_{k=1}^{T} v_{jk}\leq S,\ j=1,\ldots, N\Bigg\}.
\end{align}

We present the binary results obtained from SUR and CIA with min-up-time constraints in Figure~\ref{fig:pre-sur}. We show that SUR leads to chattering on singular arcs. Although CIA can prevent frequent switches by setting hard constraints on the controls, it leads to a serious objective value increase. In Section~\ref{sec: alg-infrequent} and Section~\ref{sec: alg-tr} we propose models and algorithms to reduce switches when taking the original objective function into account. 

\begin{figure}[ht!]
    \centering
    \subfloat[SUR results (objective value $0.021$)]{\includegraphics[width=0.5\textwidth]{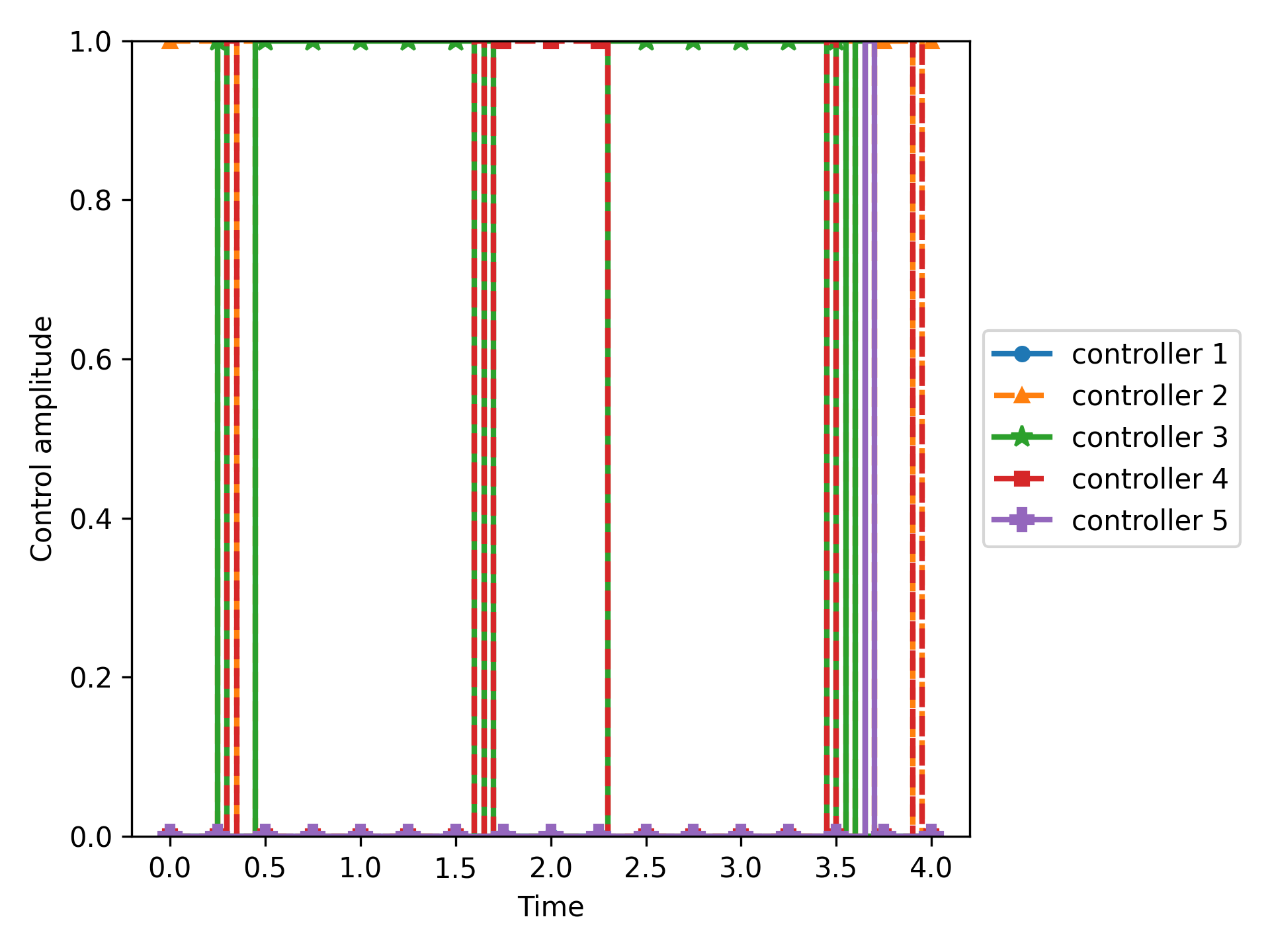}}
    \subfloat[CIA results with min-up time constraints (objective value $0.600$) ]{\includegraphics[width=0.5\textwidth]{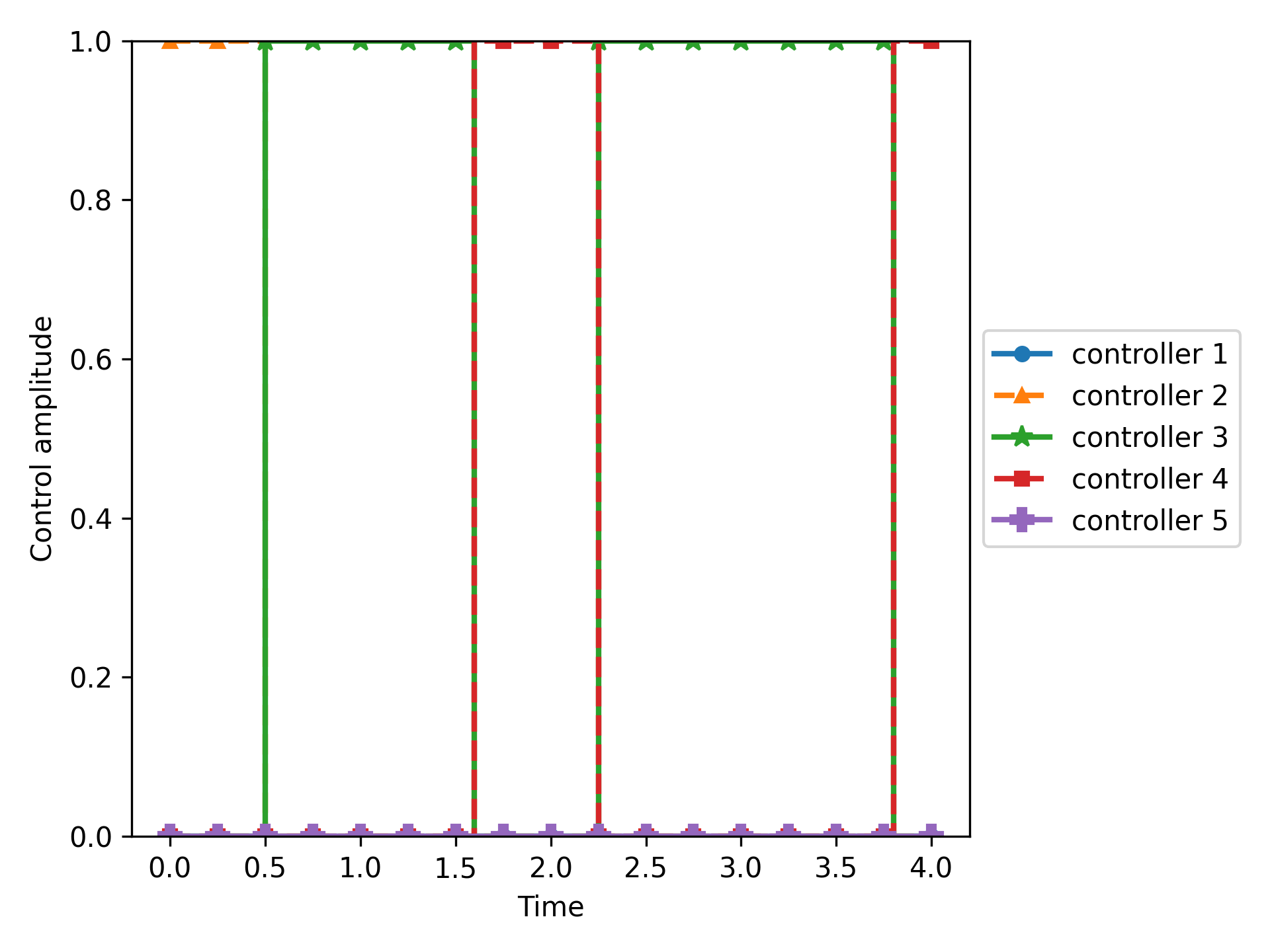}}
    \caption{Binary control results obtained by conducting SUR and CIA on continuous results of the circuit compilation example using the molecule H$_2$. The min-up time constraints reduce the switches but increase the objective value.}
    \label{fig:pre-sur}
\end{figure}
\section{Model with the TV Regularizer}
\label{sec: alg-infrequent}
Binary controls that avoid chattering on singular arcs can be obtained by solving the rounding model \eqref{eq:model-milp-rounding}. However, the objective function of the rounding model %
is not relevant to the original objective function $F$, leading to a significant increase in the value of $F$, as illustrated in Figure~\ref{fig:pre-sur}. Here we investigate a TV regularizer as an alternative strategy to obtain solutions with fewer switches. 

The TV of a function is defined as the integral of the absolute change of 
the function over the entire space. Since~\citet{rudin1992nonlinear} 
first introduced the TV regularizer by considering functions describing images, it 
has become a popular method in image noise reduction; see, for example,~\citep{condat2013direct,kunisch2004total}. 
We refer interested readers to~\citep{rodriguez2013total} for a detailed review of applying a TV regularizer to 
denoise images. \citet{stella2017simple} propose an algorithm to solve the continuous control 
problem with the TV regularizer based on the forward-backward envelope \citep{themelis2018forward}. 
\citet{sager2021mixed} considered the TV of integer control functions  
to reduce the absolute change of controls. The authors added a constraint imposing that 
the TV regularizer should be no more than a threshold. 
However, they considered the TV regularizer only when rounding the control results of continuous 
relaxation, which leads to an increase in the original objective function. 
\citet{leyffer2021sequential} extended the trust-region method to solve the integer optimal control problem with the TV regularizer in the objective function, but it is still highly dependent on the initial values. 

We use the TV regularizer to penalize the absolute change between control variables of two consecutive time steps, which is defined as 
\begin{align}
    TV(u) = \sum_{j=1}^N \sum_{k=1}^{T-1} |u_{jk} - u_{jk+1}|
\end{align}
for a control variable vector $u$. 
Let $\alpha>0$ be the parameter for the TV regularizer.  
Then the continuous relaxation model with the TV regularizer is formulated as 
\begin{align}
\label{eq:model-tv} \tag{DQCP-TV}
    \min \ & F(X_T) + \rho l(u, T) + \alpha TV(u)\quad \textrm{s.t.}\ \textrm{Constraints \eqref{eq:model-d-1-cons-h}--\eqref{eq:model-d-1-cons-i}},\ u\in [0,1]^{N\cdot T}.
\end{align}
The first term  is the original objective function, the second term is the squared $L_2$-penalty function for the SOS1 property, and the third term is the TV regularizer function. 
Because of the $L_1$ term, we cannot use GRAPE;  instead, we propose ADMM to solve the continuous relaxation. For each controller $j=1,\ldots,N$ and time step $k=1,\ldots,T-1$, we introduce auxiliary variables $v_{jk} = u_{jk} - u_{jk+1}$ to describe the change between control variables of two consecutive time steps. We reformulate the model as follows: 
\begin{subequations}
\label{eq:model-tv-admm}
\begin{align}
    \min \quad & F(X_T)+ \rho l(u, T) + \alpha \sum_{j=1}^N \sum_{k=1}^{T-1} |v_{jk}|\\
    \textrm{s.t.} \quad & 
    \label{eq:model-tv-admm-dual}
    v_{jk} = u_{jk} - u_{jk+1},\ j=1,\ldots,N,\ k=1,\ldots,T-1\\
    & \textrm{Constraints \eqref{eq:model-d-1-cons-h}--\eqref{eq:model-d-1-cons-i}}\nonumber\\
    & u\in [0,1]^{N\cdot T}.
\end{align}
\end{subequations}
Define $\mu_{jk},\ j=1,\ldots,N,\ k=1,\ldots,T-1$ as the dual variables corresponding to constraints \eqref{eq:model-tv-admm-dual}. We use $u,\ v,\ \mu$ to denote the corresponding vector forms. Let fixed parameters $\beta>0, \delta > 0$ be the Lagrangian penalty parameter and stopping criterion threshold. The update procedure of ADMM consists of three steps: (i) updating variables $u$, (ii) updating variables $v$, and (iii) updating dual variables $\mu$. We solve the minimization problem for updating variables $u$ by the modified GRAPE algorithm with a squared $L_2$-penalty function and Lagrangian penalty function. We derive an exact form for the update of variables $v$. We use gradient descent to update the dual variables $\mu$. The specific procedure for updating is presented in Algorithm~\ref{alg:tv-admm}.

\begin{algorithm}[ht] \caption{ADMM Algorithm for Solving Continuous Relaxation of \eqref{eq:model-tv}. \label{alg:tv-admm}}   
    \KwInput{Initial values of variables $u^0,\ v^0,\ \mu^0$, and number of ADMM iterations $L$.}
    Initialize iteration $l=1$.\\
    \While{$l\leq L$ and the algorithm does not converge}
    {Update variables $u$ as 
    \begin{align*}
        u^{l} = & \argmin_{u\in \mathcal{U}_C} \ F(X_T)+ \rho l(u, T) + \frac{\beta}{2} \sum_{j=1}^N \sum_{k=1}^{T-1}\|u_{jk}-u_{jk+1}-v_{jk}^{l-1}+\mu_{jk}^{l-1}\|^2\\
        & \textrm{s.t. Constraints \eqref{eq:model-d-1-cons-h}--\eqref{eq:model-d-1-cons-i}}.
    \end{align*}
    \For{$k=1,\ldots,T$}
    {\For{$j=1,\ldots,N$}
        {Update variables $v_{jk}$ as 
        $$v^{l}_{jk} = \begin{cases}
  u^{l}_{jk} - u^{l}_{jk+1}+\mu^{l-1}_{jk} - \alpha / \beta, & \text{if } u^{l}_{jk} - u^{l}_{jk+1}+\mu^{l-1}_{jk} > \alpha / \beta, \\
  u^{l}_{jk} -u^{l}_{jk+1}+\mu^{l-1}_{jk} + \alpha / \beta, & \text{if } u^{l}_{jk} - u^{l}_{jk+1}+\mu^{l-1}_{jk} < -\alpha / \beta, \\
  0, & \text{otherwise}.
  \end{cases}$$
        Update dual variables $\mu_{jk}$ as 
        $\displaystyle \mu_{jk}^{l} = \mu_{jk}^{l-1} + (u_{jk}^{l} - u_{jk+1}^{l} - v_{jk}^{l}).$}}
    \If{$\sum_{j=1}^N \sum_{k=1}^{T-1} \|u^l_{jk} - u^l_{jk+1} - v^l_{jk}\|^2 \leq \delta$}
    {Break the loop. }
    Update $l \leftarrow l+1$.}
    \KwOutput{Final solutions of variables $u^l,\ v^l,\ \mu^l$.}
\end{algorithm}

We present the continuous control results obtained from the model \eqref{eq:model-tv} in Figure~\ref{fig:pre-tv} for the circuit compilation example on the molecule H$_2$. We test the parameter of the  TV regularizer $\alpha=10^{-5},\ 10^{-4},\ 10^{-3},\ 10^{-2},\  10^{-1}$ and the Lagrangian parameter $\beta=0.1,\ 0.5,\ 1.0$. We choose $\alpha=10^{-3}$ and $\beta=0.5$, which has the smallest objective value after rounding with min-up-time constraints. The number of switches in control results decreases significantly as compared with the results in Figure~\ref{fig:pre-sos1}, showing the benefits of the TV regularizer.

\begin{figure}[ht!]
    \centering
    \includegraphics[width=0.5\textwidth]{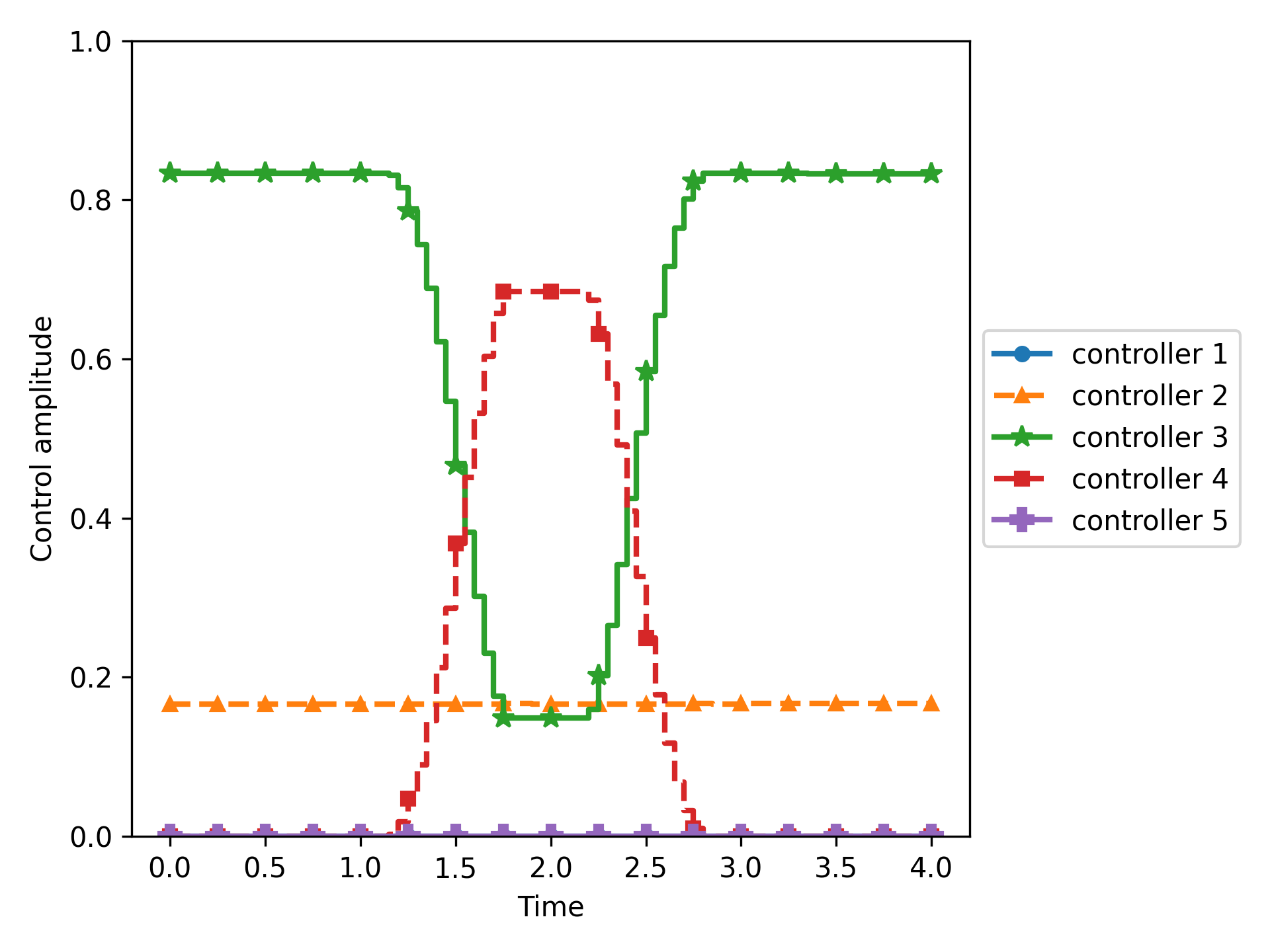}
    \caption{Control results of the continuous relaxation of the discretized model with the TV regularizer \eqref{eq:model-tv} of the circuit compilation example on the molecule H$_2$ (objective value: 1.33E$-$05).}
    \label{fig:pre-tv}
\end{figure}

\section{Improvement Heuristic: Approximate Local-Branching Method}
\label{sec: alg-tr}
Because the ADMM algorithm does not ensure global optimality for nonlinear objective functions and the objective value increases after rounding, the control results can be improved by applying an approximate local-branching (ALB) heuristic. 
The trust-region method~\citep{nocedal2006numerical} is a widely used local search method in optimization but highly depends on the initial values. 
In this section we propose a trust-region subproblem for the quantum control problem and then modify the trust-region method to solve the problem starting from the solutions obtained in Sections~\ref{sec: alg-binary}--\ref{sec: alg-infrequent} to improve the quantum controls. 

We introduce our improvement heuristic based on the approach of~\citep{leyffer2021sequential} to solve the binary model with the TV regularizer. Given a feasible point $\hat{u}$, we use the first-order gradient to approximate the objective value around the point $\hat{u}$ and define the trust-region subproblem with radius $R$ as 
\begin{subequations}
\label{eq:model-tr}
\begin{align}
    \label{eq:model-tr-obj}
    \min_u\quad & \langle \nabla F(\hat{u}), u - \hat{u}\rangle_{L_2} + 
     \alpha TV(u) - \alpha TV(\hat{u}) \\
    \label{eq:model-tr-radius}
    \textrm{s.t.} \quad & \|u - \hat{u}\|_1\leq R\\
    & \textrm{Constraints \eqref{eq:model-d-1-cons-u-sum}--\eqref{eq:model-d-1-cons-u}} \nonumber.
\end{align}
\end{subequations}
We note that \eqref{eq:model-tr-obj} approximates the objective value $F(u)$ but uses an exact form of the TV regularizer. 
Constraint \eqref{eq:model-tr-radius} indicates that we consider only the points with $L_1$ distance to $\hat{u}$ no more than $R$. 
For each controller $j=1,\ldots,N$ and time step $k=1,\ldots,T-1$, define variables $v_{jk}$ as the upper bound of the absolute change between the control values of two consecutive time steps. The trust-region subproblem is reformulated as the following mixed-integer linear program: 
\begin{subequations}
\label{eq:model-tr-milp} 
\begin{align}
    \label{eq:model-tr-milp-obj}
    \min_u\quad & \langle F(\hat{u}), u - \hat{u}\rangle_{L_2} + \alpha \sum_{k=1}^{T-1} \sum_{j=1}^N v_{jk} - \alpha TV(\hat{u}) \\
    \label{eq:tr-mip-cons-1}
    \textrm{s.t.} \quad & \sum_{k=1}^T \left(\sum_{j:\hat{u}_{jk}=0} u_{jk} + \sum_{j:\hat{u}_{jk}=1} (1-u_{jk})\right)\leq R\\
    \label{eq:tr-mip-cons-3}
    & -v_{jk}\leq u_{jk} - u_{jk+1}\leq v_{jk},\ j=1,\ldots,N,\ k=1,\ldots,T-1\\
    & \textrm{Constraints \eqref{eq:model-d-1-cons-u-sum}--\eqref{eq:model-d-1-cons-u}} \nonumber.
\end{align}
\end{subequations}
The objective function \eqref{eq:model-tr-milp-obj} is a reformulation of \eqref{eq:model-tr-obj} with variables $v$. Constraint \eqref{eq:tr-mip-cons-1} is the trust-region constraint. Constraints \eqref{eq:tr-mip-cons-3} ensure that for each controller $j=1,\ldots,N$, $v_{jk}\geq |u_{jk} - u_{jk+1}|,\ k=1,\ldots,T-1$. 
Let $\bar{u}$ be an optimal solution of the model \eqref{eq:model-tr-milp}. Define $\Delta F_p(\hat{u}, \bar{u})$ and $\Delta F_a(\hat{u}, \bar{u})$ respectively as the predictive and actual decrease of the objective function with the following formulations: 
\begin{subequations}
\label{eq:compute-decrease}
\begin{align}
    \Delta F_p(\hat{u}, \bar{u}) & =  \langle \nabla F(\hat{u}), \hat{u} - \bar{u}\rangle_{L_2} + \alpha TV(\hat{u}_{jk}) - \alpha TV(\bar{u}_{jk}),\\
    \Delta F_a(\hat{u}, \bar{u}) & = F(\hat{u}) - F(\bar{u}) + \alpha TV(\hat{u}_{jk}) - \alpha TV(\bar{u}_{jk}).
\end{align}
\end{subequations}
The trust-region  algorithm consists of an inner loop and an outer loop. In the inner loop we solve a sequence of trust-region subproblems a with monotonically decreasing radius until the ratio between the actual decrease and predictive decrease is large enough. To obtain a balance between the computational cost and the searched area, we set a threshold $\bar{R}$ for the radius. When the radius is greater than the threshold, we decrease it according to a geometric sequence~\citep{leyffer2021sequential}; otherwise, we decrease it by an arithmetic sequence. In the outer loop we repeat the inner loop for each updated point until the predictive decrease is nonpositive. This procedure is described in Algorithm~\ref{alg:tr-b}. 
\begin{algorithm}[!htbp] \caption{Trust-Region Method for Quantum Control.\label{alg:tr-b} }   
    \KwInput{Starting radius $R^0>0$, threshold of radius $\bar{R}>0$, initial feasible point $u^0$, and threshold of decrease $\eta>0$.}
    Initialize predictive decrease $\Delta F_p(u^0, \bar{u})=\infty$ and actual decrease $\Delta F_a(u^0, \bar{u})=-\infty$. 
    Set the number of outer iterations $l\leftarrow 0$.\\ 
    \While{$\Delta F_p(u^l, \bar{u}) > 0$}
    {Initialize radius $R = R^0$.\\
    \While{$\Delta F_a(u^l, \bar{u}) < \eta \Delta F_p(u^l, \bar{u})$}
    {Solve model \eqref{eq:model-tr-milp} with $\hat{u} = u^l$ to obtain the minimizer $\bar{u}$. \\
    Compute predictive and actual decrease $\Delta F_p(u^l, \bar{u})$ and $\Delta F_a(u^l, \bar{u})$ by \eqref{eq:compute-decrease}.\\
    \eIf{$R > \bar{R}$}
    {$R \leftarrow \max\left\{\lfloor R / 2 \rfloor, \bar{R}\right\}$.}
    {$R \leftarrow R - 1$.}}
    Set $l\leftarrow l + 1$. Update the central point $u^l \leftarrow \bar{u}$.}
    \KwOutput{Control results $u^l$.}
\end{algorithm}
\begin{remark}
If we relax the feasible region of each controller to $[0,1]$, the trust-region method can solve the continuous relaxation of the model with the TV regularizer. Trust-region methods show convergence to stationary points on continuous relaxation if we allow the radius $R$ to be real and adjusted accordingly (Corollary 3.9 in~\citep{aravkin2021proximal}).
\end{remark}
The trust-region approach can also improve solutions obtained by CIA approaches from Section~\ref{sec: alg-bnb}. 
Based on the restricted feasible region $\mathcal{U}_B$ instead of the TV regularizer, we propose the trust-region subproblem with additional linear constraints as follows: 
\begin{subequations}
\label{eq:tr-mip-r}
\begin{align}
    \label{eq:tr-h-obj}
    \min_u\quad & \langle \nabla F(\hat{u}), u - \hat{u}\rangle_{L_2}\\
    \textrm{s.t.} \quad & \textrm{Constraints  \eqref{eq:model-d-1-cons-u-sum}--\eqref{eq:model-d-1-cons-u}, \eqref{eq:tr-mip-cons-1}}\nonumber\\
    \label{eq:tr-h-feasible}
    & (u,v) \in \mathcal{U}_B.
\end{align}
\end{subequations}
The objective function \eqref{eq:tr-h-obj} is a first-order gradient approximation of the original objective function $F(u)$, and constraint \eqref{eq:tr-h-feasible} restricts feasible regions for controls. The computations of the predictive decrease and actual decrease are modified as 
\begin{subequations}
\label{eq:compute-decrease-hard}
\begin{align}
    \Delta F_p(\hat{u}, \bar{u}) & = \langle \nabla F(\hat{u}), \hat{u} - \bar{u}\rangle_{L_2},\\
    \Delta F_a(\hat{u}, \bar{u}) & = F(\hat{u}) - F(\bar{u}).
\end{align}
\end{subequations}
For the trust-region algorithm (Algorithm~\ref{alg:tr-b}), we replace solving model \eqref{eq:model-tr-milp} with solving model \eqref{eq:tr-mip-r} and compute the decrease by \eqref{eq:compute-decrease-hard}. The other parts remain the same. 

In theory, all time-evolution processes in our algorithms can be conducted on quantum computers. The inputs for the quantum computers are the control sequences, and we require the output of the objective function and its (approximate) gradient. The updates of variables are still computed on classical computers.

\section{Numerical Results}
We apply our algorithmic framework proposed in Sections~\ref{sec: alg-binary}--\ref{sec: alg-tr} to diverse instances of the four examples introduced in Section~\ref{sec: model}. In Section~\ref{sec: results-solver} we show the results from state-of-the-art optimization solvers as the baseline. In Section~\ref{sec: exp-design} we introduce the design of numerical instances and parameter settings. In Sections~\ref{sec: res-c}--\ref{sec: res-i} we present the numerical results, including the continuous relaxation results, binary results by combinatorial integral approximation, and binary results after an improvement heuristic. 
\label{sec: results}
\subsection{State-of-the-Art Optimization Solvers}
\label{sec: results-solver}
The NEOS server is a frequently used internet-based service containing several state-of-the-art solvers for numerical optimization problems~\citep{czyzyk_et_al_1998,dolan_2001,gropp_more_1997}. 
For nonlinear constrained continuous problems, SNOPT~\citep{andrei2017sqp} and
IPOPT~\citep{wachter2006implementation} are two widely used solvers. 
For mixed-integer nonlinear constrained problems, BARON~\citep{sahinidis1996baron} and Couenne~\citep{belotti2009couenne} are commonly
used for obtaining global optimal solutions, while MINLP~\citep{belotti2013mixed,leyffer2010software}, SCIP~\citep{gamrath2020scip},
and Bonmin~\citep{bonami2008algorithmic} are effective solvers with good performance for finding local optima. 

As a baseline we first apply the optimization solvers from the NEOS server to solve the energy minimization example problem in Section~\ref{sec: model-energy}. All the experiments are conducted on the online server. 
We set the number of qubits $q=2$ and $q=6$. 
We set the evolution time $t_f=2$ and the number of time steps $T=40$. We notice that for all the evolution times in Section~\ref{sec: results}, we use dimensionless units, with $\hbar=1$. 
We apply the solver MUSCOD-II (6.0)~\citep{kirches2013taco} to solve the continuous formulation with the differential equations \eqref{eq:model-c-1}, and we set the maximum iterations to $100$ for $q=2$ and $700$ for $q=6$. 
To apply the state-of-the-art nonlinear optimization solvers to a discretized model, we derive a discretized formulation of the ordinary differential equation  by the implicit Euler method (see, 
e.g., ~\citep{butcher2016numerical}) as 
\begin{align}
    X_k-X_{k-1} = -iH_kX_k \Delta t,\ k=1,\ldots,T.
\end{align}
We apply the solvers SNOPT (7.6.1) and IPOPT (3.13.4) to solve the continuous relaxation model and apply the solvers MINLP, Bonmin (1.8.8), Couenne (0.5.8), and SCIP (7.0.3.5) to solve the binary model. 
We set the time limit to $5$ minutes when $q=2$ and $80$ minutes when $q=6$. 
In Table~\ref{tab:res-solver-energy} we present the objective values, TV-norm values, CPU times, and explored nodes (only for binary solvers). 

Table~\ref{tab:res-solver-energy} shows that MUSCOD-II is good for 
solving the continuous
relaxation for the instance with $q=2$, but it takes a long time to solve the 
large instance with $q=6$. 
SNOPT and IPOPT perform worse than the MUSCOD-II solver especially when the quantum systems include more qubits. 
For the binary control
problem, all the methods reach the time limit.  Bonmin obtains the best
solution, but the gaps between their results and true energy are all large. 
For the instance with $q=6$, some optimization solvers such as MINLP and
Couenne run out of memory.
The CPU time increases and the number of explored nodes decreases significantly when the size of the quantum system increases. 

\begin{table}[ht]
  \centering
  \caption{Results of solvers on energy minimization example. The results are marked by ``OOM'' if a solver runs out of memory and ``LIMIT'' if a solver reaches the time or iteration limit. The explored nodes of continuous solvers are marked by ``-'' because the node exploration process is conducted only by binary solvers.}
    \begin{tabular}{l|rrr|rrrrr}
    \hline
    \multirow{2}{*}{Solver} & \multicolumn{3}{c|}{$q=2$} & 
    \multicolumn{3}{c}{$q=6$}\\
    \cline{2-7}
    & Obj& Time (s) & Nodes & Obj & Time (s) & Nodes \\
    \hline
    SNOPT & 0.084 & 0.18 & - & 0.736 & 199.64 & -\\
    \hline
    IPOPT & 0.084 & 0.10 & - & 0.736 & 46.85 & - \\
    \hline
    MUSCOD-II & $5.90$E$-09$ & 1.67 & - & 
    0.154 & 4399.54 & - \\
    \hline
    MINLP & 0.150 & LIMIT & 9056 & OOM & OOM & OOM\\
    \hline
    Bonmin & 0.148 & LIMIT & 11672 & 
    0.793 & LIMIT & 355 \\
    \hline
    Couenne & 0.471 & LIMIT & 6932 & OOM & OOM & OOM\\
    \hline
    SCIP & 0.149 & LIMIT & 25679 & 0.000 & LIMIT & 2480 \\
    \hline
    \end{tabular}%
  \label{tab:res-solver-energy}%
\end{table}%

This experiment shows that existing standard nonlinear programming and mixed-integer nonlinear programming 
solvers cannot solve discretized optimal control problems. 
In the following sections we present the numerical results of our proposed algorithms on multiple instances. 
Our pGRAPE+SUR method obtains binary solutions with an objective value $4.22$E$-04$ in $0.14$ seconds for $q=2$ and binary solutions with an objective value $0.157$ in $27.95$ seconds for $q=6$. 
We show that our methods can obtain better results and with shorter computational time than do the state-of-art solvers. 

\subsection{Experimental Design and Parameter Settings}
\label{sec: exp-design}
When the problems have no SOS1 property, the squared $L_2$ term in the penalized model 
\eqref{eq:model-d-l2-p} is eliminated, so solving \eqref{eq:model-d-l2-p} by pGRAPE is equivalent to 
solving the original model \eqref{eq:model-d-1} by GRAPE. Because the energy minimization problem 
has only two controllers, we eliminate the SOS1 property directly by substituting $u_2(t)=1-u_1(t)$. We 
remove the SOS1 property in the CNOT gate estimation problem to show the generality of our models. We solve 
the model with a squared $L_2$-penalty function for the circuit compilation problem. 

For all the examples, we apply the pGRAPE algorithm to solve the continuous
relaxation. We also employ the TR and ADMM algorithms to solve the continuous relaxation
with the TV regularizer. Then we obtain the binary results without constraints by
SUR, the binary results with min-up-time constraints,
and the binary results with max-switching constraints by integral minimization. 
Furthermore, we apply the approximate local-branching methods to improve the binary solutions
with the TV regularizer and hard control constraints. 

For the energy minimization problem, when $q=2$, the matrix $J$ is two-dimensional and includes only one independent element because it is symmetric. Hence, we set diagonals in $J$ to $0$ and other elements to $1$. When $q=4,\ 6$, we generate $5$ instances where each instance has a random symmetric matrix $J$ with zero diagonals and elements within a range $[-1,1]$. We will present averaged objective value results for these problems. The instances are represented by Energy2, Energy4, and Energy6 corresponding to the number of qubits $q=2,\ 4,\ 6$ in our following results. For the CNOT problem, we conduct experiments with evolution times $t_f=5,\ 10,\ 15,\ 20$ and represent the corresponding instances as CNOT5, CNOT10, CNOT15, CNOT20, respectively. For the NOT gate estimation problem, we follow the parameter settings in~\citep{Motzoi2009} and set $\mu_1=0,\ \mu_2=2\pi,\ \omega_1=1,\ \omega_2=\sqrt{2}$. We conduct numerical simulations with evolution times $t_f=2,\ 6,\ 10$ and represent the instances as NOT2, NOT6, and NOT10, respectively.  For the circuit compilation problem, we test two molecules, H$_2$ (dihydrogen) and LiH (lithium hydride), and generate the UCCSD circuits with minimum energy by the VQE algorithm in the Python package Qiskit~\citep{aleksandrowicz2019qiskit}. The instances are represented by CircuitH2 and CircuitLiH. We set the parameters corresponding to quantum systems as $J_c=0.2\pi,\ J_f=3\pi,\ J=0.1\pi$. We test values of the penalty parameter for the squared $L_2$-penalty function $\rho=10^{-6},\ 10^{-5},\ 10^{-4},\ 10^{-3},\ 10^{-2},\ 10^{-1},\ 1,\ 10$, and we choose the one with the smallest objective value among the parameters with a squared penalized term less than $10^{-6}$. 
We set the penalty parameter to be $\rho=1.0$ and $\rho=0.1$ for two the instances, respectively. We test values for the TV parameter $\alpha = 10^{-5},\ 10^{-4},\ 10^{-3},\ 10^{-2}$ and choose $\alpha$ with the smallest rounding objective value with min-up-time constraints. The settings of $\alpha$ and other parameters are presented in Table~\ref{tab: setting}. All the numerical simulations were conducted in Python 3.8 on a MacOS computer with 8 cores, 16 GB of RAM, and a 3.20 GHz processor. All the computational time results are the times on classical computers. 
\begin{table}[ht!]
  \centering
  \caption{Parameter settings of examples. The parameters include the number of qubits ($q$), number of controllers ($N$), evolution time ($t_f$), number of time steps ($T$), TV parameter ($\alpha$), minimum uptime steps ($T_\textrm{minup}$), and maximum switches ($S$).}
    \begin{tabular}{lrrrrrrr}
    \hline
    Instance & $q$ & $N$ & $t_f$ & $T$ & $\alpha$ & $T_\textrm{minup}$ & $S$ \\
    \hline
    Energy2 & 2     & 2     & 2 & 40 & 0.01 & 10 & 5 \\
    \hline
    Energy4 & 4     & 2      & 2 & 40   & 0.01 & 10 & 5 \\
    \hline
    Energy6 & 6     & 2      & 2 & 40  & 0.01  & 10 & 5\\
    \hline
    CNOT5 & 2     & 2      & 5 & 100  & 0.01 & 10 & 20 \\
    \hline
    CNOT10 & 2     & 2      & 10 & 200 & 0.001 & 10 & 20 \\
    \hline
    CNOT15 & 2     & 2      & 15 & 300 & 0.0001 & 10 & 20 \\
    \hline
    CNOT20 & 2     & 2      & 20 & 400  & 0.0001 & 10 & 20 \\
    \hline
    NOT2 & 1 & 2 & 2 & 20 & 0.001 & 5 & 4\\
    \hline
    NOT6 & 1 & 2 & 6 & 60 & 0.001 & 5 & 12\\
    \hline
    NOT10 & 1 & 2 & 10 & 100 & 0.001 & 5 & 20\\
    \hline
    CircuitH2 & 2     & 5      & 4 & 80  & 0.001 & 10 & 8\\
    \hline
    CircuitLiH & 4     & 12     & 20 & 200 & 0.001 & 5 & 40\\
    \hline
    \end{tabular}%
  \label{tab: setting}%
\end{table}%

In Sections~\ref{sec: res-c}--\ref{sec: res-i}, for brevity we select four different instances,  Energy6, CNOT20, NOT10, and CircuitLiH, to analyze the results of objective values and CPU time. Detailed results of all the methods and instances are presented in Tables~\ref{tab: res-obj-c}, \ref{tab: res-time-c}, \ref{tab: res-obj-r}, \ref{tab: res-time-r}, \ref{tab: res-obj-improve}, and \ref{tab: res-time-improvement} in Appendix~\ref{app: results}. Our full code and results are available on our Github repository \citep{codebinary2022}.

\subsection{Results of Continuous Relaxations}
\label{sec: res-c}
In Figure~\ref{fig:res-circuit-l2} we show how the common logarithm of the squared $L_2$-norm value varies with the penalty parameter represented by the lines. We also present an asymptotic bound $\log_{10} c/\rho$ represented by the dashed lines, where $c$ is a constant marking that a selected data point lies on the dashed lines in the simulation results. Notably, the figure %
confirms that $l(u^*_\rho, T) \sim O(1/\rho)$, as shown in Theorem~\ref{theo:l2-bound}. 
\begin{figure}[ht!]
    \centering
    \subfloat[Instance CircuitH2]{\includegraphics[width=0.5\textwidth]{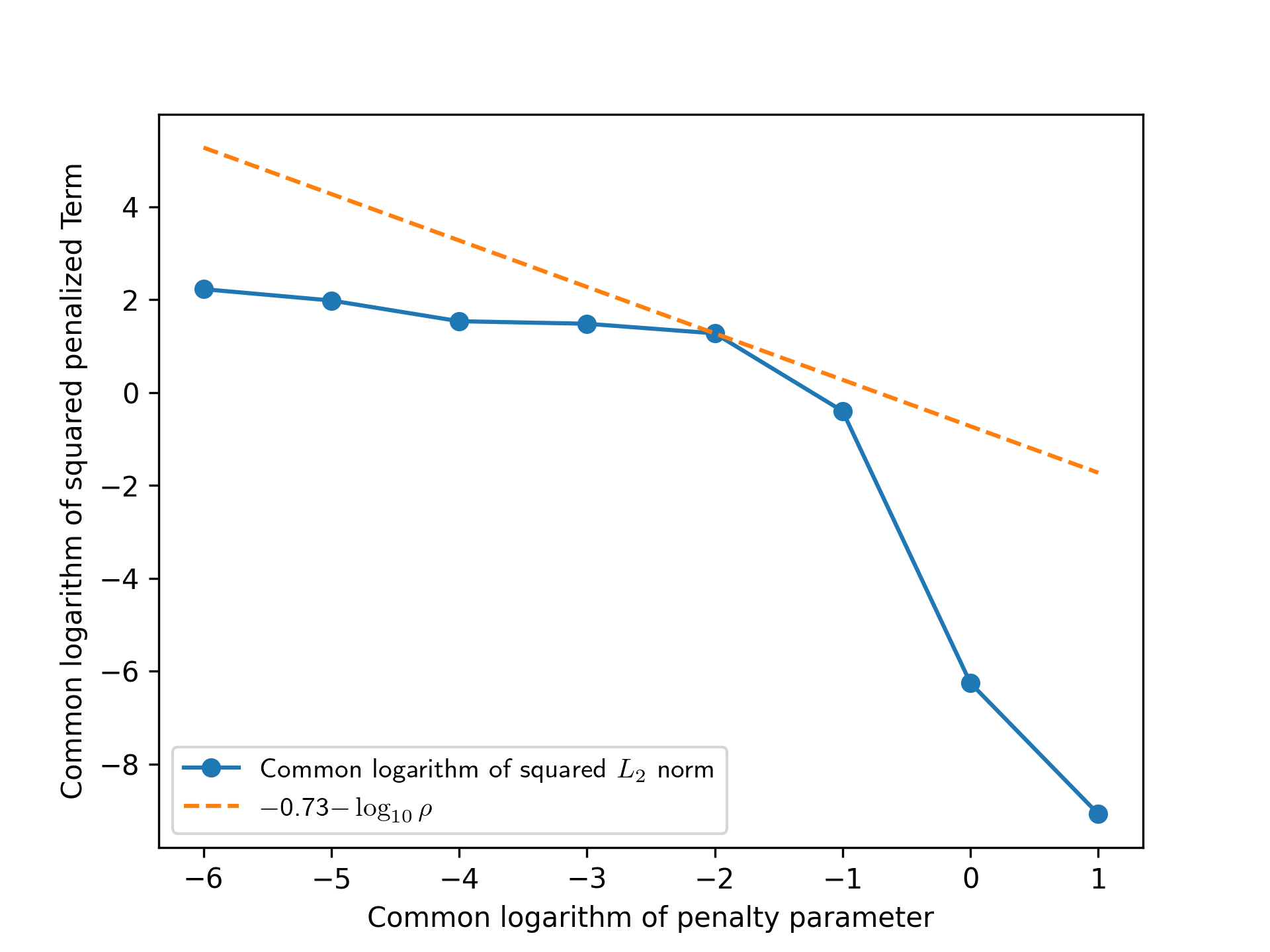}}
    \subfloat[Instance CircuitLiH]{\includegraphics[width=0.5\textwidth]{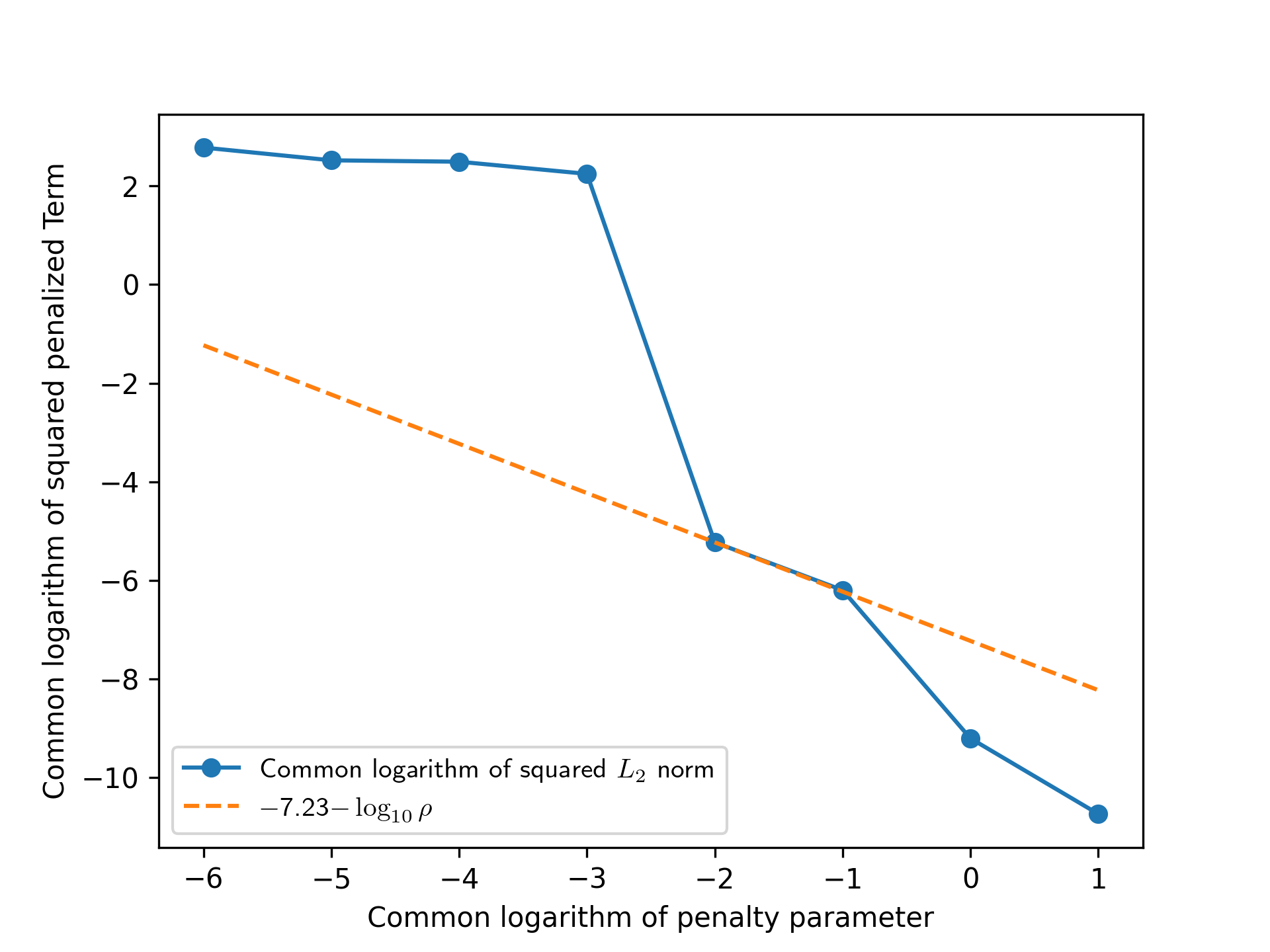}}
    \caption{Common logarithm of squared $L_2$-norm varying with $\log_{10} \rho$, where $\rho$ is the squared $L_2$ penalty parameter. Blue lines are the common logarithm of the squared $L_2$-norm. Orange dashed lines are functions $\log_{10} c-\log_{10} \rho$, where $c$ is a constant. The figures show that the squared $L_2$-norm $l(u^*_\rho, T)$ decreases as the penalty parameter $\rho$ increases with a rate $O(1/\rho)$.}
    \label{fig:res-circuit-l2}
\end{figure}

We show that pGRAPE always obtains the lowest objective value but the highest TV regularizer values because it solves the model without the TV regularizer. ADMM has the best performance for reducing the TV regularizer values, showing the benefits of introducing a TV regularizer in Section~\ref{sec: alg-infrequent}. The detailed results are presented in Appendix~\ref{app: results}.

We measure the problem size of each instance by $2^q\times N\times T$, where $q$, $N$, and $T$ represent the number of qubits, number of controllers, and number of time steps. For four groups of instances,  Energy, CNOT, NOT, and Circuit, we present in Figure~\ref{fig:c-time} common logarithm-logarithm figures to show how the CPU time and the number of iterations vary with the problem size. 
The CPU time increases significantly with an increase in the number of qubits because the dimension of simulation Hamiltonian matrices increases exponentially. The CPU time also increases when the number of controllers and time steps increases. 
Compared with pGRAPE, TR and ADMM spend more time solving the model because it contains the TV regularizer. 
The CPU time of ADMM is more than that of TR for Energy instances, while TR spends more time on CNOT, NOT, and Circuit instances than ADMM spends. 
\begin{figure}[ht!]
    \includegraphics[width=\textwidth]{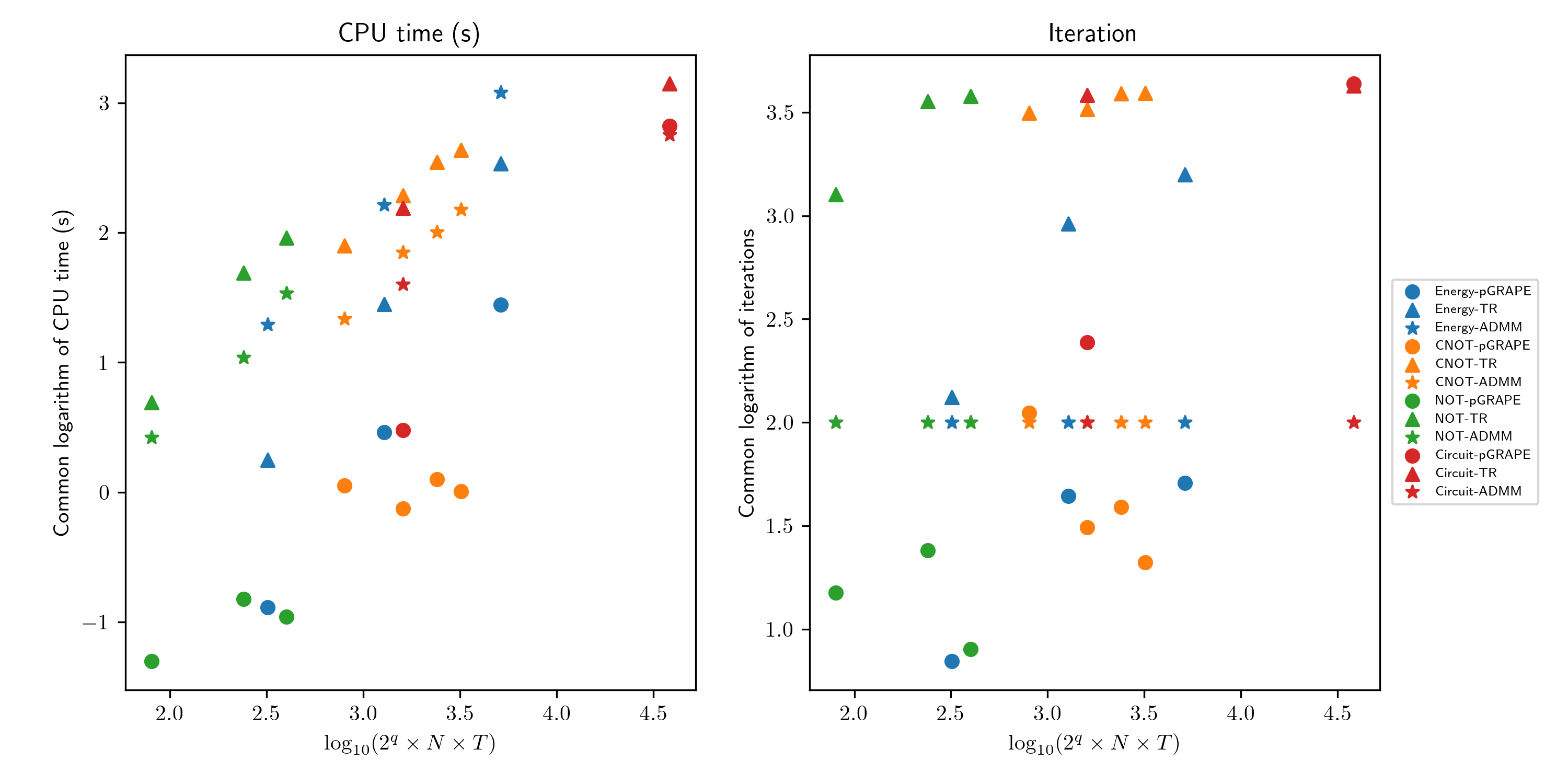}
    \caption{Common logarithm-logarithm figure of CPU time and iterations of continuous relaxation for all the instances and number of time steps. Blue, orange, green, and red colors represent instances Energy, CNOT, NOT, and Circuit. Circles, triangles, and stars represent obtaining continuous results by pGRAPE, TR, and ADMM, respectively. Dots with increasing size and transparency represent results by MT and MS, respectively.}
    \label{fig:c-time}
\end{figure}

\subsection{Results of Rounding Techniques}
\label{sec: res-r}
The results for Sections~\ref{sec: res-r} and~\ref{sec: res-i} are summarized in Figure~\ref{fig:all-instance}, which shows the objective values and TV regularizer values of the binary results obtained by SUR as well as CIA with min-up-time constraints and max-switching constraints of selected instances. Compared with pGRAPE+MT/MS, TR+MT/MS or ADMM+MT/MS obtains lower objective values and TV regularizer values for the binary results with min-up-time constraints and max-switching constraints because they solve the model with the TV regularizer. 

From the results of pGRAPE we show that the squared $L_2$-penalty function ensures that the difference between binary controls of SUR and continuous results is small. To further demonstrate the convergence of pGRAPE+SUR in Section~\ref{sec: alg-sur}, we present how the maximum absolute integral error, which is the left-hand side of equation \eqref{eq:sur-convergence},
varies with the number of time steps in Figure~\ref{fig:sur-error} represented by the lines. The figures show that the error converges to zero with the increase in the number of time steps, as is claimed in Proposition~\ref{prop: b-c-convergence}. 
We also present the theoretical upper bound of the integral error, which is the right-hand side of \eqref{eq:sur-convergence}, represented by the dashed lines. We observe that the integral error is always less than the upper bound, demonstrating the conclusion in Theorem~\ref{theo:sur-convergence}. 
In addition, we present the figures of how $\epsilon(\Delta t)$ and the upper bound $\sqrt{t_f{l(u^c,T)}\Delta t}$ vary with the number of time steps for instances CircuitH2 and CircuitLiH in Figure~\ref{fig:epsilon}. 
We show that $\epsilon(\Delta t)$ is always smaller than the upper bound, demonstrating the conclusion~\eqref{eq:sur-convergence-prop} in Corollary~\ref{cor:sur-convergence-1}.
\begin{figure}[ht!]
    \centering
    \subfloat[Instance CicruitH2 \label{fig:h2-sur-error}]{\includegraphics[width=0.5\textwidth]{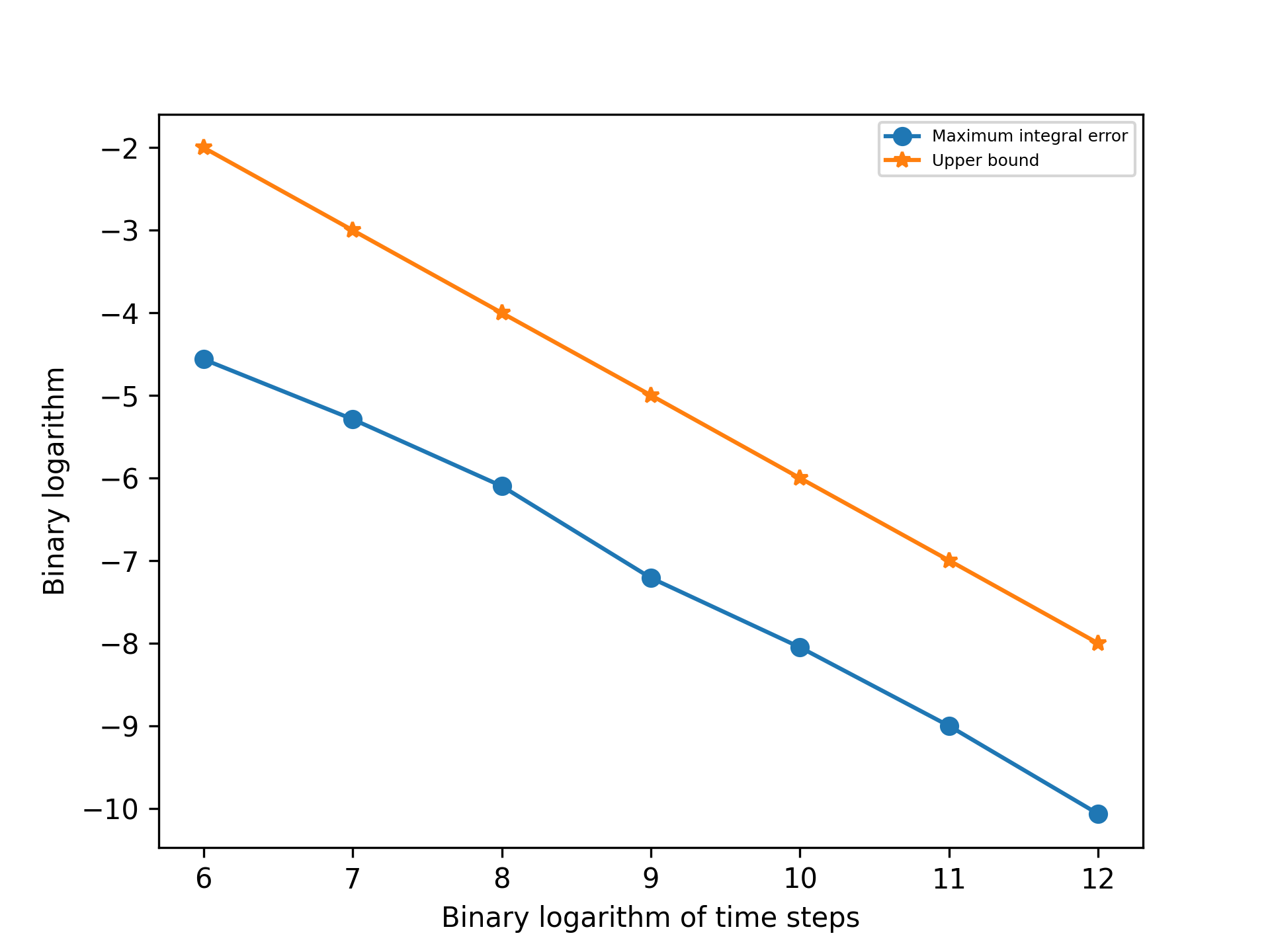}}
    \subfloat[Instance CircuitLiH \label{fig:lih-sur-error}]{\includegraphics[width=0.5\textwidth]{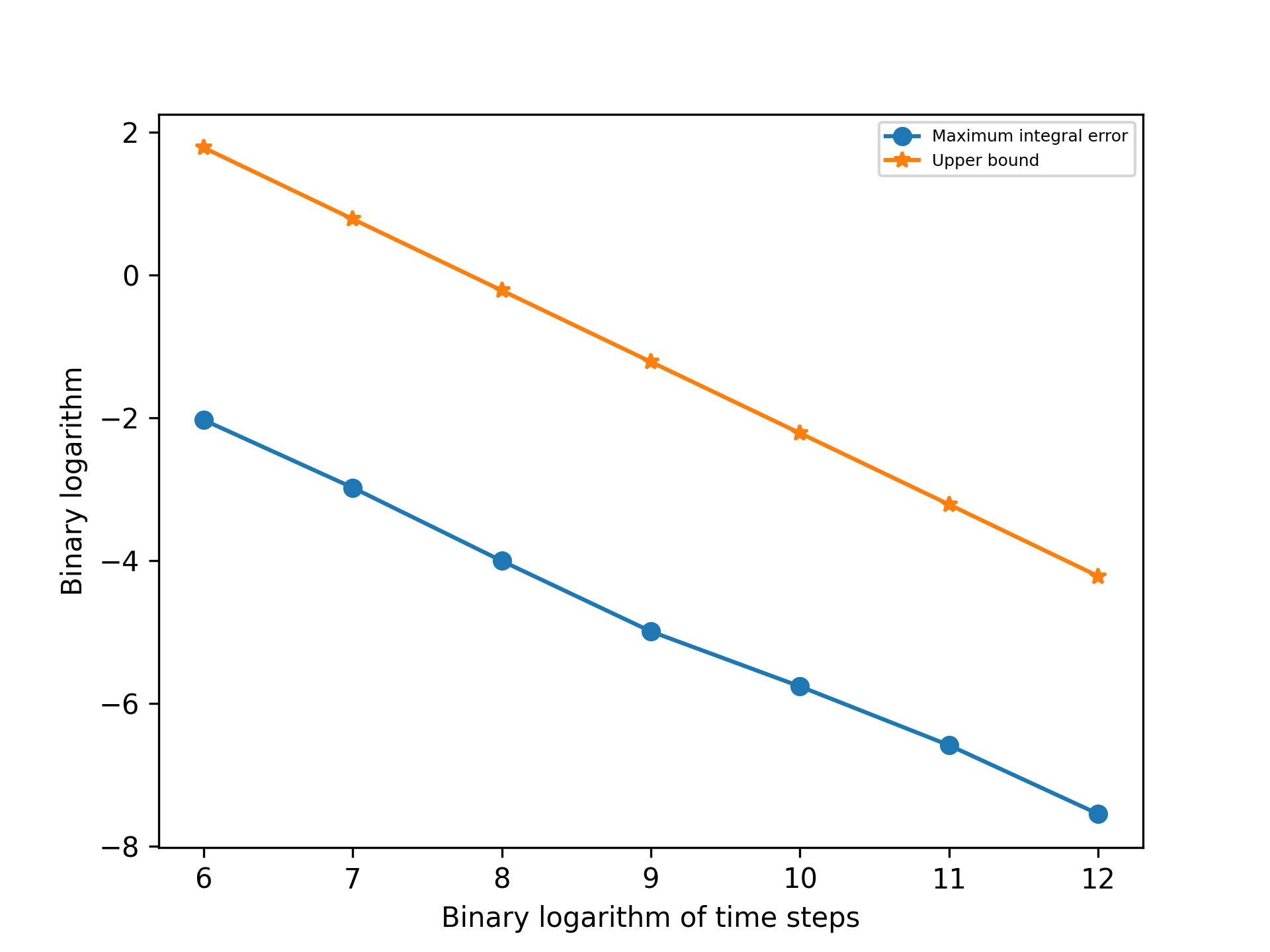}}
    \caption{Binary logarithm of maximum absolute integral error and upper bound of continuous results and SUR binary results. Blue lines marked by circles represent the maximum absolute integral error, and orange lines marked by stars represent the upper bound. 
    The error is always smaller than the upper bound and converges to zero as the number of time steps increases, demonstrating the conclusion~\eqref{eq:sur-convergence}. }
    \label{fig:sur-error}
\end{figure}
\begin{figure}[ht!]
    \centering
    \subfloat[Instance CicruitH2 \label{fig:h2-epsilon}]{\includegraphics[width=0.5\textwidth]{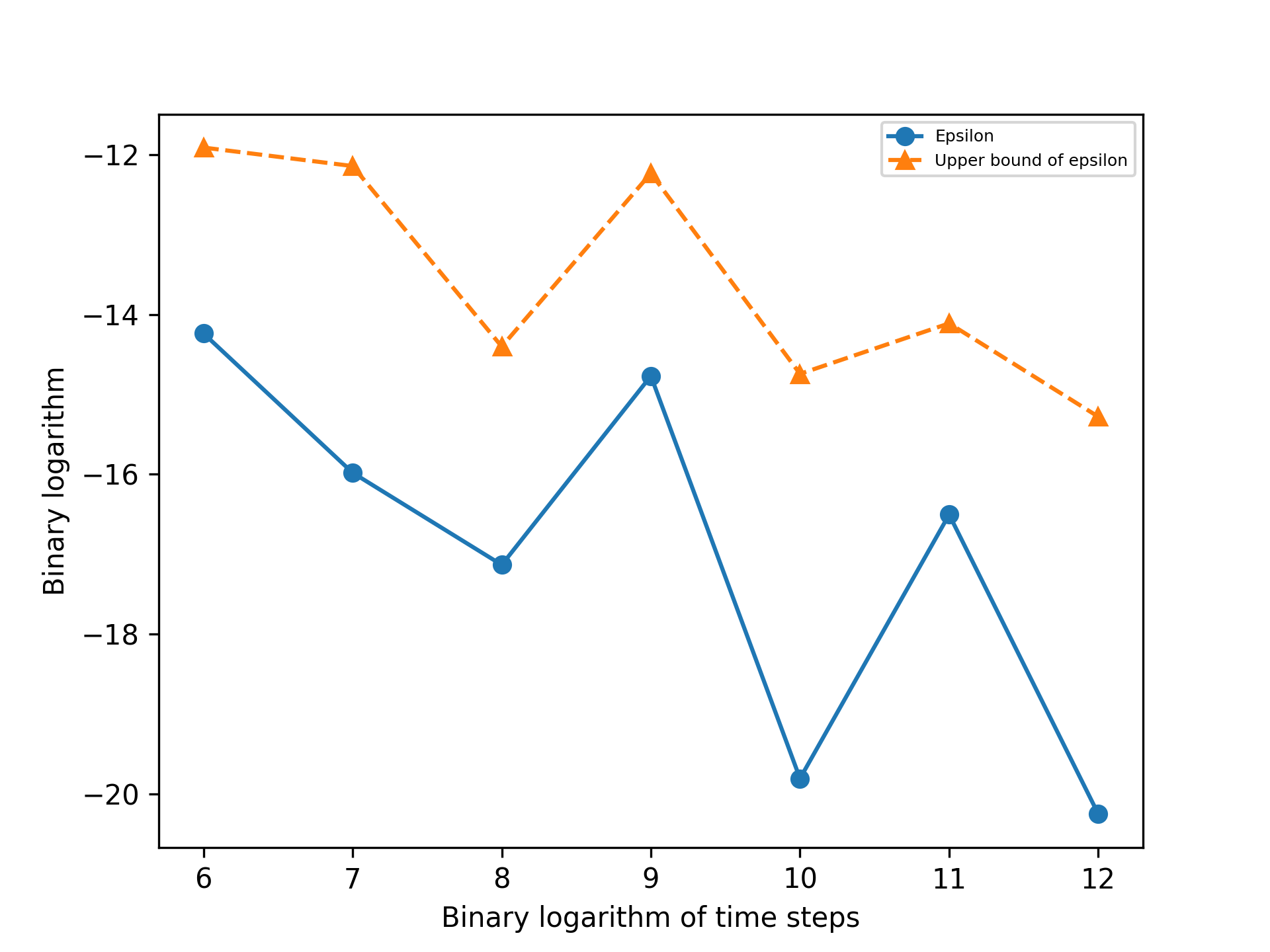}}
    \subfloat[Instance CircuitLiH \label{fig:lih-epsilon}]{\includegraphics[width=0.5\textwidth]{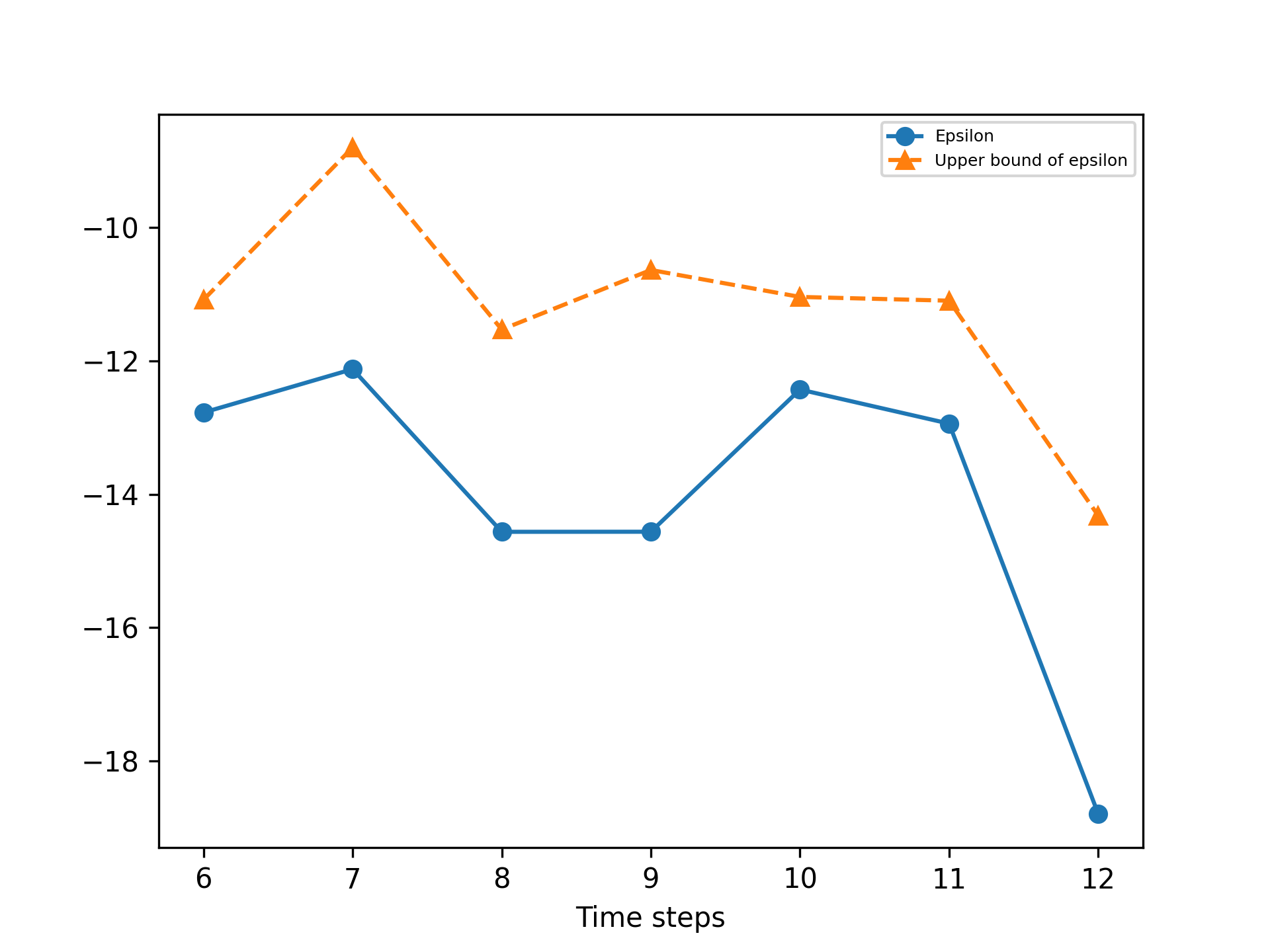}}
    \caption{Binary logarithm of the maximum SOS1 difference of continuous control $\epsilon(\Delta t)$ and its upper bound in \eqref{eq:sur-convergence-prop}. Blue lines represent $\epsilon(\Delta t)$ and orange dashed lines represent the upper bound. We show that $\epsilon(\Delta t)$ converges to $0$ with the increase in time steps and is always smaller than the upper bound, demonstrating Corollary \ref{cor:sur-convergence-1}.}
    \label{fig:epsilon}
\end{figure}

We set the time limit to $60$ seconds for solving the combinatorial integral approximation with min-up-time and max-switching constraints. In Figure~\ref{fig:r-time} we present the CPU time and iterations of all the instances. All the CPU times of SUR and Energy instances are less than 0.01 seconds so we do not show them in the figure. The CPU times increase significantly  with the number of variables. For the instance CircuitLiH, all the methods exceed the computational time limit. 
In most cases, iterations and CPU times of MS are more than MT because the problem with MS has a larger integrality gap.

\begin{figure}[ht!]
    \includegraphics[width=\textwidth]{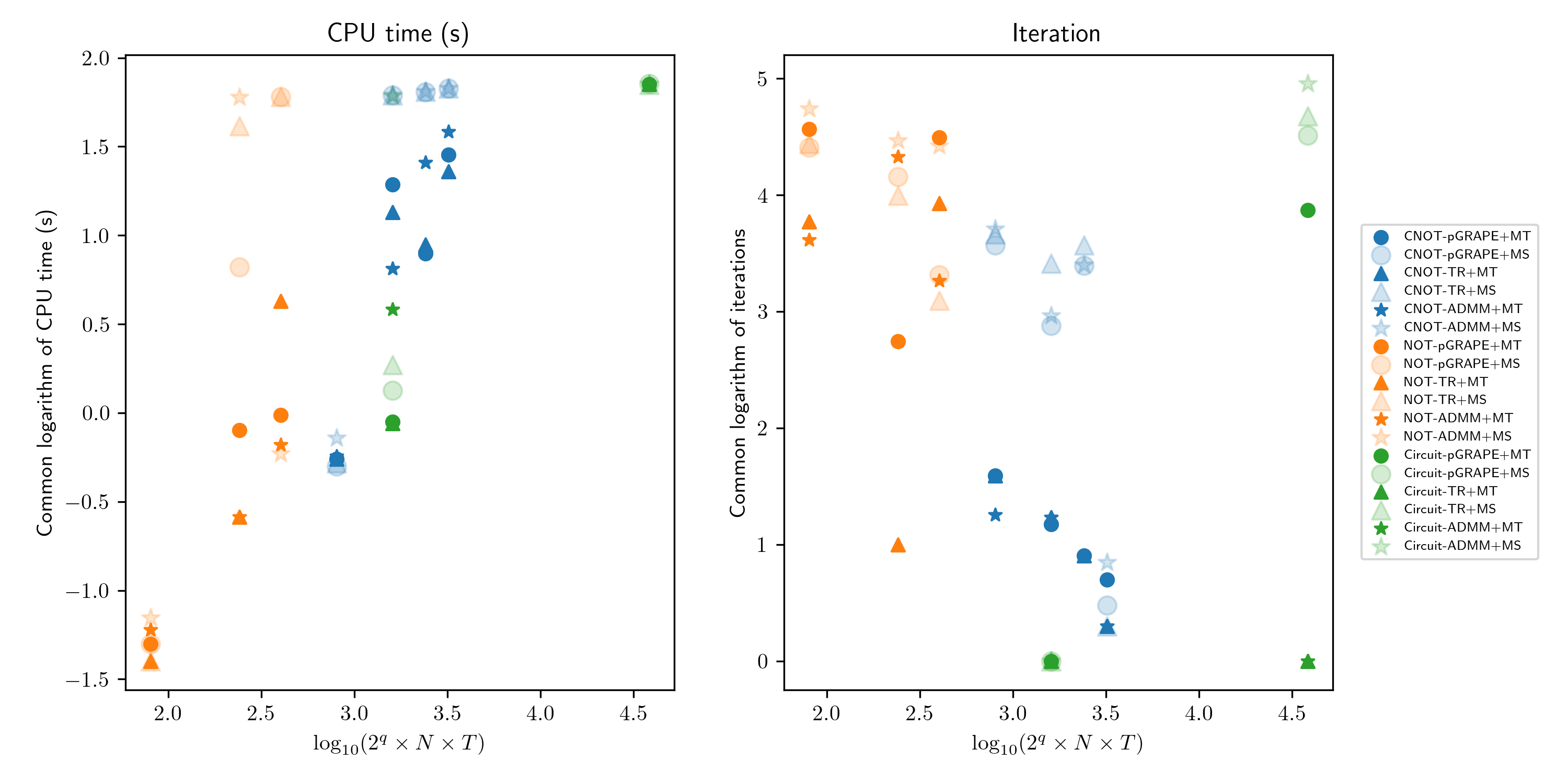}
    \caption{Common logarithm-logarithm figure of CPU time and iterations of CIA for all the instances. Blue, orange, and green colors represent instances of CNOT, NOT, and Circuit, respectively.  Circles, triangles, and stars represent obtaining continuous results by pGRAPE, TR, and ADMM, respectively. Dots with increasing size and transparency represent results by SUR, MT, and MS, respectively. }
    \label{fig:r-time}
\end{figure}

\subsection{Results of Improvement Heuristic}
\label{sec: res-i}
The objective values and TV regularizer values after the improvement heuristic with the control results of combinatorial integral approximation as initial points for selected instances are represented by opaque small dots in Figure~\ref{fig:all-instance}. 
Because there are multiple local minima for these quantum control problems, each method may obtain a different performance after the improvement heuristic. 
For systems with two controllers, all the methods obtain performance at a similar level. The performance of different methods varies more on the circuit compilation problem because it contains more local optima. 

We show in Figure~\ref{fig:h-time} the CPU time and iterations for the improvement heuristic of selected instances.
For each group of instances  Energy, CNOT, NOT, and Circuit, the CPU time increases with the problem size. 
The number of iterations highly depends on the initial points.

In addition, we observe that for both continuous relaxation and ALB, the CPU time is exponentially increasing but the number of iterations increases slightly. We note that the main increase of the computational time comes from the time-evolution process for simulating the quantum system because the dimension of the Hamiltonian matrices increases exponentially with the number of qubits $q$. In our algorithm, we could directly conduct the time evolution on quantum computers with a control pulse, and then only use the final state to compute the objective value and approximate the gradient by the finite difference method. Hence, we do not need to evaluate the state at each time step, indicating the scalability of our approach for handling large-scale instances.

\begin{figure}[ht!]
    \includegraphics[width=\textwidth]{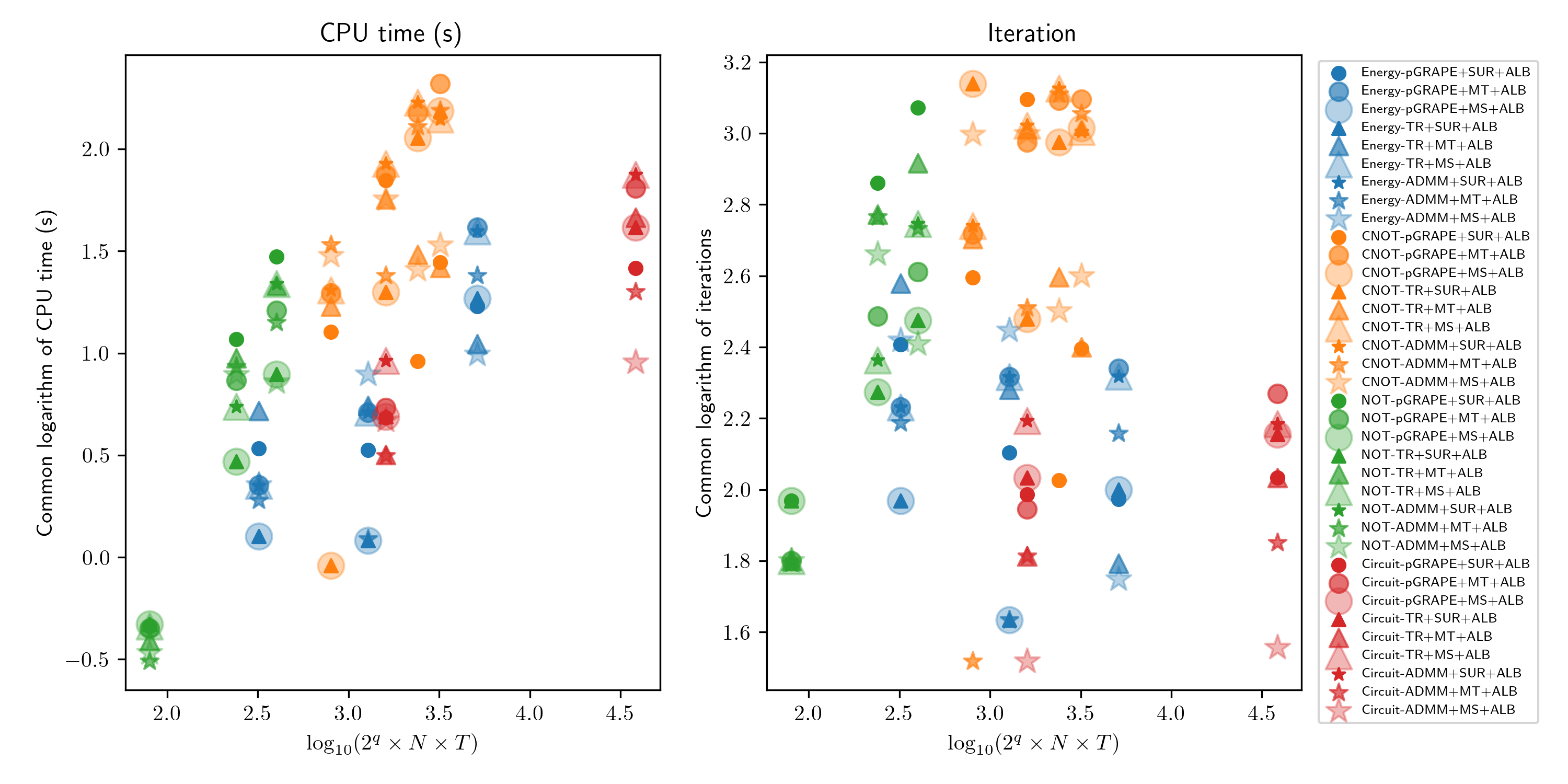}
    \caption{Common logarithm-logarithm figure of CPU time and iterations of ALB for all the instances. Blue, orange, green, and red colors represent groups of instances Energy, CNOT, NOT, and Circuit. Circles, triangles, and stars represent obtaining continuous results by pGRAPE, TR, and ADMM, respectively. Dots with increasing size and transparency represent results by SUR, MT, and MS, respectively. }
    \label{fig:h-time}
\end{figure}

In Figure~\ref{fig:all-instance} we show the objective values and TV regularizer values of the combinatorial integral approximation and the improvement heuristic. 
We show that the objective value of SUR keeps the same or slightly increases, demonstrating the theorems and propositions in Section~\ref{sec: alg-sur}. 
Adding min-up-time and max-switching constraints leads to an increase in objective values because of adding additional constraints stated in Section~\ref{sec: alg-bnb} to the feasible set. 
The increase from adding min-up-time constraints is higher because the restricted feasible region contains fewer feasible solutions. 
From the aspect of the chattering measured by the TV regularizer value, we show that SUR leads to more chattering compared with continuous solutions because of the rounding process. 
As discussed in Section~\ref{sec: alg-bnb}, imposing min-up-time and max-switching constraints reduces the chattering, and adding min-up-time constraints reduces it more significantly because the constraints are stricter.
The improvement heuristic reduces the objective values and chattering remarkably for the results obtained from all the methods with rounding techniques, especially pGRAPE, which shows the benefits of our improvement algorithms in Section~\ref{sec: alg-tr}. 
 In Figure~\ref{fig:all-instance}, the dot closest to the lower-left corner indicates that it obtains the best balance between improving objective values and reducing chattering. 
If the original objective function matters more, one can choose methods that consist of SUR and the improvement heuristic. If one focuses more on reducing the switches, the methods with min-up-time or max-switching constraints and the improvement heuristic are better.

\begin{figure}[ht!]
    \centering
    \includegraphics[width=\textwidth]{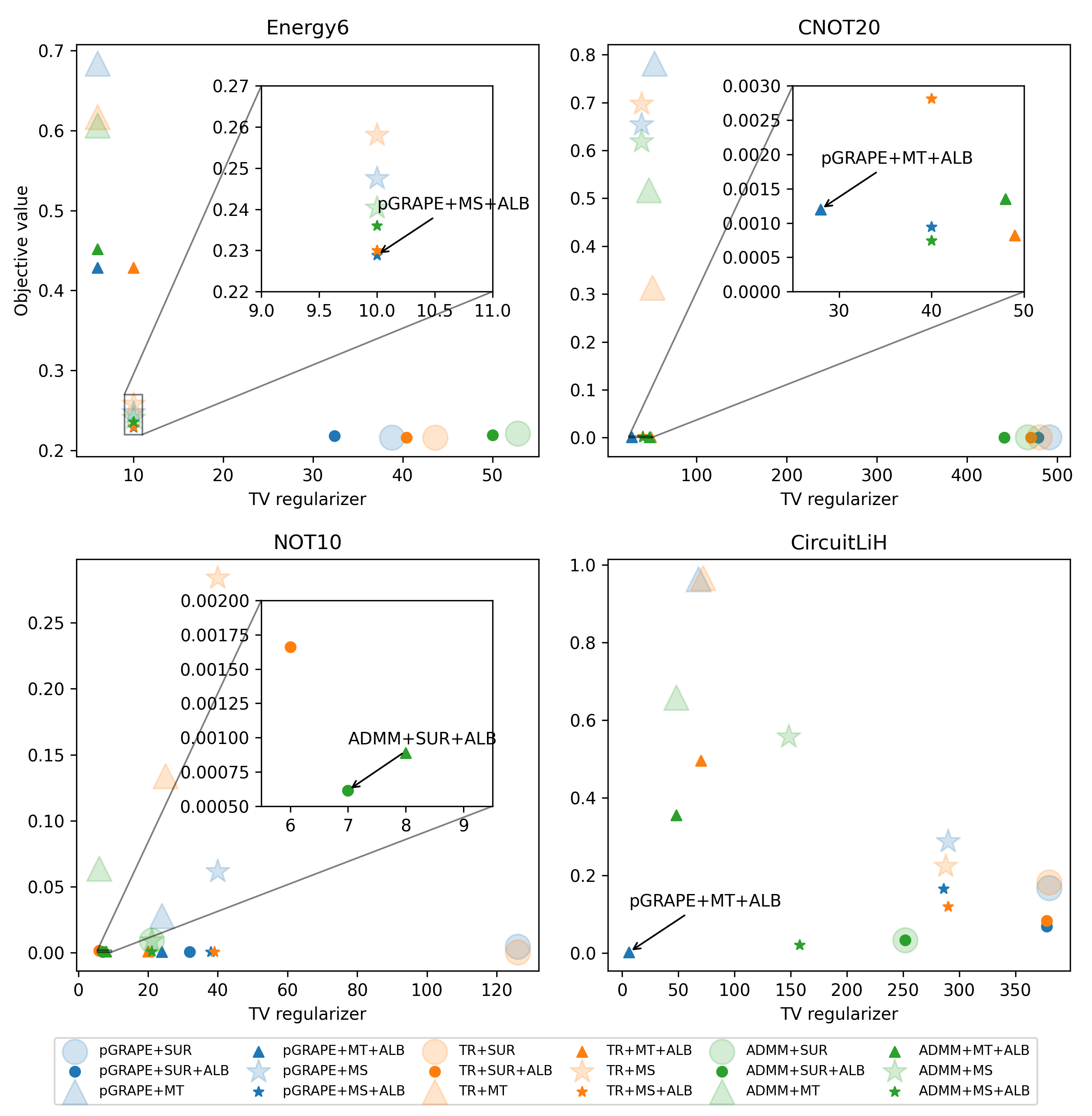}
    \caption{Objective values and TV regularizer values of binary results of selected instances. Blue, orange, and green dots represent solving continuous relaxation by pGRAPE, TR, and ADMM, respectively. Circles, triangles, and stars represent obtaining binary results by SUR, MT, and MS, respectively. 
    Big dots represent results before ALB, and opaque small dots represent results after ALB. We annotate the method represented by the closest point to the lower-left corner. }
    \label{fig:all-instance}
\end{figure}

\section{Conclusion}
\label{sec: conclusion}
In this paper we built a generic control model with both continuous and discretized formulations for the quantum control problem. We introduced a penalized squared $L_2$ function into the model to ensure that only one control is active at all times. In addition, we proposed a model with the TV regularizer aiming to reduce the absolute change of controls. 

We developed an algorithmic framework combining the GRAPE approach, combinatorial integral approximation, and local-branching improvement heuristic. With numerical simulations on multiple examples with various objective functions and controllers, we 
demonstrated the feasibility of the discrete optimal quantum control problem and illustrated that our algorithms obtained 
trade-off controls between high quality and low absolute changes within a short computational time. 
Specifically, the performance of different relaxation models varies among instances, and therefore, testing all three methods (pGRAPE, TR, and ADMM) is helpful to select the best control.
If one is more interested in obtaining controls with a lower energy or infidelity, we recommend methods using the sum-up-rounding technique. On the other hand, if one aims to reduce switches, it is better to choose methods with min-up-time constraints added. Both cases require the use of the approximate local branching algorithm to improve final solutions. 
In practice, the performance can be improved further by fine-tuning the corresponding parameters. 

All the numerical simulations in our paper were conducted on classical computers , and the CPU time mainly comes from the time-evolution process of quantum states. 
Because our algorithm only requires the final state to compute the objective value and approximate the gradient, running the time-evolution process in our algorithms on quantum computers is easy in practice and could help eliminate a significant amount of the exponential slow-down as we observe in our simulations.
Furthermore, with continuous control results, we can extract appropriate controller sequences and optimize the location of switching points to obtain high-quality controls. Switching-time optimization using continuous control results is future research.

\section*{Acknowledgment} 
This material is based upon work supported by the U.S. Department of Energy, Office of Science, under contract number DE-AC02-06CH11357.
This work was supported by the U.S.~Department of Energy, Office of Science, Office of Advanced Scientific Computing Research, Accelerated Research for Quantum Computing program.
The authors gratefully acknowledge the support from the Givens Associate program in Argonne National Laboratory. Xinyu Fei and Siqian Shen are thankful for the support from U.S.\ National Science Foundation (NSF) grants \#CMMI-1727618 and Department of Energy (DOE) grant \#DE-SC0018018. 

Certain commercial equipment, instruments, or materials are identified in this paper in order to specify the experimental procedure adequately.  Such identification is not intended to imply recommendation or endorsement by NIST, nor is it intended to imply that the materials or equipment identified are necessarily the best available for the purpose.

\bibliographystyle{unsrtnat}
\bibliography{Xinyu}

\begin{appendices}
\section{Proofs of Results in Section~\ref{sec: alg-binary}}
\label{app: proof}
In this section, we provide proofs of the results in Section~\ref{sec: alg-binary}. 
\subsection{Proofs of Results in Section~\ref{sec: alg-penalty}}
\label{app: proof-sos1}

\begin{proof}[Proof of Theorem~\ref{theo:l2-bound}]
Because $F$ is upper bounded by constant $C_F$, we have $|F(u)|\leq C_F$ for all the $u$. We prove that for any $\rho$ and $T$, $|\rho l(u^*_\rho, T)|\leq 2C_F$. 
Assume that there exists an optimal solution $u^{(1)}$ of \eqref{eq:model-d-l2-p} such that $|\rho l(u^{(1)}, T)|>2C_F$, then we have $F(u^{(1)})+\rho l(u^{(1)}, T) > C_F$. 
This is a contradiction to the definition of $u^{(1)}$ because we can find a feasible solution $u^{(2)}$ of \eqref{eq:model-d-1} such that $F(u^{(2)}) + \rho l_(u^{(2)}, T)=F(u^{(2)})\leq C_F$. 
Therefore we drive that $|l(u^*_\rho, T)|\leq 2C_F / \rho$, which means that $l(u^*_\rho, T)\sim O(1/\rho)$. 
\end{proof}

\subsection{Proofs of Results in Section~\ref{sec: alg-rounding}}
\label{app: proof-sur}
\begin{proof}[Proof of Proposition~\ref{prop: nosos1}]
Because the SOS1 property holds for binary control $u^b$, we have
\begin{align}
    \max_{k=1,\ldots,T} \left\|\sum_{\tau=1}^k (u^c_\tau - u^b_\tau) \Delta t\right\|_\infty &= \max_{k=1,\ldots,T} \max_{j=1,\ldots,N} \left|\sum_{\tau=1}^k (u^c_{j\tau} - u^b_{j\tau}) \Delta t\right| \nonumber \\
    &\geq \max_{k=1,\ldots,T} \frac{1}{N} \left|\sum_{j=1}^{N} \sum_{\tau=1}^k (u^c_{j\tau} - u^b_{j\tau}) \Delta t\right| \nonumber\\
    &= \max_{k=1,\ldots,T} \frac{1}{N} \left|\sum_{\tau=1}^k \left (\sum_{j=1}^{N} u^c_{j\tau} - 1\right ) \Delta t\right| = \frac{1}{N}\epsilon(\Delta t).
\end{align}
\end{proof}
\begin{proof}[Proof of Theorem~\ref{theo:sur-convergence}]
For simplicity, we use $\epsilon$ to denote $\epsilon(\Delta t)$. It is obvious that for any time step $k=1,\ldots,T$, exactly one control can be set to $1$ by $u^b$ by the construction. Hence $u^b$ satisfies the SOS1 property. 
For a specific time step $k=1,\ldots,T$, define 
\begin{align}
    i_k=\argmax_{j=1,\ldots,N}  \left |\sum_{\tau=1}^{k} u^c_{j\tau} \Delta t - \sum_{\tau=1}^{k} u^b_{j\tau}\Delta t\right |.
\end{align}
We assume that there exists a time step $r$ such that the claim does not hold. For simplicity, we use $i$ to represent the corresponding $i_r$. We consider the assumption in two cases. 
\paragraph{Case 1:} We assume that, 
\begin{align}
\label{eq:sur-a1}
    \sum_{\tau=1}^r u^c_{i\tau} \Delta t-\sum_{\tau=1}^{r} u^b_{i\tau} \Delta t<-\left(N-1\right) \Delta t - \frac{2N - 1}{N} \epsilon.
\end{align}
Let $\hat{k}$ be the highest index of the time step in which control $i$ is rounded up, 
\begin{align}
    \hat{k}=\argmax _{1 \leq k \leq r}\left\{k\ |\  u^b_{ik}=1\right\}.
\end{align}
Then by \eqref{eq:sur-a1}, we have, 
\begin{subequations}
\label{eq:sur-a1-proof-1g}
\begin{align}
    \label{eq:sur-a1-proof-1}
    \sum_{\tau=1}^{\hat{k}} u^c_{i\tau} \Delta t + \left(N-1\right) \Delta t + \frac{2N - 1}{N} \epsilon 
    \leq &\sum_{\tau=1}^r u^c_{i\tau} \Delta t + \left(N-1\right) \Delta t + \frac{2N - 1}{N} \epsilon\\
    < & \sum_{\tau=1}^{r} u^b_{i\tau} \Delta t = \sum_{\tau=1}^{\hat{k}} u^b_{i\tau} \Delta t.
\end{align}
\end{subequations}
Because $u^b_{i\hat{k}}=1$, we know that $i$ has the maximum value of \eqref{eq:sur-a1-proof-1} among $j=1,\ldots,N$. Hence it follows from~\eqref{eq:sur-a1-proof-1g}, 
\begin{align}
    \sum_{\tau=1}^{\hat{k}} u^c_{j\tau} \Delta t-\sum_{\tau=1}^{\hat{k}} u^b_{j\tau} \Delta t<-\left(N-1\right) \Delta t - \frac{2N - 1}{N} \epsilon,\ j=1,\ldots,N.
\end{align}
Summing up over all the controls $j$,
\begin{align}
    - \epsilon \leq \sum_{j=1}^{N} \left(\sum_{\tau=1}^{\hat{k}} u^c_{j\tau}\Delta t -\sum_{\tau=1}^{\hat{k}} u^b_{j\tau} \Delta t\right)<-N\left(N-1\right) \Delta t - (2N - 2)\epsilon.
\end{align}
The first inequality comes from the definition of $\epsilon$. This leads to $0<-N(N - 1)\Delta t - (N - 1)\epsilon$, which is a contradiction because $\epsilon\geq 0$ and $N \geq 2$. 
\paragraph{Case 2:} We assume that, 
\begin{align}
\label{eq:sur-a2}
    \sum_{\tau=1}^{r} u^c_{i\tau} \Delta t-\sum_{\tau=1}^{r} u^b_{i\tau} \Delta t>\left(N-1\right) \Delta t + \frac{2N - 1}{N} \epsilon.
\end{align}
From the definition of $\epsilon$, we have
\begin{align}
    \label{eq:sur-proof-a2-1}
    \sum_{1=j \neq i}^{N}\left(\sum_{\tau=1}^{r} u^c_{j\tau} \Delta t-\sum_{\tau=1}^{r} u^b_{j\tau} \Delta t\right)+\sum_{\tau=1}^{r} u^c_{i\tau} \Delta t-\sum_{\tau=1}^{r} u^b_{i\tau} \Delta t\leq \epsilon
\end{align}
and by substituting \eqref{eq:sur-a2} into~\eqref{eq:sur-proof-a2-1}, it holds that 
\begin{align}
    \sum_{1=j \neq i}^{N}\left(\sum_{\tau=1}^{r} u^c_{j\tau} \Delta t-\sum_{\tau=1}^{r} u^b_{j\tau} \Delta t\right)+(N - 1)\Delta t + \frac{N - 1}{N} \epsilon < 0.
\end{align}
Because the left-hand side can be written as the sum of $N - 1$ terms as
\begin{align}
    \Delta t + \frac{1}{N}\epsilon + \sum_{\tau=1}^{r} u^c_{j\tau} \Delta t-\sum_{\tau=1}^{r} u^b_{j\tau} \Delta t, 
\end{align}
at least one of them must be negative, therefore there exists $\hat{j}$ such that 
\begin{align}
    \label{eq:sur-proof-a2-2}
    \Delta t + \frac{1}{N}\epsilon + \sum_{\tau=1}^{r} u^c_{\hat{j}\tau} \Delta t-\sum_{\tau=1}^{r} u^b_{\hat{j}\tau} \Delta t < 0.
\end{align}
Let $\hat{k}$ be the highest index of the time step in which control $\hat{j}$ is rounded up, 
\begin{align}
    \hat{k}=\argmax _{1 \leq k \leq r}\left\{k\ |\ u^b_{\hat{j}k}=1\right\}.
\end{align}
Then we have
\begin{align}
    \sum_{\tau=1}^{\hat{k}} u^c_{\hat{j}\tau} \Delta t-\sum_{\tau=1}^{\hat{k}-1} u^b_{\hat{j}\tau} \Delta t \leq \Delta t+\sum_{\tau=1}^{r} u^c_{\hat{j}\tau}\Delta t -\sum_{\tau=1}^{r} u^b_{\hat{j}\tau} \Delta t<-\frac{1}{N}\epsilon.
\end{align}
The first inequality follows from $u_{\hat{j}\hat{k}}^b=1$ and $u_{\hat{j}k}^b=0,\ k\geq \hat{k}$. The second inequality follows from~\eqref{eq:sur-proof-a2-2}. 
Because $\hat{j}$ is the control which is rounded up at time step $\hat{k}$, for any $j=1,\ldots,N$, it holds that, 
\begin{align}
    \sum_{\tau=1}^{\hat{k}} u^c_{j\tau} \Delta t-\sum_{\tau=1}^{\hat{k}} u^b_{j\tau} \Delta t \leq \sum_{\tau=1}^{\hat{k}} u^c_{j\tau} \Delta t -\sum_{\tau=1}^{\hat{k}-1} u^b_{j\tau} \Delta t<-\frac{1}{N}\epsilon.
\end{align}
Summing over all the controls, we have 
\begin{align}
    \sum_{j=1}^{N} \sum_{\tau=1}^{\hat{k}}  \left(u^c_{j\tau} - u^b_{j\tau}\right)\Delta t <-\epsilon.
\end{align}
This contradicts the definition of the parameter $\epsilon$.
\end{proof}

\begin{proof}[Proof of Corollary~\ref{cor:sur-convergence-1}]
Based on the formulation of the discretized control in the continuous relaxation of the model \eqref{eq:model-d-l2-p}, we have
\begin{align}
    \epsilon(\Delta t) & = 
    \max_{k=1,\ldots,T} \left|\sum_{\tau=1}^k \left(\sum_{j=1}^N u^c_{j\tau} - 1\right) \Delta t\right|
     \leq \sum_{k=1}^T \left|\sum_{j=1}^N u^c_{jk} - 1\right|\Delta t\nonumber\\
    & \leq \sqrt{T} \sqrt{\sum_{k=1}^T \left(\sum_{j=1}^N u^c_{jk} - 1\right)^2}\Delta t 
    =\sqrt{Tl(u^c, T)} \Delta t = \sqrt{t_f{l(u^c,T)}\Delta t}.
\end{align}
From Theorem~\ref{theo:sur-convergence}, we directly obtain the statement \eqref{eq:sur-convergence-prop}. From Theorem~\ref{theo:l2-bound}, if the original function $F$ is bounded, the optimized squared $L_2$ term $l(u^c, T)$ is uniformly bounded over time steps $T$, then the absolute integral error between continuous and discretized control converges with $O(\sqrt{\Delta t})$.
\end{proof}
\begin{proof}[Proof of Proposition~\ref{prop: b-c-convergence}]
For any time step $k=1,\ldots,T$, we have 
\begin{align}
\label{eq:state-c}
    \lim_{\Delta t \rightarrow 0} X^c_k 
    & = \lim_{\Delta t \rightarrow 0}\prod_{\tau = 1}^k e^{-i\Delta tH_\tau } X_0
     = \lim_{\Delta t \rightarrow 0}\prod_{\tau = 1}^k e^{-i\Delta t\left(H^{(0)}+\sum_{j=1}^N u^c_{j\tau} H^{(j)}\right) }X_0\nonumber \\
     & = \lim_{\Delta t\rightarrow 0} e^{-i\Delta tH^{(0)} -i\Delta t\sum_{j=1}^N \sum_{\tau = 1}^k u^c_{j\tau}  H^{(j)}} X_0.
\end{align}
The last equality follows by Trotter expansion that $\displaystyle e^{\Delta t(A+B)}= e^{\Delta t A}e^{\Delta t B} + O(\Delta t^2)$ and taking the limit as $\Delta t$ goes to zero. Similarly, for binary control, we have 
\begin{align}
\label{eq:state-b}
    \lim_{\Delta t \rightarrow 0} X^b_k 
     = \lim_{\Delta t\rightarrow 0} e^{-i\Delta t H^{(0)} -i\Delta t \sum_{j=1}^N \sum_{\tau = 1}^k u^b_{j\tau}  H^{(j)}}X_0.
\end{align}
From Theorem \ref{theo:sur-convergence}, we have for any time step $k=1,\ldots,T$ and controller $j=1,\ldots,N$, 
\begin{align}
    \left| \sum_{\tau = 1}^k u^c_{j\tau}\Delta t - \sum_{\tau = 1}^k u^b_{j\tau} \Delta t\right| \leq \left\|\sum_{\tau=1}^k \left( u^c_{\tau} - u^b_{\tau}\right)\Delta t \right\|_\infty \leq \left(N-1\right) \Delta t + \frac{2N - 1}{N} \epsilon (\Delta t).
\end{align}
Combining with Corollary \ref{cor:sur-convergence-1} and taking the limit as $\Delta t$ goes to zero, we have 
\begin{align}
    \label{eq:int-error-quantum}
    \lim_{\Delta t\rightarrow 0} \sum_{\tau = 1}^k u^c_{j\tau}\Delta t = \lim_{\Delta t\rightarrow 0} \sum_{\tau = 1}^k u^b_{j\tau} \Delta t.
\end{align}
We substitute \eqref{eq:int-error-quantum} into the formulation of states \eqref{eq:state-c} and \eqref{eq:state-b}, then we obtain the conclusion that
\begin{align}
    \lim_{\Delta t\rightarrow 0} X^b_k = \lim_{\Delta t \rightarrow 0} X^c_k,\ k=1,\ldots,T.
\end{align}
Substituting states into the objective function, we prove that 
\begin{align}
    \lim_{\Delta t\rightarrow 0} F(X_T^b) = \lim_{\Delta t\rightarrow 0} F(X_T^c).
\end{align}
\end{proof}

\section{Detailed Numerical Results}
\label{app: results}
In this section, we present the results of continuous relaxation, combinatorial integral approximation, and the improvement heuristic for all the methods and instances in Table~\ref{tab: res-obj-c}--\ref{tab: res-time-improvement}. We also present the objective values and TV regularizer values for all the instances and annotate the best method in Figure~\ref{fig:obj-all-instance}. 

\begin{table}[H]
\centering
\caption{Objective value results of continuous relaxation. }
\begin{tabular}{l|rrr|rrrr}
       \hline
      \multirow{2}{*}{Instance} & \multicolumn{3}{c|}{Objective value} &
       \multicolumn{3}{c}{TV regularizer}\\
       \cline{2-7}
       & pGRAPE & TR & ADMM 
    & pGRAPE & TR & ADMM \\
       \hline
    Energy2 & 1.10E$-12$ & 9.11E$-05$ & 8.94E$-05$ & 0.999 & 0.567 & 0.523 \\
       \hline
    Energy4 & 0.154 & 0.155 & 0.168 & 5.064 & 4.114 & 2.752 \\
       \hline
    Energy6 & $0.213$ &	$0.213$ & $0.219$ & 5.999 & 4.508 & 3.237 \\
       \hline
       CNOT5 & 0.169 & 0.124 & 0.191 & 16.000 & 9.419 & 6.094 \\
       \hline
       CNOT10 & 1.16E$-09$ & 2.65E$-04$ & 3.21E$-04$ & 20.979 & 15.194 & 11.056 \\
       \hline
       CNOT15 & 1.00E$-10$ & 8.43E$-06$ & 3.65E$-06$ & 29.846 & 24.348 & 16.795 \\
       \hline
       CNOT20 & 5.93E$-10$ & 4.06E$-06$ & 8.07E$-07$ & 26.162 & 23.481 & 15.099 \\
                 \hline
    NOT2 & 0.163 & 0.163 & 0.163 & 2.226&	1.188&	0.908\\
    \hline
    NOT6 & 4.28E$-$10&	6.55E$-$05&	1.07E$-$04 & 3.251 &	0.652	&0.578\\
    \hline
    NOT10 & 6.55E$-$11 &	5.14E$-$05 &	4.79E$-$05 & 2.740 &	1.530 &	1.005\\
      \hline
        CircuitH$_2$ & 4.37E$-07$ & 1.24E$-04$ & 1.33E$-05$ & 18.720 & 8.421 & 2.744\\ 
       \hline
        CircuitLiH & 1.14E$-03$ & 1.45E$-03$ & 1.43E$-03$ & 53.976 & 48.720 & 0.677 \\
       \hline
\end{tabular}
\label{tab: res-obj-c}
\end{table}

\begin{table}[ht!]
\centering
\caption{CPU time and iterations of continuous relaxation. }
\begin{tabular}{l|rrr|rrr}
       \hline
       \multirow{2}{*}{Instance} & \multicolumn{3}{c|}{CPU time (s)} & \multicolumn{3}{c}{Iterations} \\
       \cline{2-7}
       & pGRAPE & TR & ADMM & pGRAPE & TR & ADMM \\
       \hline
      Energy2 & 0.13 & 1.77 & 19.56 & 7 & 132 & 100\\
      \hline
      Energy4 & 2.89 & 27.99 & 163.08 & 44 & 914 & 100\\
      \hline
      Energy6 & 27.94 & 341.43 & 1195.41 & 51 & 1578 & 100\\ %
      \hline
      CNOT5 & 1.12 & 79.55 & 21.75 & 111 & 3146 & 100\\
      \hline
      CNOT10 & 0.75 & 192.78 & 70.21 & 31 & 3279 & 100\\
      \hline
      CNOT15 & 1.26 & 348.96 & 100.59 & 39 & 3900 & 100 \\
      \hline
      CNOT20 & 1.02 & 432.73 & 150.25 & 21 & 3929 & 100\\
          \hline
    NOT2 & 0.05 & 4.93 & 2.63 & 15 & 1268 & 100\\
    \hline
    NOT6 & 0.15 & 48.81 & 10.93 & 24 & 3563 & 100\\
    \hline
    NOT10 & 0.10 & 91.39 & 34.05 & 6 & 3762 & 100\\
      \hline
      CircuitH$_2$ & 3.01 & 154.05 & 39.86 & 244 & 3832 & 100\\
      \hline
      CircuitLiH & 663.08 & 1403.03 & 564.95 & 4345 & 4247 & 100 \\
      \hline
\end{tabular}
\label{tab: res-time-c}
\end{table}

\begin{sidewaystable}[htbp]
\centering
\captionof{table}{Objective value results of combinatorial integral approximation. \label{tab: res-obj-r}}
\begin{adjustbox}{width=\linewidth}
\begin{tabular}{l|rrr|rrr|rrrrrrrrrrrrrrrrr}
       \hline
       \multirow{2}{*}{Instance} & \multicolumn{9}{c}{Objective value/TV regularizer} \\
       \cline{2-10}
       & pGRAPE+SUR & TR+SUR & ADMM+SUR & pGRAPE+MT & TR+MT & ADMM+MT & pGRAPE+MS & TR+MS & ADMM+MS \\
       \hline
        Energy2 & 4.22E$-04$/54 & 4.91E$-03$/54 & 4.01E$-04$/48 
       & $0.159$/4 & $0.159$/6 & $0.154$/6 & $0.029$/10 & $0.040$/10 & $0.028$/10 \\
       \hline
        Energy4 & $0.158$/26.8 & $0.158$/32.4 & $0.170$/44.8 & $0.367$/6 & $0.317$/6 & $0.363$/6 & $0.163$/10 & $0.162$/10 & $0.195$/10\\
       \hline
    Energy6 & $0.216$/38.8 & $0.216$/43.6 & $0.221$/52.8 & $0.684$/6 & $0.616$/6 & $0.606$/6 & $0.247$/10 & $0.258$/10 & $0.240$/10\\
       \hline
       CNOT5 & 0.170/16 & 0.332/24 & 0.190/41 & 0.243/10 & 0.525/6 & 0.285/7 & 0.170/16 & 0.593/15 & 0.191/32 \\
       \hline
       CNOT10 & 6.01E$-04$/116 & 1.78E$-03$/116 & 1.68E$-03$/86 & 0.158/22 & 0.323/21 & 0.084/15 & 0.011/39 & 0.019/38 & 0.006/32 \\
       \hline
       CNOT15 & 1.12E$-03$/266 & 2.30E$-03$/276 & 2.90E$-03$/279 & 0.539/37 & 0.290/36 & 0.176/27 & 0.325/38 & 0.284/40 & 0.214/39 \\
       \hline
       CNOT20 & 1.45E$-03$/491 & 8.30E$-04$/480 & 1.46E$-03$/467 & 0.782/53 & 0.314/51 & 0.517/47 & 0.654/39 & 0.697/39 & 0.619/39\\
       \hline
        NOT2 & 0.164/3 &	0.164/1 &	0.164/1&	0.164/1&	0.164/1&	0.164/1&	0.164/3&	0.164/1&	0.164/1\\
        \hline
        NOT6 & 2.38E$-$03/40 &  3.10E$-$03/46 & 	1.55E$-$03/52 & 	4.05E$-$02/8 & 	3.65E$-$02/9 & 	7.56E$-$02/10 & 	1.38E$-$02/19 & 	1.58E$-$01/14 & 	1.54E$-$01/14\\
        \hline
        NOT10 & 4.73E$-$03/126&	3.91E$-$04/126&	9.22E$-$03/21&	2.78E$-$02/24&	1.34E$-$01/25&	6.38E$-$02/6&	6.16E$-$02/40&	2.84E$-$01/40& 9.22E$-$03/21\\
        \hline
       CircuitH$_2$ & 0.027/32 & 0.027/36 & 0.006/76 & 0.600/8 & 0.591/8 & 0.063/8 & 0.026/22 & 0.038/24 & 0.008/22\\
       \hline
       CircuitLiH & 0.168/380 & 0.182/380 & 0.033/252 & 0.963/68 & 0.966/72 & 0.658/48 & 0.287/290 & 0.224/288 & 0.557/148\\
       \hline
\end{tabular}
\end{adjustbox}
\vspace{0.15in}
\centering
\captionof{table}{CPU time and iterations of combinatorial integral approximation. \label{tab: res-time-r}}
\begin{tabular}{l|rrr|rrrrrrrrrrrrrr}
       \hline
       \multirow{2}{*}{Instance} & \multicolumn{6}{c}{CPU time (s)/Iterations}\\
       \cline{2-7}
       & pGRAPE+MT & TR+MT & ADMM+MT 
       & pGRAPE+MS & TR+MS & ADMM+MS \\
      \hline
      Energy2 & $<0.01$/39 & $<0.01$/39 & $<0.01$/18 & $<0.01$/3715 & $<0.01$/4627 & $<0.01$/5106 \\
      \hline
      Energy4 & $<0.01$/15 & $<0.01$/17 & $<0.01$/17 & $<0.01$/759 & $<0.01$/2578 & $<0.01$/925 \\
      \hline
      Energy6 & $<0.01$/8 & $<0.01$/8 & $<0.01$/8 & $<0.01$/2473 & $<0.01$/3730 & $<0.01$/2508 \\
       \hline
       CNOT5 & 0.55/5 & 0.55/3 & 0.57/2 & 0.50/2 & 0.52/2 & 0.72/7\\
       \hline
       CNOT10 & 19.29/36786 & 13.50/5916 & 6.49/4136 & 61.83/25685 & 61.83/27432 & 61.83/54921 \\
       \hline
       CNOT15 & 7.92/555 & 8.89/10 & 25.79/21326 & 64.36/14349 & 64.62/9878 & 64.70/29421 \\
       \hline
       CNOT20 & 28.46/31236 & 22.86/8482 & 38.42/1837 & 67.41/2069 & 67.22/1242 & 67.30/26378 \\
       \hline
        NOT2 & 0.05/1& 0.05/1& 0.04/1& 0.04/1& 0.06/1& 0.07/1\\
        \hline
        NOT6 &0.80/7403 & 6.62/32640 & 0.26/1& 41.37/47882 & 0.26/1& 60.19/90179\\
        \hline
        NOT10 & 0.97/1255 & 60.51/1815 & 4.26/4136& 60.50/29462& 0.66/1& 0.59/1\\
        \hline
       CircuitH$_2$ & 0.89/1 & 0.87/1 & 3.83/1514 & 1.34/1 & 1.87/1 & 60.76/24639\\
       \hline
       CircuitLiH & 70.94/349 & 70.87/889 & 71.01/1570 & 71.45/1545 & 70.83/1595 & 71.23/2137 \\
       \hline
\end{tabular}
\end{sidewaystable}

\begin{sidewaystable}[htbp]
\centering
\captionof{table}{Objective value results of improvement heuristic. \label{tab: res-obj-improve}}
\vspace{-0.1in}
\begin{adjustbox}{width=\linewidth}
\begin{tabular}{l|rrr|rrr|rrrrrrrrrrrrrrrr}
       \hline
       \multirow{4}{*}{Instance} & \multicolumn{9}{c}{Objective value/TV regularizer}\\
       & \multicolumn{9}{c}{Improvement of objective value/Improvement of TV regularizer}\\
      \cline{2-10}
       & pGRAPE+ & TR+ & ADMM+ & pGRAPE+ & TR+ & ADMM+ & pGRAPE+ & TR+ & ADMM+ \\
       & SUR+ALB & SUR+ALB & SUR+ALB & MT+ALB & MT+ALB & MT+ALB & MS+ALB & MS+ALB & MS+ALB\\
       \hline
       \multirow{2}{*}{Energy2} 
    & $0.003$/10 & $0.005$/8 & $0.003$/4 & $0.003$/4 & $0.003$/4 & $0.041$/6 & $0.001$/10 & $0.003$/10 & $0.002$/8 \\
        & $-0.20$\%/81.48\% & 0.00\%/85.19\% & $-0.20$\%/91.67\% 
        & 18.55\%/0.00\% & 18.55\%/0.00\% & 13.36\%/0.00\%
        & 2.88\%/0.00\% & 3.85\%/0.00\% & 2.67\%/20.00\%\\
       \hline
      \multirow{2}{*}{Energy4} & 
    $0.160$/17.2 & $0.162$/14 & $0.166$/14.4 & $0.199$/6 & $0.198$/6 & $0.218$/5.6 & $0.160$/9.2 & $0.160$/10 & $0.184$/8.4\\
    & $-0.14$\%/35.82\% & $-0.42$\%/56.79\% & $0.53$\%/67.86\% 
       & 26.61\%/0.00\% & 17.34\%/0.00\% & 22.84\%/6.67\%
       & 0.38\%/8.00\% & 0.32\%/16.00\% & 1.31\%/0.00\%\\
       \hline
       \multirow{2}{*}{Energy6} &
    $0.218$/32.4 & $0.216$/40.4 & $0.219$/50.0 & $0.429$/6 & $0.429$/6 & $0.451$/6 & $0.229$/10 & $0.230$/10 & $0.216$/10\\
        & $-0.25$\%/16.49\% & $-0.01$\%/7.34\% & 0.25\%/5.30\%
        & 80.74\%/0.00\% & 49.09\%/0.00\% & 39.28\%/0.00\%
        & 2.47\%/0.00\% & 3.77\%/0.00\% & 0.57\%/0.00\%\\
       \hline
       \multirow{2}{*}{CNOT5} & 0.176/9 & 0.206/7 & 0.176/9 & 0.195/9 & 0.196/10 & 0.195/9 & 0.170/16 & 0.173/20 & 0.172/16 \\
       & $-3.53$\%/43.75\% & 37.95\%/70.83\% & 7.37\%/78.05\%
       & 19.75\%/10.00\% & 62.67\%/$-66.67$\% & 31.58\%/$-28.57$\%
       & 0.00\%/0.00\% & 70.83\%/$-33.33$\% & 9.95\%/$50.00$\% \\
       \hline
       \multirow{2}{*}{CNOT10} & 1.58E$-03$/30 & 2.46E$-03$/24 & 1.15E$-03$/20 & 4.06E$-03$/23 & 9.43E$-03$/16 & 6.04E$-03$/15  & 9.80E$-04$/39 & 1.31E$-03$/38 & 1.18E$-03$/36 \\
       & $-66.42$\%/74.14\% & $-38.20$\%/70.69\% & 0.00\%/76.74\%
       & 97.47\%/$-4.55$\% & 97.21\%/23.81\% & 92.86\%/0.00\%
       & 90.91\%/0.00\% & 94.74\%/0.00\% & 83.33\%/$-12.50$\% \\
       \hline
       \multirow{2}{*}{CNOT15} & 5.59E$-04$/262 & 4.67E$-04$/256 & 8.51E$-04$/263 & 6.31E$-03$/33 & 6.25E$-03$/34 & 1.63E$-03$/30 & 1.30E$-03$/39 & 3.72E$-03$/39 & 1.91E$-03$/40 \\
       & 44.10\%/1.50\% & 76.65\%/7.25\% & 71.64\%/5.73\%
       & 98.89\%/10.81\% & 97.93\%/5.56\% & 99.43\%/$-11.11$\%
       & 99.69\%/$-2.63$\% & 98.94\%/2.50\% & 99.07\%/$-2.56$\%\\
       \hline
       \multirow{2}{*}{CNOT20} & 4.56E$-04$/479 & 6.13E$-04$/471 & 5.07E$-04$/441 & 1.20E$-03$/38 & 8.22E$-04$/49 & 1.35E$-03$/48 & 9.47E$-04$/40 & 2.81E$-03$/40 & 7.45E$-04$/40\\
       & 54.44\%/2.44\% & 38.68\%/1.88\% & 49.27\%/5.57\%
       & 99.87\%/28.30\% & 99.74\%/3.92\% & 99.81\%/$-2.13$\%
       & 99.86\%/$-2.56$\% & 99.57\%/$-2.56$\% & 99.88\%/$-2.56$\% \\
       \hline
       \multirow{2}{*}{NOT2} & 0.164/1&	0.164/1&	0.164/1&	0.164/1&	0.164/1&	0.164/1&	0.163/3&	0.163/3&	0.163/3 \\
       & 0.46\%/66.67\%	& 0.36\%/0.00\%	&0.36\%/0.00\%&	0.36\%/0.00\%&	0.36\%/0.00\%&	0.36\%/0.00\%&	0.62\%/0.00\%&	0.52\%/$-$200.00\%&	0.52\%/$-$200.00\% \\
       \hline
       \multirow{2}{*}{NOT6} & 1.42E$-$03/6&	7.33E$-$04/10&	7.24E$-$04/13&	1.45E$-$03/9&	1.27E$-$02/11&	3.39E$-$03/12&	1.06E$-$03/22&	8.90E$-$04/18&	1.74E$-$03/16\\
       & 40.24\%/85.00\%&	76.34\%/86.96\%&	53.25\%	/75.00\%& 96.41\%/$-$12.50\%&	65.21\%/$-$22.22\%&	95.52\%/$-$20.00\%&	92.33\%/$-$15.79\%&	99.44\%/$-$28.57\%&	98.87\%	/$-$14.29\%\\
       \hline
       \multirow{2}{*}{NOT10} & 5.50E$-$04/32&	1.66E$-$03/6&	6.16E$-$04/7&	7.37E$-$04/24&	9.09E$-$04/20&	8.90E$-$04/8&	5.44E$-$04/38&	7.57E$-$04/39&	8.90E$-$04/21	\\
       & 88.37\%/74.60\% &	$-$76.45\%/92.06\% &	93.32\%/66.67\% &	97.35\%/0.00\% &	99.32\%/20.00\%&	98.60\%/$-$33.33\%&	99.12\%/5.00\% &	99.73\%/2.50\% &	90.35\%/0.00\% \\
       \hline
        \multirow{2}{*}{CircuitH$_2$} 
        & 0.003/24 & 0.013/32 & 0.006/76 
        & 0.245/12 & 0.007/2 & 0.054/10 
        & 0.014/18 & 0.003/22 & 0.008/22 \\ %
        & 88.89\%/25.00\% & 51.85\%/11.11\% & 0.00\%/0.00\%
        & 59.17\%/$-50.00$\%& 98.82\%/75.00\% & 14.29\%/$-25.00$\% 
        & 46.15\%/18.18\% & 92.11\%/8.33\% & 0.00\%/0.00\%\\
       \hline
       \multirow{2}{*}{CircuitLiH} 
       & 0.069/378 & 0.083/378 & 0.033/252 & 0.002/6 & 0.496/70 & 0.355/48 & 0.165/286 & 0.120/290 & 0.021/158 \\ %
        & 58.93\%/0.53\% & 54.40\%/0.53\% & 0.00\%/0.00\% 
        & 99.79\%/91.18\% & 48.65\%/2.78\% & 46.05\%/0.00\% 
        & 42.51\%/1.38\% & 46.43\%/$-0.69$\% & 96.23\%/$-6.76$\%\\
       \hline
\end{tabular}
\end{adjustbox}
\end{sidewaystable}

\begin{sidewaystable}[htbp]
\vspace{0.05in}
\centering
\captionof{table}{CPU time and iterations of improvement heuristic. \label{tab: res-time-improvement}}
\vspace{-0.1in}
\begin{adjustbox}{width=\textwidth}
\begin{tabular}{l|rrr|rrr|rrrrrrrrrrr}
       \hline
       \multirow{3}{*}{Instance} & \multicolumn{9}{c}{CPU time (s)/Iterations}\\
       \cline{2-10}
       & pGRAPE+ & TR+ & ADMM+ & pGRAPE+ & TR+ & ADMM+ & pGRAPE+ & TR+ & ADMM+ \\
       & SUR+ALB & SUR+ALB & SUR+ALB & MT+ALB & MT+ALB & MT+ALB & MS+ALB & MS+ALB & MS+ALB\\
      \hline
      Energy2 & 3.41/255 & 5.21/379 & 2.26/262 & 2.26/170 & 2.25/170 & 1.59/150 & 1.27/93 & 1.90/154 & 1.24/91 \\
      \hline
      Energy4 & 3.36/127 & 5.53/191 & 7.90/280 & 5.12/207 & 5.13/207 & 4.94/194 & 1.21/43 & 1.23/43 & 2.14/68\\
      \hline
      Energy6 & 16.95/94 & 11.11/62 & 9.88/56 & 41.42/219 & 39.53/208 & 33.36/174 & 18.51/100 & 24.02/144 & 10.84/56\\
       \hline
       CNOT5 & 12.69/394 & 16.96/506 & 30.09/993 & 19.67/521 & 20.34/550 & 21.46/575 & 0.91/1377 & 33.99/33 & 12.58/464\\
       \hline
       CNOT10 & 70.36/1245 & 56.97/1031 & 56.90/998 & 74.83/945 & 84.95/1050 & 51.37/574 & 19.89/301 & 24.04/323 & 16.65/270\\
       \hline
       CNOT15 & 9.12/106 & 30.48/395 & 25.57/317 & 150.25/1238 & 169.47/1335 & 107.17/895 & 113.02/944 & 129.21/1288 & 82.44/935 \\
       \hline
       CNOT20 & 27.82/249 & 26.22/251 & 33.77/397 & 208.79/1246 & 139.69/1010 & 157.57/1009 & 153.63/1035 & 155.07/1138 & 157.87/1182\\
       \hline
       {NOT2} &0.46/93& 0.40/63& 0.52/93& 0.39/63& 0.46/62& 0.31/63&0.34/63&0.47/63& 0.31/62&\\
       \hline
      {NOT6} & 11.71/725&7.36/306& 2.49/188&9.46/592&5.47/231&8.63/582& 7.82/458&5.82/236& 6.90/470\\
       \hline
       {NOT10} & 29.69/1181&16.16/409& 7.87/298& 21.46/824&  21.89/558&  14.11/541&  7.22/256& 17.05/465& 4.06/161\\
       \hline
       CircuitH$_2$ & 4.84/97 & 3.18/65 & 4.70/33 & 5.40/88 & 9.21/156 & 2.73/39 & 4.88/108 & 3.14/65 & 1.68/33\\ 
       \hline
       CircuitLiH & 26.08/108 & 46.26/108 & 9.05/36 & 64.29/186 & 74.88/153 & 59.21/135 & 41.43/143 & 20.02/71 & 39.78/158\\
       \hline
\end{tabular}
\end{adjustbox}

\end{sidewaystable}

\begin{figure}[H]
    \centering
    \includegraphics[width=\textwidth]{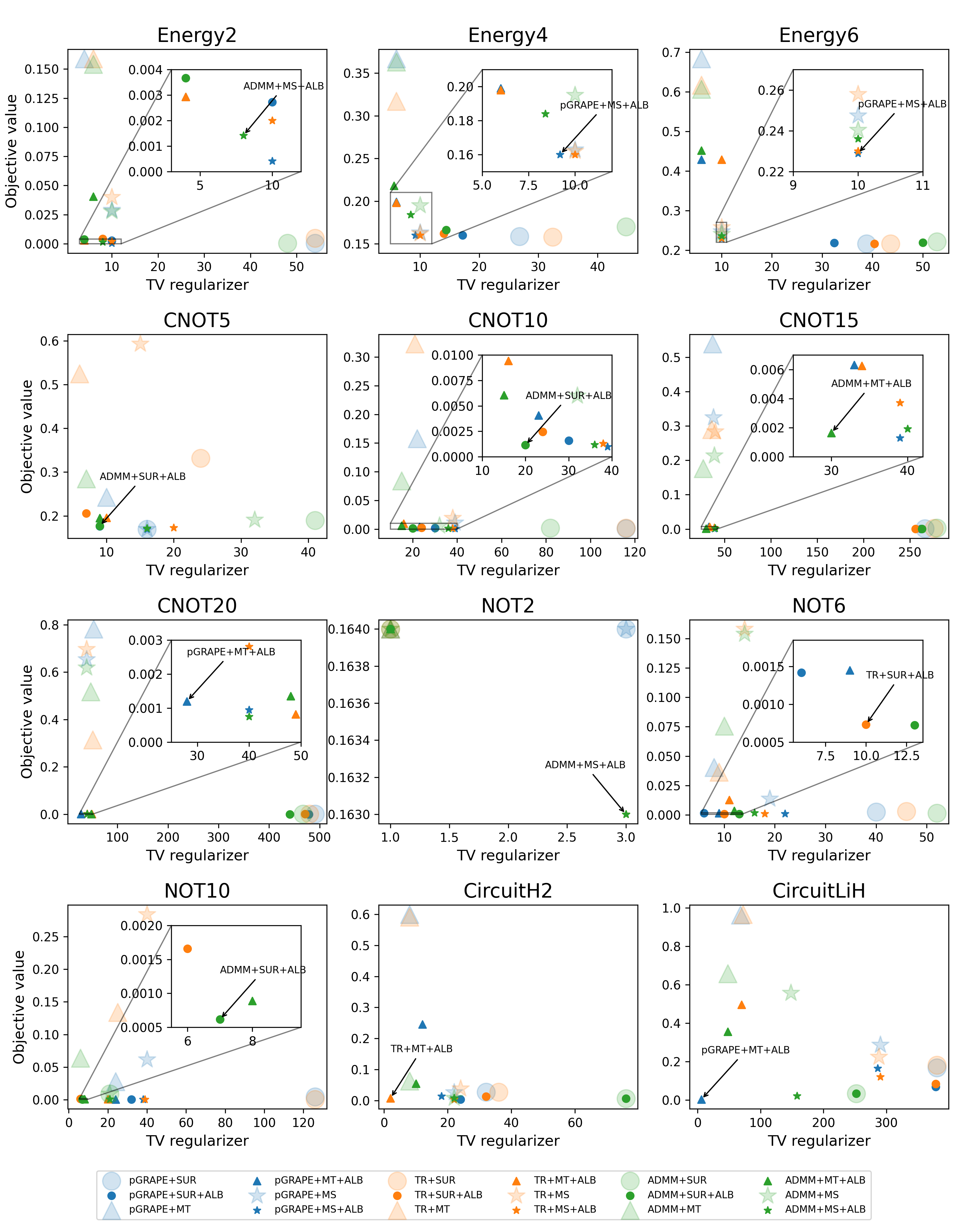}
    \caption{Objective values and TV regularizer values of binary results of all the instances.}
    \label{fig:obj-all-instance}
\end{figure}

\section{Control Results of NOT Estimation Problem}
\label{app: ctrl}
For the NOT gate estimation problem, we present the control obtaining the best trade-off between objective values and TV regularizer values which are annotated in Figure~\ref{fig:not-ctrl}. All three instances show that controller 1 has a more significant impact on the final infidelity. When the evolution time is short ($t_f=2$), the optimal control sets controller 1 active all the time. When the evolution time is longer ($t_f=6,\ 10$), the active time of controller 2 is shorter. The main reason is that only considering controller 1 can also obtain low infidelity with a sufficiently long evolution time but the performance of a single controller is worse than two controllers with the same evolution time~\citep{rebentrost2009optimal,Motzoi2009}.
\begin{figure}[htbp]
    \centering
    \subfloat[$t_f=2$ \label{fig:not2-ctrl}]{\includegraphics[width=0.32\textwidth]{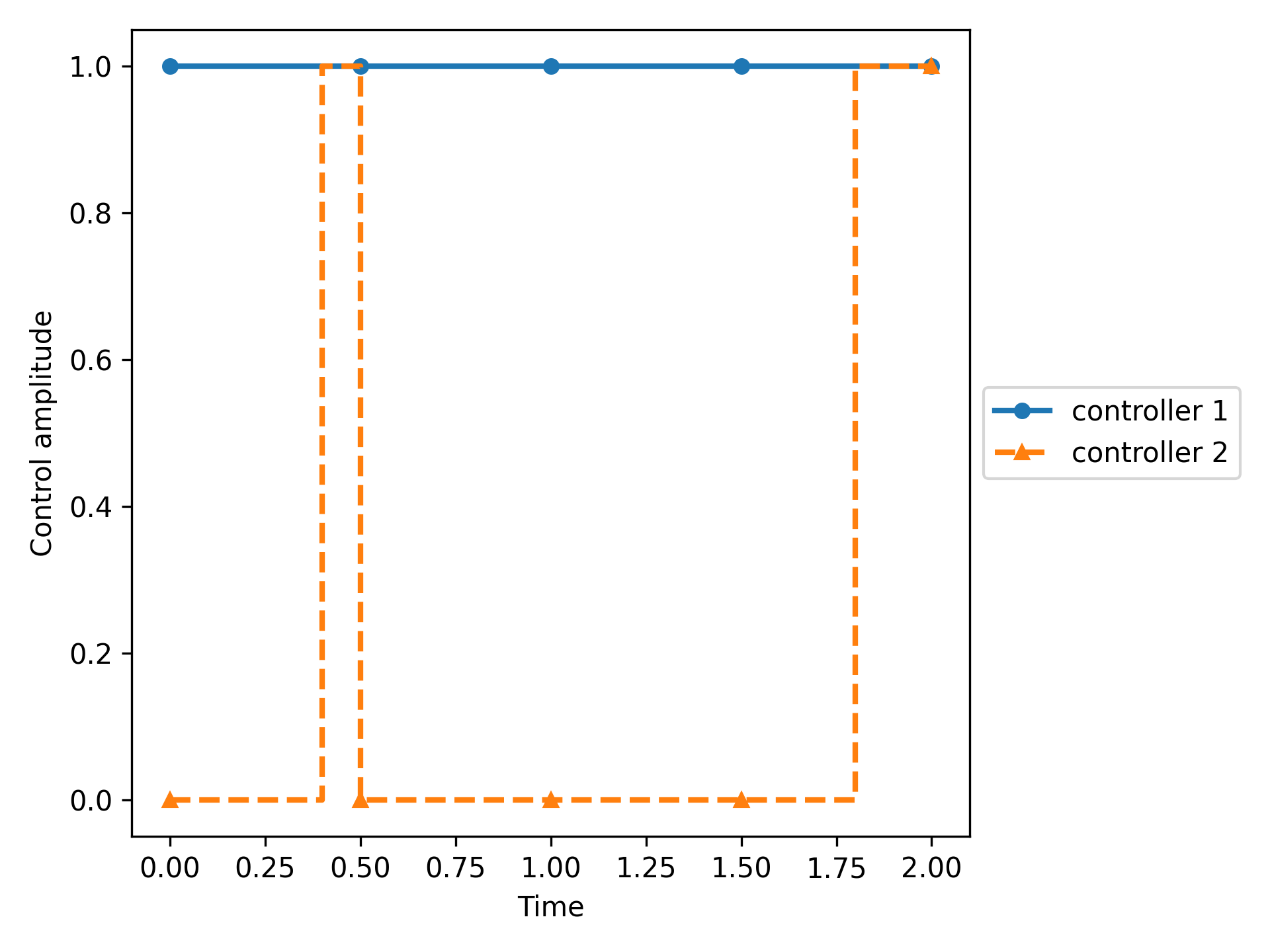}}
    \subfloat[$t_f=6$ \label{fig:not6-ctrl}]{\includegraphics[width=0.32\textwidth]{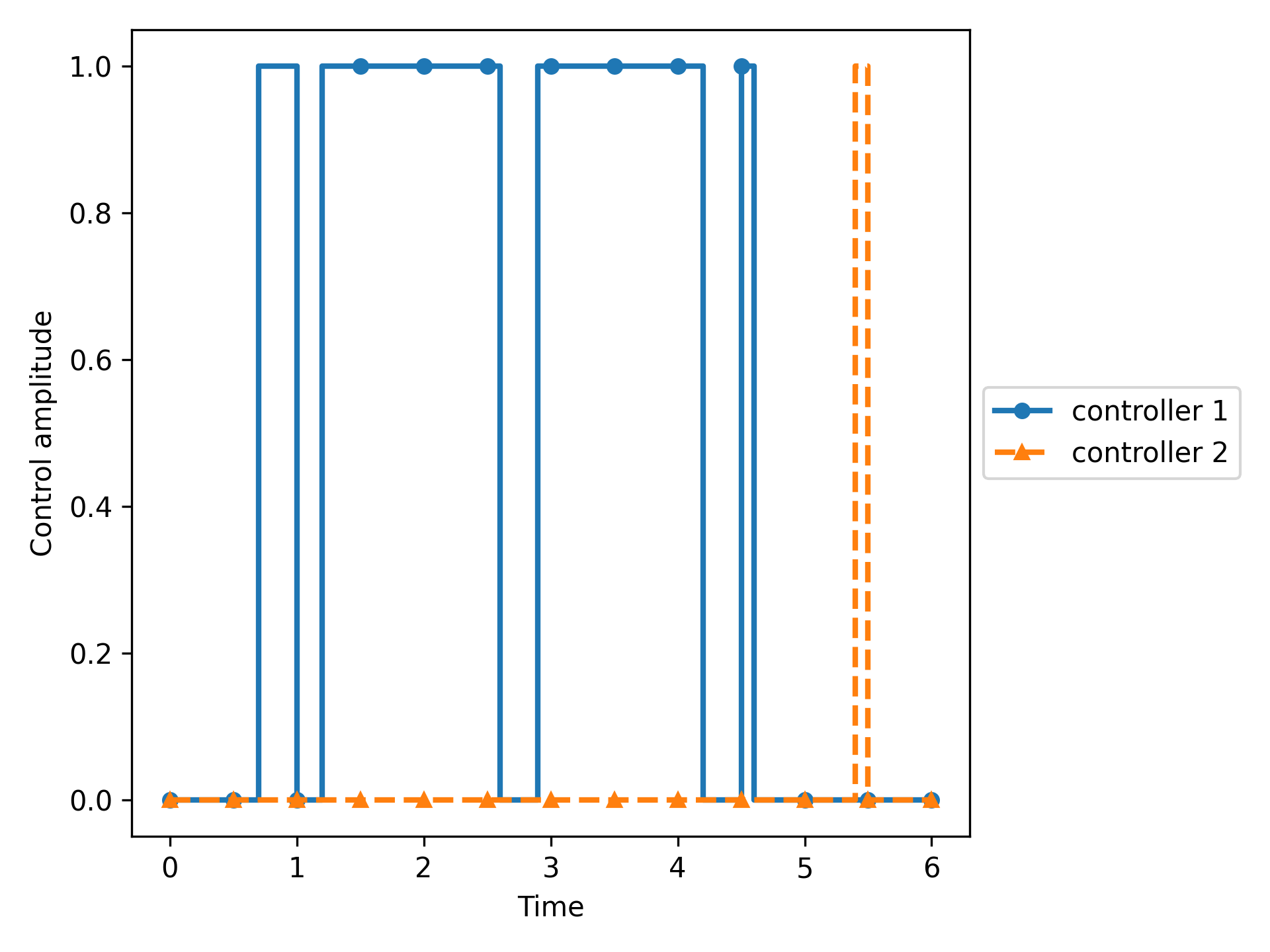}}
    \subfloat[$t_f=10$ \label{fig:not10-ctrl}]{\includegraphics[width=0.32\textwidth]{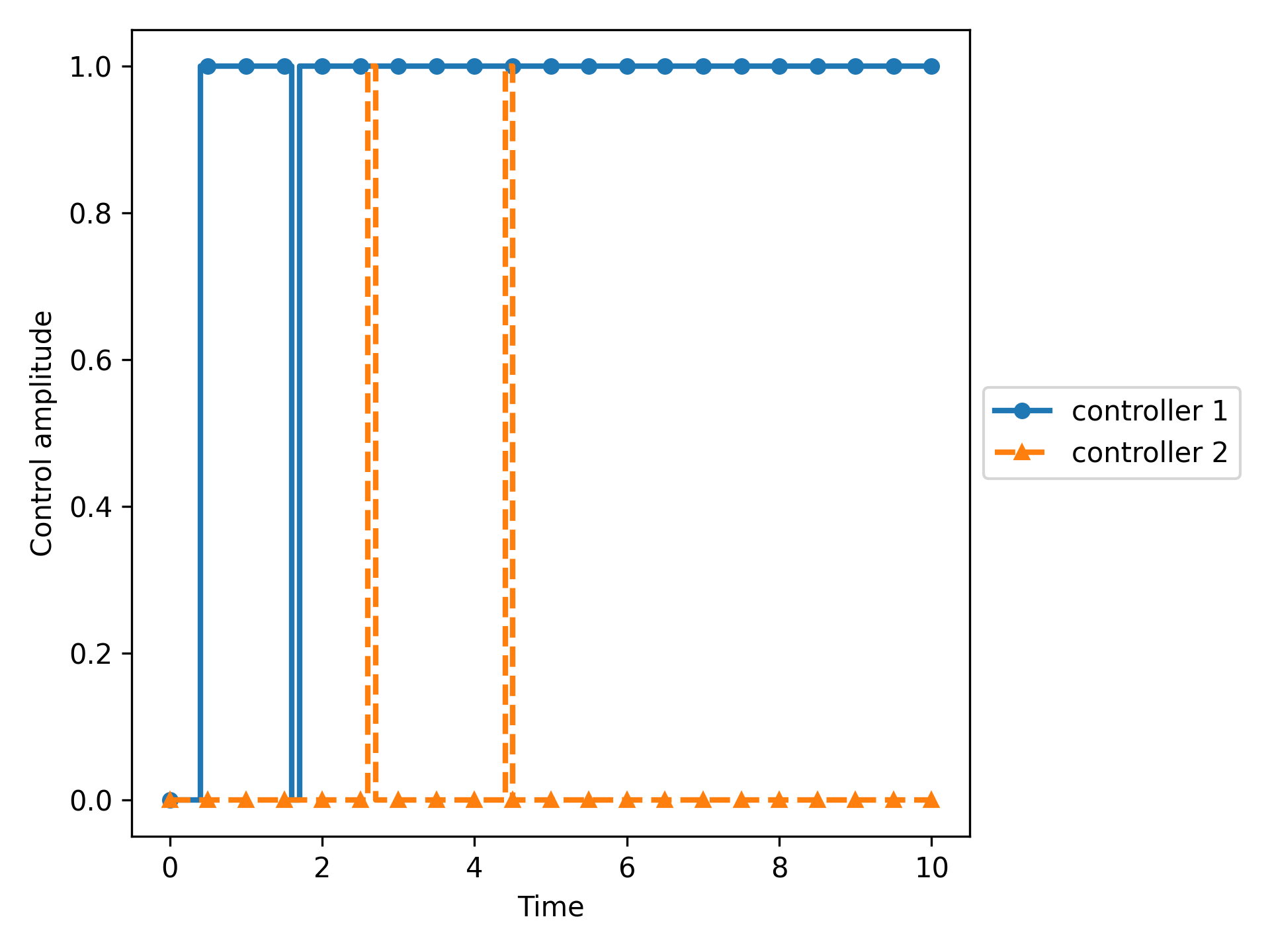}}
    \caption{Control results of NOT gate estimation example with different evolution times. Blue lines represent controller 1 and orange dashed lines represent controller 2.}
    \label{fig:not-ctrl}
\end{figure}
\end{appendices}

\vfill
\framebox{\parbox{.90\linewidth}{\scriptsize The submitted manuscript has been created by
        UChicago Argonne, LLC, Operator of Argonne National Laboratory (``Argonne'').
        Argonne, a U.S.\ Department of Energy Office of Science laboratory, is operated
        under Contract No.\ DE-AC02-06CH11357.  The U.S.\ Government retains for itself,
        and others acting on its behalf, a paid-up nonexclusive, irrevocable worldwide
        license in said article to reproduce, prepare derivative works, distribute
        copies to the public, and perform publicly and display publicly, by or on
        behalf of the Government.  The Department of Energy will provide public access
        to these results of federally sponsored research in accordance with the DOE
        Public Access Plan \url{http://energy.gov/downloads/doe-public-access-plan}.}}

\end{document}